\def\VersionLong{}
\def\VersionFinal{}
\ifdefined\VersionLong%
	\newcommand{\LongVersion}[1]{#1}
	\newcommand{\ShortVersion}[1]{}
\else
	\newcommand{\LongVersion}[1]{}
	\newcommand{\ShortVersion}[1]{#1}
\fi

\documentclass[runningheads,a4paper,10pt,envcountsame]{llncs}

\usepackage{comment}
\ifdefined\VersionLong%
        \includecomment{LongVersionBlock}
        \excludecomment{ShortVersionBlock}
\else
        \excludecomment{LongVersionBlock}
        \includecomment{ShortVersionBlock}
\fi
\includecomment{lncs}

\usepackage[ruled,lined,linesnumbered,noend]{algorithm2e}
\DontPrintSemicolon{}
\SetKwInOut{Input}{Input}
\SetKwInOut{Output}{Output}
\SetKw{KwRemove}{remove}
\SetKw{KwPush}{add}
\SetKw{KwPop}{pop}
\SetKw{KwFrom}{from}
\SetKw{KwTo}{to}
\SetKw{KwCompute}{compute}
\SetKw{KwReturn}{return}
\SetKw{KwBreak}{break}
\SetKw{KwPick}{pick}
\SetKw{KwLet}{let}
\SetKw{KwBe}{be}
\SetKw{KwSplit}{split}
\SetKw{KwInto}{into}
\SetKw{KwFind}{find}
\SetKw{KwFeed}{feed}
\SetKwProg{Fn}{Function}{:}{}
\SetAlgorithmName{Algorithm}{algorithm}{List of Algorithms}

\usepackage[svgnames,table]{xcolor}

\usepackage{graphicx}
\usepackage{stmaryrd}
\usepackage{amssymb, amsmath}
\usepackage[
 	colorlinks=true,
 	citecolor=darkgreen,
	linkcolor=darkblue,
	urlcolor=darkpurple,
]{hyperref}
\usepackage[capitalise,english,nameinlink]{cleveref} %

\usepackage[pdf,singlefile]{graphviz}
\usepackage{amsfonts}
\usepackage{tikz}
\usepackage{booktabs}
\usetikzlibrary{graphs}
\usetikzlibrary{automata,positioning,arrows.meta,decorations.pathmorphing,decorations.pathreplacing,decorations.shapes}
\usepackage{multirow}
\usepackage{array}

\usepackage{colortbl}
\definecolor{darkblue}{rgb}{0.0,0.0,0.6}
\definecolor{darkgreen}{rgb}{0, 0.5, 0}
\definecolor{darkpurple}{rgb}{0.7, 0, 0.7}
\definecolor{darkblue}{rgb}{0, 0, 0.7}

\crefalias{AlgoLine}{line}
\crefname{line}{\text{line}}{\text{lines}} %
\crefname{item}{\text{item}}{\text{items}} %
\crefname{example}{\text{Example}}{\text{Examples}} %
\crefname{assumption}{\text{Assumption}}{\text{Assumptions}} %
\crefname{algorithm}{\text{Algorithm}}{\text{Algorithms}}

\ifdefined\VersionWithComments%
	\usepackage{draftwatermark}
	\SetWatermarkText{draft}
 	\SetWatermarkScale{15}
	\SetWatermarkColor[gray]{0.9}
\fi

\ifdefined\VersionWithComments%
	\usepackage[colorinlistoftodos,textsize=footnotesize]{todonotes}
\else
	\usepackage[disable]{todonotes}
\fi

\newcommand{\gennote}[4][]{\todo[linecolor=#3,backgroundcolor=#3!25,bordercolor=#3#1]{#4: #2}}
\newcommand{\mw}[1]{\gennote{#1}{orange}{MW}}
\newcommand{\ks}[1]{{\gennote{#1}{purple}{KS}}}

\newcommand{\reviewer}[2]{{\gennote{``#2''}{purple}{Reviewer #1}}}

\ifdefined\VersionFinal%
\else
	\usepackage[pagewise]{lineno} %
	\linenumbers%
	
\fi

\tikzstyle{defproblem} = [
 draw=black,
 fill=cyan!10,
 text=black,
 line width=0.5pt,
  text width = \linewidth - 1.6 ex - 1pt,
  inner sep = 0.8 ex,
  rounded corners=4pt]
\newcommand{\recallResult}[2]
{%
	\smallskip

	\noindent\fcolorbox{black}{green!15}{
		\begin{minipage}{.95\columnwidth}
			\noindent\textbf{\cref{#1} (recalled).}
			{\em{}#2}
		\end{minipage}
	}

	\smallskip
}

\tikzstyle{rqanswer} = [
 draw=black,
 fill=gray!30,
 text=black,
 line width=0.5pt,
  text width = \linewidth - 1.6 ex - 1pt,
  inner sep = 0.8 ex,
  rounded corners=4pt]

\usepackage{subcaption}

\usepackage{paralist} %
\newenvironment{ienumeration}
	{\begin{inparaenum}[\itshape i\upshape)]}
	{\end{inparaenum}}
 \newenvironment{myitemize}
	{\ifdefined\VersionLong\begin{itemize}\else\begin{inparaitem}[]\fi}
	{\ifdefined\VersionLong\end{itemize}\else\end{inparaitem}\fi}

\usepackage{booktabs}

\newcommand{\N}{\mathbb{N}}

\newcommand{\powerset}[1]{\mathcal{P}({#1})}
\newcommand{\finpowerset}[1]{\mathcal{P}_{\mathrm{fin}}({#1})}
\newcommand{\setdiff}{\triangle}

\newcommand{\flatten}[1]{{\mathrm{elem}({#1})}}
\newcommand{\KTrue}{\ensuremath{\mathit{true}}}

\newcommand{\word}[1][]{w#1}
\newcommand{\INPUT}{\Sigma}
\newcommand{\OUTPUT}{\Gamma}
\newcommand{\emptyword}{\epsilon}
\newcommand{\Autom}{\mathcal{A}}
\newcommand{\M}{\mathcal{M}}
\newcommand{\init}{\mathit{init}}
\newcommand{\loc}{q}
\newcommand{\Loc}{Q}
\newcommand{\initLoc}{\loc_{\init}}

\newcommand{\guard}{\varphi}
\newcommand{\transition}{\delta}
\newcommand{\Lg}{\mathcal{L}}
\newcommand{\Algebra}{\mathcal{B}}
\newcommand{\domain}{\mathfrak{D}_{\Algebra}}
\newcommand{\Predicates}{{\Psi}_{\Algebra}}
\newcommand{\sem}[1]{\llbracket{ #1 }\rrbracket}
\newcommand{\Table}{T}
\newcommand{\prefix}{s}
\newcommand{\Prefixes}{S}
\newcommand{\NextPrefixes}{R}

\newcommand{\Suffixes}{E}
\newcommand{\partitioning}{P}
\newcommand{\Partitions}{\Pi_{\Algebra}}
\newcommand{\partition}{\pi}
\newcommand{\prefixes}{\mathrm{prefixes}}
\newcommand{\cex}{\mathit{cex}}
\newcommand{\targetLg}{\Lg_{\mathrm{tgt}}}
\newcommand{\hypothesis}{\Autom_{\mathrm{hyp}}}
\newcommand{\targetM}{\M_{\mathrm{tgt}}}
\newcommand{\hypothesisM}{\M_{\mathrm{hyp}}}
\newcommand{\hypothesisMe}{\M_{\mathrm{hyp}}^e}
\newcommand{\hypothesisMeNext}{{\hypothesisMe}'}
\newcommand{\hypothesisMNext}{{\M_{\mathrm{hyp}}}'}
\newcommand{\SigmaEf}{\Sigma_E^{\mathit{final}}}

\newcommand{\row}{\mathit{row}}

\newcommand{\trans}{\delta_t}
\newcommand{\outp}{\delta_o}
\newcommand{\transe}{\delta^e_t}
\newcommand{\outpe}{\delta^e_o}
\newcommand{\LambdaM}{\Lambda^*_{M}}
\newcommand{\eMealy}{{\M^e}}

\newcommand{\targetP}{\partition_{\mathrm{tgt}}}

\newcommand{\na}{\mathit{na}}

\newcommand{\MH}{\textsf{MH}}
\newcommand{\ATGS}{\textsf{ATGS}}

\newcommand{\comparisonBad}{\cellcolor{red!25}\bf}

\newcommand{\comparisonGood}{\cellcolor{green!25}\bf}

\ifdefined\VersionWithComments%
 	\definecolor{colorok}{RGB}{80,80,150}
\else
	\definecolor{colorok}{RGB}{0,0,0}
\fi

\newcommand{\eg}{\textcolor{colorok}{e.\,g.,}\xspace}
\newcommand{\ie}{\textcolor{colorok}{i.\,e.,}\xspace}
\newcommand{\st}{\textcolor{colorok}{s.t.}\xspace}

\newcommand{\wrt}{\textcolor{colorok}{w.r.t.}\xspace}

\makeatletter
\def\orcidID#1{\smash{\href{https://orcid.org/#1}{\protect\raisebox{-1.25pt}{\protect\includegraphics{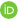}}}}}
\makeatother

\title{Active Learning of Symbolic Mealy Automata}
\titlerunning{Active Learning of Symbolic Mealy Automata}

\author{
\ifdefined\VersionAnonymous%
\else
Kengo Irie\inst{1}\and
Masaki Waga\orcidID{0000-0001-9360-7490/}\inst{1,2}\and
Kohei Suenaga\orcidID{0000-0002-7466-8789}\inst{1}\LongVersion{\thanks{%
    This is the author (and extended) version of the manuscript of the same name, published in the proceedings of the 22nd International Colloquium
on Theoretical Aspects of Computing (ICTAC 2025).
    The final version is available at \url{www.springer.com}.
    }%
}
\fi
}
\institute{%
\ifdefined\VersionAnonymous%
\else
Graduate School of Informatics, Kyoto University, Kyoto, Japan
\and
National Institute of Informatics, Tokyo, Japan
\fi
}

\begin{document}

\maketitle              %
\ifdefined\VersionFinal
\else
\pagestyle{plain}
\fi
\begin{abstract}
We propose $\LambdaM$---an active learning algorithm that learns symbolic Mealy automata, which support infinite input alphabets and multiple output characters.
Each of these two features has been addressed separately in prior work.
Combining these two features poses a challenge in learning the outputs corresponding to potentially infinite sets of input characters at each state.
To address this challenge, we introduce the notion of \emph{essential input characters}, a finite set of input characters that is sufficient to learn the output function of a symbolic Mealy automaton. 
$\LambdaM$ maintains an underapproximation of the essential input characters and refines this set during learning.
We prove that $\LambdaM$ terminates under certain assumptions.
Moreover, we provide upper and lower bounds for the query complexity.
Their similarity suggests the tightness of the bounds.
We empirically demonstrate that $\LambdaM$ is 
\begin{ienumeration}
    \item efficient regarding the number of queries on practical benchmarks and 
    \item scalable according to evaluations with randomly generated benchmarks.
\end{ienumeration}

 {}
\end{abstract}

\section{Introduction}

\emph{Active automata learning}~\cite{DBLP:journals/iandc/Angluin87} is the problem of exactly identifying an unknown automaton via a finite number of queries.
Since the seminal work of Angluin~\cite{DBLP:journals/iandc/Angluin87}, \LongVersion{active automata learning}\ShortVersion{it} has received much attention from both the machine learning theory and system verification communities.
In the context of verification, for example, it is used to identify an automaton representing the behavior of an unknown system for testing~\cite{DBLP:conf/forte/PeledVY99,DBLP:conf/dagstuhl/Meinke16,DBLP:conf/hybrid/Waga20} or controller synthesis~\cite{DBLP:journals/tase/ZhangFL20}.

The L* algorithm~\cite{DBLP:journals/iandc/Angluin87}, the best-known active DFA learning algorithm, infers the minimum DFA recognizing the target regular language $\targetLg$ using \emph{membership} and \emph{equivalence} queries: 
In a membership query, the learner asks if a word belongs to $\targetLg$; in an equivalence query, the learner asks if the DFA built by the learner---called a \emph{hypothesis DFA}---\LongVersion{correctly }recognizes $\targetLg$.
L* is proved to infer the minimum DFA that recognizes $\targetLg$ within a polynomial number of queries.

Although active automata learning is known to be theoretically interesting and practically useful,
classical algorithms are\LongVersion{ too} restrictive for real-world systems due to the following gaps:
\begin{ienumeration}
\item real-world systems usually take inputs of infinite values (\eg{} numbers), and modeling them as state machines typically requires manual or automatic identification of alphabet abstraction (\eg{}\LongVersion{ identification of} guard predicates), and
\item real-world systems usually produce multi-valued outputs, and we need to learn a state machine representing a \emph{function} rather than a \emph{language}.
\end{ienumeration}

On the one hand, Drews and D'Antoni proposed\LongVersion{ the} $\Lambda^*$\LongVersion{ algorithm}~\cite{DBLP:conf/tacas/DrewsD17} to address the first challenge by extending L*.
$\Lambda^*$ learns a \emph{symbolic finite automaton (s-FA)},\LongVersion{ which is} an automaton with predicate-labeled transitions to handle \emph{large} or even \emph{infinite} input alphabets.
The learner in $\Lambda^*$ constructs a hypothesis automaton by
\begin{ienumeration}
    \item learning an automaton (called an \emph{evidence automaton}) over a finite subset of the alphabet using a variant of L* and
    \item learning predicates from concrete characters to generalize the evidence automaton into a symbolic finite automaton.
\end{ienumeration}
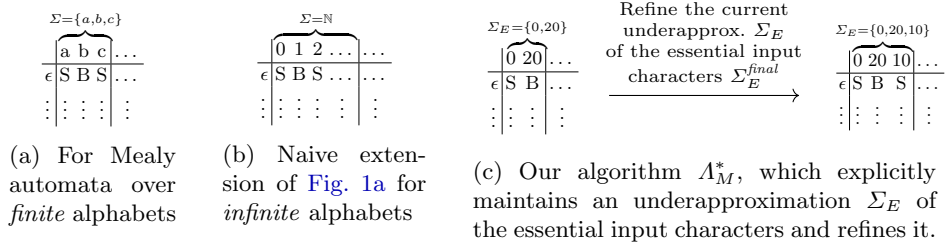
\begin{figure}[t]
    \begin{minipage}{0.18\textwidth}
        \centering
        
        \scalebox{.8}{
         \begin{tabular}{l|ccc|c}
              \multicolumn{5}{l}{\hspace{0.2em}$\overbrace{\hspace{2.7em}}^{\INPUT = \{a, b, c\}}$\vspace{-0.5em}} \\
              & a & b & c  & \dots \\\hline
              $\epsilon$ & S & B & S & \dots\\
              \vdots & \vdots & \vdots & \vdots & \vdots 
        \end{tabular}
        }
        \subcaption{For Mealy automata over \emph{finite} alphabets}
        \label{figure:Mealy_style_learning:finiteOutputs}
    \end{minipage}
    \hfill
    \begin{minipage}{0.22\textwidth}
        \centering
        \scalebox{.8}{
        \begin{tabular}{l|cccc|c}
            \multicolumn{6}{l}{\hspace{0.7em}$\overbrace{\hspace{4.2em}}^{\INPUT = \mathbb{N}}$\vspace{-0.5em}} \\
            & 0 & 1 &2 & \dots & \dots \\\hline
            $\epsilon$ & S & B & S & \dots & \dots\\
            \vdots & \vdots & \vdots & \vdots & \vdots & \vdots 
        \end{tabular}
        }
        \subcaption{Naive extension of \cref{figure:Mealy_style_learning:finiteOutputs} for \emph{infinite} alphabets}
        \label{figure:Mealy_style_learning:infiniteOutputs}
    \end{minipage}
    \hfill
    \begin{minipage}{0.5\textwidth}
        \centering
        \begin{tikzpicture}
            \node at (2.8, 0) {
            \scalebox{.8}{
              \begin{tabular}{l|cc|c}
                \multicolumn{4}{l}{\hspace{-0.2em}$\overbrace{\hspace{2.0em}}^{\Sigma_E = \{0, 20\}}$\vspace{-0.5em}} \\
                  & 0 & 20 & \dots \\\hline
                  $\epsilon$ & S & B & \dots \\
                  \vdots & \vdots & \vdots & \vdots \\
                \end{tabular}
                }
              };

            \scriptsize
            \draw[->] (3.8,-.25) -- (6.3, -.25) node[midway, above, align = center] {Refine the current\\ underapprox.\ $\Sigma_E$\\ of the essential input\\ characters $\SigmaEf$};
            \normalsize

            \node at (7.5, 0) {
            \scalebox{.7}{
                \begin{tabular}{l|ccc|c}
                \multicolumn{5}{l}{\hspace{-0.2em}$\overbrace{\hspace{3.2em}}^{\Sigma_E = \{0, 20, 10\}}$\vspace{-0.5em}} \\
                  & 0 & 20 & 10&\dots \\\hline
                  $\epsilon$ & S & B &S & \dots\\
                  \vdots & \vdots & \vdots & \vdots & \vdots \\
                \end{tabular}
                }
              };
        \end{tikzpicture}
        \subcaption{Our algorithm $\LambdaM$, which explicitly maintains an underapproximation $\Sigma_E$ of the essential input characters and refines it.}
        \label{figure:Mealy_style_learning:LambdaStarM}
    \end{minipage}
    
    \caption{Observation tables for learning Mealy-style automata. The cell indexed by $(s, e)$ contains the output for the input $s \cdot e$. The algorithm in~\cite{DBLP:phd/de/Niese2003} identifies the output function\LongVersion{ of Mealy automata} by recording the output for each input character (\cref{figure:Mealy_style_learning:finiteOutputs}), which does not work for infinite alphabets (\cref{figure:Mealy_style_learning:infiniteOutputs}). Instead, we learn the output function focusing on the essential characters $\Sigma_E$\LongVersion{, which we refine through learning} (\cref{figure:Mealy_style_learning:LambdaStarM}).}
    \label{figure:Mealy_style_learning}
\end{figure}

On the other hand, Niese~\cite{DBLP:phd/de/Niese2003} extended L* to learn a \emph{deterministic Mealy automaton} to address the second challenge.
\cref{figure:Mealy_style_learning:finiteOutputs} shows the idea of learning Mealy-style outputs with a data structure called an \emph{observation table}. 
Each cell in the observation table is indexed by a row--column pair $(s, e) \in \Sigma^* \times \Sigma^*$ of words, and the cell indexed by $(s, e)$ contains the output of the target Mealy automaton after processing $s \cdot e$.
The row indexed by $s$ corresponds to the state reached by consuming $s$, and
by recording the output after processing $s \cdot a$ for each $a \in \Sigma$, Niese's algorithm identifies a Mealy-style output function.
\paragraph{Contribution: Active learning of symbolic Mealy automata}
We propose an active learning algorithm $\LambdaM$---an extension of $\Lambda^*$---that learns \emph{symbolic Mealy automata (s-MAs)}.
An s-MA is a special case of a \emph{symbolic finite transducer}~\cite{DBLP:conf/popl/VeanesHLMB12} with predicate-labeled transitions, like s-FAs, and outputs associated with transitions, like Mealy automata.
To the best of our knowledge, this is the first active learning algorithm for automata that supports general Boolean algebras for input predicates and Mealy-style multiple outputs with the minimality guarantee.

$\LambdaM$ is based on the automata learning algorithms mentioned above.
However, it is more than a naive combination of them.
In particular, the major challenge is that Niese's algorithm requires recording output characters for every input character, which is infeasible when dealing with an infinite input alphabet (\cref{figure:Mealy_style_learning:infiniteOutputs}).

Our central observation is that, for learnable symbolic automata, tracking only a finite subset of the input alphabet, which we call \emph{essential input characters} $\SigmaEf$, is sufficient for learning.
Furthermore, we found that tracking $\SigmaEf$ alone is enough to learn the output function of an s-MA.
Following this observation, $\LambdaM$, refines an underapproximation $\Sigma_E$ to identify $\SigmaEf$ and learns the output function by focusing on $\Sigma_E$ (\cref{figure:Mealy_style_learning:LambdaStarM}).%
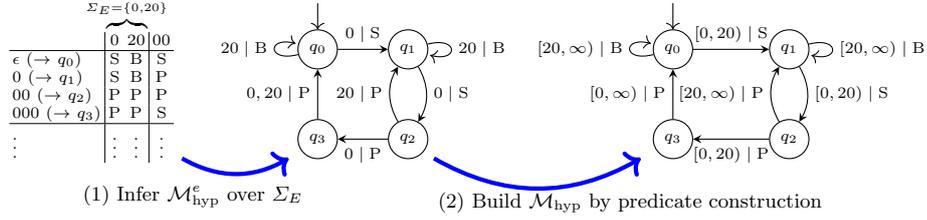
\begin{figure}[tbp]
    \centering
    \begin{tikzpicture} [node distance = 1.4cm, on grid, auto,scale=0.85,every node/.style={transform shape}]
        \node (table) at (-3.5, -0.5) {\scriptsize
            \begin{tabular}{l|cc|c}
               \multicolumn{4}{c}{\hspace{3.9em}$\overbrace{\hspace{2.3em}}^{\Sigma_E = \{0,20\}}$\vspace{-0.3em}} \\
                & 0 & 20 & $0 0$\\\hline
                $\epsilon$ ($\to$ $q_0$) & S & B  & S\\
                0 ($\to$ $q_1$) & S & B & P\\
                $0 0$ ($\to$ $q_2$) & P & P & P \\
                $0 0 0$ ($\to$ $q_3$) & P & P & S \\\hline
                \vdots & \vdots & \vdots  & \vdots 
            \end{tabular}};

        \node (q0) [state,minimum size=0.5cm, initial above, initial text = {}] {\scriptsize$q_0$};
        \node (q1) [state, right = of q0, minimum size=0.5cm] {\scriptsize$q_1$};
        \node (q2) [state, below = of q1,minimum size=0.5cm] {\scriptsize$q_2$};
        \node (q3) [state, below = of q0,minimum size=0.5cm] {\scriptsize$q_3$};
 
        \path [-stealth]
            (q0) edge  [font=\scriptsize]node {$0\mid \text{S}$} (q1) 
            (q0) edge [loop left,font=\scriptsize]  node {$20\mid \text{B}$}()
            (q1) edge [loop right,font=\scriptsize] node {$20\mid \text{B}$}()
            (q1) edge  [bend left,font=\scriptsize ]node {$0\mid \text{S}$} (q2)
            (q2) edge  [bend left,font=\scriptsize ]node {$20\mid \text{P}$} (q1)
            (q2) edge [font=\scriptsize]  node {$0\mid \text{P}$}(q3)
            (q3) edge  [font=\scriptsize]node {$0,20\mid \text{P}$}   (q0)
        ;

        \node (q'0) [state, initial above, initial text = {}, minimum size=0.5cm] at (5.5,0) {\scriptsize $q_0$};
        \node (q'1) [state, node distance=1.8cm, right = of q'0, minimum size=0.5cm] {\scriptsize$q_1$};
        \node (q'2) [state, below = of q'1, minimum size=0.5cm] {\scriptsize$q_2$};
        \node (q'3) [state, below = of q'0, minimum size=0.5cm] {\scriptsize$q_3$};
 
        \path [-stealth]
            (q'0) edge  [font=\scriptsize]node {$[0,20)\mid \text{S}$} (q'1) 
            (q'0) edge [loop left,font=\scriptsize]  node {$[20,\infty)\mid \text{B}$}()
            (q'1) edge [loop right,font=\scriptsize] node {$[20,\infty)\mid \text{B}$}()
            (q'1) edge  [bend left,font=\scriptsize]node {$[0,20)\mid \text{S}$} (q'2)
            (q'2) edge  [bend left,font=\scriptsize ]node {$[20,\infty)\mid \text{P}$} (q'1)
            (q'2) edge [font=\scriptsize]  node {$[0,20)\mid \text{P}$}(q'3)
            (q'3) edge  [font=\scriptsize]node {$[0,\infty)\mid \text{P}$}   (q'0);

        \path [->, ultra thick, color=blue, bend right]
            (-2.1, -1.7) edge node[below=0.2,pos=0.1,align=center, color=black] {(1) Infer $\hypothesisMe$ over $\Sigma_E$} (-0.45, -1.7)
            (1.8, -1.7) edge node[below=0.3,pos=0.95,align=center, color=black] {(2) Build $\hypothesisM$ by predicate construction} (5.0, -1.7)
        ;
    \end{tikzpicture}
    \caption{Outline of $\LambdaM$. A Mealy automaton $\hypothesisMe$ is constructed from the observation table~(1), and the input character on each transition of $\hypothesisMe$ is generalized into a predicate, such as an interval~(2).}%
    \label{figure:outline_algorithm}
\end{figure}
\Cref{figure:outline_algorithm} outlines $\LambdaM$.
First, we use a variant of Niese's algorithm to infer a Mealy automaton $\hypothesisMe$ over $\Sigma_E$.
Then, following the idea of $\Lambda^*$, we generalize $\hypothesisMe$ into an s-MA $\hypothesisM$, which is a hypothesis automaton in $\LambdaM$.
The refinement of $\Sigma_E$ occurs through equivalence queries.

We prove the termination of the $\LambdaM$ algorithm and provide a detailed analysis of its query complexity.
By explicitly formulating the essential characters $\SigmaEf$, we give a detailed analysis of the number of explored inputs and improve the complexity analysis of the query compared to~\cite{DBLP:conf/tacas/DrewsD17}.
In \cite{DBLP:conf/tacas/DrewsD17}, the upper bound of the number of equivalence queries is quadratic in the number of states, whereas we provide a linear bound.
We also give a theoretical \emph{lower} bound of the number of equivalence queries. 
To the best of our knowledge, this is the first lower bound in the context of active learning of symbolic automata.
Our lower bound supports the tightness of our complexity analysis.
We implemented the $\LambdaM$ algorithm and conducted experiments to evaluate its efficiency.
Our experiments suggest that the practical query complexity is significantly lower than the theoretical upper bound because the discovery of states and essential characters usually occurs in parallel.

Overall, our contributions are summarized as follows. %
\begin{itemize}
      \item We propose the $\LambdaM$ algorithm for active learning of symbolic Mealy automata. The central idea is to explicitly maintain a set of essential inputs.
      \item We give upper and lower bounds of the number of queries required in $\LambdaM$.
      \item Our experiments show that, despite the theoretical lower bound for a specific\LongVersion{ corner} case, $\LambdaM$ typically requires much fewer queries than the theoretical bounds.
\end{itemize}

\ShortVersion{\paragraph{Related Work}}\LongVersion{\subsection{Related Work}\label{section:related_work}}
\begin{table}[tbp]
 \caption{Active learning algorithms for automata with an infinite input alphabet.}%
 \label{table:related_algorithms}
 \centering
 \scriptsize
 \begin{tabular}{c c c c c}
  \toprule
  & Outputs & Learned Predicates & \begin{tabular}{c}
       Learnable  \\
       Boolean Algebra 
  \end{tabular} & Minimal? \\
  \midrule
  \textbf{Ours} & \comparisonGood{} \begin{tabular}{c}
Multiple Constants \\
(Mealy-style)
\end{tabular}   & \comparisonGood{} Concrete & \comparisonGood{} General & \comparisonGood{} Yes \\
  \cite{DBLP:conf/tacas/DrewsD17,DBLP:conf/cav/ArgyrosD18}& \comparisonBad{} \begin{tabular}{c}
Only Boolean \\
(Moore-style)
\end{tabular} & \comparisonGood{} Concrete & \comparisonGood{} General & \comparisonGood{} Yes \\
  \cite{DBLP:conf/vmcai/HowarSM11} & \comparisonGood{} \begin{tabular}{c}
Multiple Constants \\
(Mealy-style)
\end{tabular} & \comparisonBad{} Abstract & \comparisonBad{} Fixed & \comparisonGood{} Yes\\
  \cite{DBLP:conf/icissp/LathouwersEH20} & \comparisonGood{} \begin{tabular}{c}
Constants and Identity Function\\
(Mealy-style)
\end{tabular} & \comparisonGood{} Concrete & \comparisonBad{} Fixed & \comparisonBad{} No \\

\cite{DBLP:conf/pts/AartsJU10} & \comparisonGood{} \begin{tabular}{c}
Any Handcrafted Functions\\
(Mealy-style)
\end{tabular} & \comparisonBad{} \begin{tabular}{c}
       N/A\\
      (No Learning)
\end{tabular}& \comparisonBad{}  \begin{tabular}{c}
       N/A\\
      (No Learning)
\end{tabular}& \comparisonGood{} Yes \\
\cite{DBLP:journals/corr/MensM15} & \comparisonBad{} \begin{tabular}{c}
Only Boolean\\
(Moore-style)
\end{tabular} & \comparisonGood{} Concrete & \comparisonBad{} Fixed & \comparisonGood{} Yes \\
  \bottomrule
 \end{tabular}
\end{table}
Several algorithms have been proposed for actively learning finite-state automata with an infinite input alphabet.
\Cref{table:related_algorithms} summarizes some of these algorithms.
In~\cite{DBLP:conf/tacas/DrewsD17}, Drews and D'Antoni presented the $\Lambda^*$ algorithm to learn symbolic finite automata.
$\Lambda^*$ is a combination of L*~\cite{DBLP:journals/iandc/Angluin87} and predicate learning with a \emph{partitioning function}. %
Our algorithm is an extension of $\Lambda^*$~\cite{DBLP:conf/tacas/DrewsD17} for Mealy-style multiple outputs. 
For the complexity analysis of $\Lambda^*$, they introduced the notion of learnability of a Boolean algebra and classified the Boolean algebras with respect to the learning complexity.
Compared with their query complexity analysis, ours is finer thanks to the theoretical gadgets (\eg{} $\Sigma_{E}^{\mathit{final}}$) we introduce.

In~\cite{DBLP:conf/cav/ArgyrosD18}, Argyros and D'Antoni presented another algorithm, called MAT*, to learn symbolic finite automata.
MAT* is a combination of TTT~\cite{DBLP:conf/rv/IsbernerHS14} and \emph{active} predicate learning, %
unlike the \emph{passive} predicate learning in $\Lambda^*$ and our algorithm $\LambdaM$.
Extending\LongVersion{ our algorithm} $\LambdaM$ to active predicate learning is a viable future direction.

In~\cite{DBLP:conf/vmcai/HowarSM11}, Howar et al.\ proposed alphabet abstraction to extend existing automata learning algorithms to handle infinite input alphabets.
However, the learned predicates are defined in terms of the target automaton's evaluation, making them abstract.
Additionally, their algorithm can only handle abstract predicates over a Boolean algebra defined by the target automaton. 

In~\cite{DBLP:conf/icissp/LathouwersEH20}, Lathouwers et al. proposed a learning algorithm for symbolic finite transducers.
Their algorithm supports both constant outputs, like those in $\LambdaM$, and identity functions as outputs. %
However, their algorithm is limited to a specific Boolean algebra known as \emph{equality algebra}.
Therefore, it cannot handle, for example, the benchmarks in \cref{section:experiments:practical}.
Moreover, their learning algorithm uses \emph{Moore}-style criteria to identify the state space, while the learned transducer has \emph{Mealy}-style outputs.
As a result, it essentially learns an automaton with Moore-style outputs, and
the learned transducer may not be minimal. %
In \cite{DBLP:conf/pts/AartsJU10}, Aarts et al. proposed a learning algorithm for symbolic Mealy machines, finite automata with abstraction to handle infinite or large alphabets. 
Their method relies on externally provided abstractions, which presumably must be handcrafted.
In contrast, our method automatically learns abstractions through queries to an oracle. 
\LongVersion{
}
In~\cite{DBLP:journals/corr/MensM15}, a learning algorithm for symbolic finite automata over ordered domains was proposed.
In contrast, our algorithm and its termination proof apply to more general Boolean algebras, \eg{} the equality algebra\LongVersion{ and the propositional algebra}.

\LongVersion{In~\cite{DBLP:conf/popl/MoermanS0KS17}, a learning algorithm for \emph{nominal automata}~\cite{DBLP:journals/corr/BojanczykKL14}, another class of automata for infinite alphabets, was proposed.
A comparison of the expressive power of symbolic and nominal automata is a future work.}
\section{Preliminaries}\label{section:preliminaries}
Let $\N$ be the set of naturals.
For a set $X$, we let $\powerset{X}$ be the powerset of $X$ and $\finpowerset{X}$ be the finite powerset of $X$.
For a set $X$,
we denote its size by $|X|$.
For a set $X$ of \emph{characters}, we let $X^*$ be the set of \emph{words} over $X$, \ie{} $X^* = \bigcup_{n = 0}^{\infty} X^n$, where $X^n$ is the set of length-$n$ sequences over $X$ ($n \in \N$). 
The word of length $0$ (\ie{} the empty word) is denoted by $\emptyword$.
We write $X^+$ for $X^* \setminus \{ \epsilon \}$.
For words $w ,w' \in X^*$, we denote the concatenation of $w$ and $w'$ by $w \cdot w'$.
For a set $X$ and for a finite sequence $L = l_1 l_2 \dots l_n \in ({\powerset{X}})^*$ of subsets of $X$, we write $\flatten{L}$ for $\bigcup_{i = 1}^{n} l_{i}$.
For sets $X, Y$, we write $X \setdiff Y$ for their symmetric difference, \ie{} $X \setdiff Y = (X \setminus Y) \cup (Y \setminus X)$.

\subsection{Symbolic Finite Automata and Symbolic Mealy Automata}

\begin{LongVersionBlock}
 \begin{definition}
 [DFA]
 A \emph{deterministic finite automaton (DFA)} is a 5-tuple $\Autom = (\INPUT, \Loc, \initLoc, F,\delta)$, where $\INPUT$ is a finite set of input characters, $\Loc$ is a non-empty finite set of states, $\initLoc \in \Loc$ is the initial state, $F \subseteq \Loc$ is the set of accepting states, and $\delta \colon \Loc \times \INPUT \to \Loc$ is the transition function.
\end{definition}
\end{LongVersionBlock}

\LongVersion{\begin{definition}
 [Deterministic Mealy Automata]}
 A \emph{deterministic Mealy automaton} is a 6-tuple $\M = (\INPUT,\Loc,\initLoc,\OUTPUT,\trans,\outp)$, where $\INPUT$ is a finite set of \emph{input characters}, $\Loc$ is a non-empty finite set of \emph{states}, $\initLoc \in \Loc$ is the \emph{initial state}, $\OUTPUT$ is a finite set of \emph{output characters}, $\trans \colon \Loc \times \INPUT \to \Loc$ is the \emph{transition function}, and $\outp \colon \Loc \times \INPUT \to \OUTPUT$ is the \emph{output function}. %
\LongVersion{\end{definition}}%
\LongVersion{
}%
\LongVersion{By abuse of notation, we extend the transition function $\trans \colon \Loc \times \INPUT \to \Loc$ to $\trans \colon \Loc \times \INPUT^* \to \Loc$ by $\trans(q,\epsilon) = q$ and $\trans(q,w\cdot a) = \trans(\trans(q,w),a)$, where $w \in \INPUT^*$ and $a \in \INPUT$. We also extend the output function $\outp \colon \Loc \times \INPUT \to \OUTPUT$ to $\outp \colon \Loc \times \INPUT^+ \to \OUTPUT$
by $\outp(q,a\cdot w) = \outp(\trans(q,a),w)$, where $a \in \INPUT$ and $w \in \INPUT^+$.}%
\ShortVersion{By abuse of notation, we extend $\trans$ and $\outp$ for words in the standard manner.}
We also define the function $\M \colon \INPUT^+ \to \OUTPUT$ by $\M(w) := \outp(\initLoc,w)$.
In \emph{symbolic} finite automata~\cite{DBLP:conf/popl/DAntoniV14}, transitions are labeled by predicates over an \emph{effective Boolean algebra}\footnote{We use the definition of Boolean algebras used in the context of symbolic automata, \eg{}~\cite{DBLP:conf/tacas/DrewsD17}, for clear comparison with existing results. In a lattice-theoretic definition, this is often formulated as a Boolean algebra $\Predicates$ together with a homomorphism $\sem{\_}$}.
\begin{definition}
    [Effective Boolean Algebra]
    A \emph{Boolean algebra} $\Algebra$ is a tuple $\Algebra = (\domain, \Predicates,\bot,\top,\lor,\land,\neg,\sem{\_})$, where $\domain$ is the domain, $\Predicates$ is a set of predicates closed under Boolean connectives, $\bot \in \Predicates$ is the bottom predicate, $\top \in \Predicates$ is the top predicate, and $\sem{\_} \colon \Predicates \to \powerset{\domain}$ is a denotation function. It satisfies the following conditions:
    \begin{ienumeration}
        \item $\sem{\bot} = \emptyset$;
        \item $\sem{\top} = \domain$; and
        \item For all $\varphi, \psi \in \Predicates$, $\sem{\varphi \lor \psi} = \sem{\varphi} \cup \sem{\psi}$, $\sem{\varphi \land \psi} = \sem{\varphi} \cap \sem{\psi}$, $\sem{\neg \varphi} = \domain \setminus \sem{\varphi}$.
    \end{ienumeration}
    A Boolean algebra $\Algebra$ is \emph{effective} if the denotation function is computable.
\end{definition}

\begin{example}
    [Equality algebra]%
    \label{example:equality-algebra}
    For any set $\domain$, the \emph{equality algebra} over $\domain$ is defined as follows. 
    The set of predicates is the Boolean closure of $\{\lambda x. x = i \mid i \in \domain\}$.
    The denotation function %
    is such that $\sem{\lambda x. x = i} = \{i\}$.
\end{example}
\begin{example}
    [Interval algebra]\label{example:interval-algebra}
    The \emph{interval algebra} over naturals is defined as follows. The domain is $\domain = \mathbb{N}$. 
    The set of predicates is the Boolean closure of $\{[i,j)\mid i,j \in \domain\} \cup \{[i, +\infty) \mid i \in \domain\}$.
    The denotation function %
    is such that $\sem{[i,j)} = \{k \in \domain \mid i \leq k < j\} $ and $\sem{[i,+\infty)} = \{k \in \domain \mid i \leq k \}$.
\end{example}
\begin{definition}
 [s-FA]
 A \emph{symbolic finite automaton (s-FA)} $\Autom$ is a tuple $\Autom = (\Algebra, \Loc, \initLoc, F, \transition)$, where
 $\Algebra$ is an effective Boolean algebra,
 $\Loc$ is a non-empty finite set of states,
 $\initLoc \in \Loc$ is the initial state,
 $F \subseteq \Loc$ is the set of accepting states, and
 $\transition \subseteq \Loc \times \Predicates \times \Loc $ is the transition relation.
\end{definition}

For $a \in \domain$ and $\loc_1, \loc_2 \in \Loc$,
we write $\loc_1 \xrightarrow{a} \loc_2$ if there is $\guard$ such that $a \in \sem{\guard}$ and $(\loc_1, \guard, \loc_2) \in \transition$.
A word $w = a_1 a_2 \ldots a_k$ is \emph{accepted} by an s-FA $\Autom$ if there are $q_1, q_2, \dots, q_k$ such that $\loc_0 =\initLoc$, $\loc_{i-1} \xrightarrow{a_i} \loc_i$ for each $i \in \{1,2\dots,k\}$, and $\loc_k \in F$.
The language $\Lg(\Autom)$ of an s-FA $\Autom$ is the set of words accepted by $\Autom$. %
An s-FA $\Autom$ is \emph{deterministic}
if $\loc_1 \neq \loc_2$ implies $\sem{\guard_{1} \land \guard_{2}} = \emptyset$ for any $(\loc, \guard_1, \loc_1), (\loc, \guard_2, \loc_2) \in \transition$.
An s-FA $\Autom$ is \emph{complete} if  $\sem{\bigvee_{(\loc,\guard_{i},\loc_{i}) \in \transition} \guard_{i} }= \domain$ for all $\loc \in \Loc$.
\LongVersion{In this paper, we consider only deterministic and complete s-FA.  

}
A list\LongVersion{\footnote{We define partitions as lists of predicates rather than sets of predicates for the sake of s-FA learning.}} $\partition = \varphi_{1} \dots \varphi_{k} \in \Predicates^{*}$ of predicates is called a \emph{partition} if for any $i, j \in \{1,2,\dots,k \}$, $i \neq j$ implies $\sem{\guard_{i} \land \guard_{j}} = \emptyset$ and $\sem{\bigvee_{i = 1}^k \guard_{i}} = \domain$. We write $\Partitions$ for the set of partitions. For any deterministic and complete s-FA $\Autom$ and for any state $\loc$ of $\Autom$, the list of the predicates associated with the transitions leaving from $\loc$ constitutes a partition.
In this paper, we study \emph{symbolic Mealy automata (s-MAs)}, which is a special case of \emph{symbolic finite transducers}~\cite{DBLP:conf/popl/VeanesHLMB12}.
Intuitively, an s-MA is a Mealy automaton over a possibly infinite input alphabet, where the transitions are labeled by predicates over a Boolean algebra.
\begin{definition}
[s-MA]
A \emph{symbolic Mealy automaton (s-MA)} $\M$ is a tuple $\M = (\Algebra,\Loc,\initLoc,\OUTPUT,\transition)$, where
 $\Algebra$ is an effective Boolean algebra,
 $\Loc$ is a non-empty finite set of states,
 $\initLoc \in \Loc$ is the initial state,
 $\OUTPUT$ is a finite set of output characters, and 
 $\transition \subseteq \Loc \times \Predicates \times \Loc \times \OUTPUT$ is the transition-output relation.
 We write $q_1 \xrightarrow{a\mid o} q_2$ for %
 $(\loc_1,\guard,\loc_2,o) \in \transition$ and $a \in \sem{\guard}$.
\end{definition}
An s-MA $\M$ is \emph{deterministic} if $q_1\neq q_2$ or $o_1 \neq o_2$ implies $\llbracket \varphi_{1}\land \varphi_{2} \rrbracket = \emptyset$ for any $(q,\varphi_1,q_1,o_1)$, $(q,\varphi_2,q_2,o_2) \in \delta$.  An s-MA $\M$ is \emph{complete} if $\sem{\bigvee_{(q,\varphi_{i},q_{i},o_{i}) \in \delta} \varphi_{i}} = \domain$ for any $q \in \Loc$.
If $\M$ is deterministic and complete, we define the \emph{transition function} $\trans \colon \Loc \times \domain \to \Loc$ and the \emph{output function} $\outp \colon \Loc \times \domain \to \OUTPUT$ of $\M$ as follows: $\trans(q,a) = q'$ and $\outp(q,a) = o$ if there exists $\varphi$ such that $(q, \varphi, q', o) \in \delta$ and $a \in \llbracket \varphi \rrbracket$\LongVersion{; notice that $\trans$ and $\outp$ are well-defined since $\M$ is deterministic and complete}.
We extend $\trans$ and $\outp$ to words $w \in \domain^+$ as follows: $\trans \colon \Loc \times \domain^* \to \Loc$ is defined by $\trans(q,\epsilon) = q$ and $\trans(q,wa) = \trans(\trans(q,w),a)$; %
$\outp \colon \Loc \times \domain^+ \to \OUTPUT$ is defined by  $\outp(q,aw) = \outp(\trans(q,a),w)$. %
For any word $w \in \domain^+$, we write
$\M(w)$ for $\outp(\initLoc,w)$. 
In this paper, we consider only deterministic and complete s-MAs.

\subsection{The $\Lambda^*$ Algorithm for Learning s-FAs}\label{section:lambda_star}

$\Lambda^*$~\cite{DBLP:conf/tacas/DrewsD17} is an algorithm for active learning of s-FAs.
Given an effective Boolean algebra $\Algebra$, $\Lambda^*$ %
learns an s-FA recognizing the target language $\targetLg$ using two kinds of queries to an oracle: \emph{membership} and \emph{equivalence} queries.
In a membership query, the learner submits a word $\word \in \domain^*$ to the oracle; the oracle then answers whether $w \in \targetLg$ or not.
In an equivalence query, the learner submits a hypothesis s-FA $\hypothesis$ to the oracle;
the oracle returns either
``\KTrue'' indicating $\Lg(\hypothesis) = \targetLg$ or a counterexample $\cex \in \domain^*$ satisfying $\cex \in \Lg(\hypothesis) \setdiff \targetLg$.
To construct a hypothesis s-FA, $\Lambda^*$ first learns an \emph{evidence automaton} $\Autom^e$, whose transitions are labeled with some elements of $a \in \domain$.
$\Lambda^*$ uses a variant of L* to learn an evidence automaton\LongVersion{, possibly with partial transition functions}.
Then, $\Lambda^*$ applies a \emph{partitioning function} to convert the\LongVersion{ concrete} characters appearing in\LongVersion{ the evidence automaton} $\Autom^e$ into %
predicates. %
 
$\Lambda^*$ uses a data structure called \emph{observation tables} to learn an evidence automaton. An observation table consists of a finite set of prefixes, a finite set of suffixes, and a 2-D array that keeps information on the acceptance of the words obtained by concatenating each prefix and each suffix.
The acceptance is checked by asking membership queries. %

\begin{definition}
    [Observation Table for $\Lambda^*$]
An \emph{observation table} $T$ is a tuple $T = (\domain, S, R, E, f)$, where $\domain$ is the domain of the Boolean algebra $\Algebra$, $S,R,E \subseteq \domain^*$\LongVersion{ are finite subsets of words over $\domain$}, and $f \colon (S\cup R) \times E \to \{0,1\}$ is a function where $f(w,e)= 1$ if $w \cdot e \in \targetLg$, and $f(w,e) = 0$ otherwise. In addition, the following conditions hold:
\begin{ienumeration}
    \item $S\cup R$ is prefix-closed;
    \item $\emptyword \in S$;
    \item $E$ is suffix-closed (thus, we have $\emptyword \in E$);
    \item For all $s\in S$, there exists a character $a \in \domain$ such that $s\cdot a \in S\cup R$.

\end{ienumeration}
\end{definition} 

For $w \in S\cup R$, we let $\row(w) = \{e \in E \mid f(w,e) = 1\}$.
$\Lambda^*$ increases $S$, $R$, and $E$ until the observation table becomes \emph{cohesive}, which is the condition necessary for constructing an evidence automaton\footnote{Evidence-closedness is not necessary for constructing a well-defined hypothesis but is necessary for \cref{theo:symbcomp}, which can also reduce the number of equivalence queries.}.

\begin{ShortVersionBlock}
    \begin{definition}
    [Cohesiveness]
    Let $T = (\domain, S, R ,E, f)$ be an observation table.
    $T$ is \emph{closed} if for any $r \in R$, there exists $s \in S$ satisfying $\row(s) = \row(r)$.
    $T$ is \emph{consistent} if for any $w_1,w_2 \in S\cup R$ satisfying $\row(w_1) = \row(w_2)$ and for any $a \in \domain$, $w_1 \cdot a,w_2 \cdot a \in  S\cup R$ implies $\row(w_1\cdot a) = \row(w_2\cdot a)$.
    $T$ is \emph{evidence-closed} if for any $s \in S$ and $e \in E$, $s\cdot e \in S\cup R$ holds.
    $T$ is \emph{cohesive} if it is closed, consistent and evidence-closed.
\end{definition}
\end{ShortVersionBlock}
\begin{LongVersionBlock}
\begin{definition}
    [Cohesiveness]
    An observation table $T = (\domain, S, R, E, f)$ is \emph{cohesive} if it satisfies the following conditions.
    \begin{description}
        \item[Closedness] For all $r \in R$ there exists $s \in S$ such that $\row(s) = \row(r)$.
        \item[Consistency] For all $w_1,w_2 \in S\cup R$ such that $\row(w_1) = \row(w_2)$, for all $a \in \domain$ such that $w_1 \cdot a,w_2 \cdot a \in  S\cup R$, $\row(w_1\cdot a) = \row(w_2\cdot a)$.
        \item[Evidence-closedness] For all $s \in S$ and $e \in E$, $s\cdot e \in S\cup R$.
    \end{description}
\end{definition}
\end{LongVersionBlock}

From a cohesive observation table, one can construct a (possibly incomplete) DFA called an \emph{evidence automaton}. %
\LongVersion{\begin{definition}
    [Evidence Automata] \label{EvidenceAutom}
     For a cohesive observation table $T = (\domain, S, R ,E, f)$,
     the \emph{evidence automaton}\LongVersion{ $\Autom^e$ defined from $T$} is $\Autom^e = (\domain,\Loc,\initLoc,F,\delta^e)$, where
\begin{myitemize}
    \item $\Loc = \{\row(s) \mid s \in S \}$,
    \item $\initLoc = \row(\epsilon)$,
    \item $F = \{ \row(s) \mid s \in S$ and $ f(s,\epsilon) = 1\}$, and
    \item $\delta^e(\row(w),a) = \row(w\cdot a)$ for $w \in S$ and $a \in \domain$ satisfying $w \cdot a \in S \cup R$.
\end{myitemize}
\end{definition}}%
\LongVersion{

}%
\LongVersion{In $\Lambda^*$}\ShortVersion{Then}, a \emph{partitioning function} $\partitioning$ is used
to construct an s-FA from an evidence automaton. %
Given a list of pairwise disjoint sets of characters, $\partitioning$ returns a partition consistent with the list.
\begin{definition}
[Partitioning Function]\label{partition}
For an effective Boolean algebra $\Algebra$, a \emph{partitioning function} over $\Algebra$ is a function $P \colon (\finpowerset{\domain})^{*} \to \Partitions$ that satisfies the following:
Let $L_{\domain} = l_{1}\ldots l_{k}\in (\finpowerset{\domain})^*$ satisfying $l_i \cap l_j = \emptyset$ for any $i,j \in \{1,\dots,k\}$ with $i \neq j$;
We have $|L_{\domain}| = |P(L_{\domain})|$;
For $\varphi_{1} \ldots \varphi_{k} = P(L_{\domain})$,
we have
\begin{ienumeration}
\item $a \in \sem{\varphi_{i}}$ for any $i \in \{1,2,\dots,k \}$ and $a\in l_{i}$ and                
\item\label{item:partitioning-stability} for any $L'_{\domain} = l'_{1}\ldots l'_{k}$, if $l_i \subseteq l'_i \subseteq \sem{\varphi_i}$ for all $i \in \{1,2,\dots,k\}$, then $\partitioning(L_{\domain}) = \partitioning(L'_{\domain})$.
\end{ienumeration}
\end{definition}
\LongVersion{Condition~\ref{item:partitioning-stability}) in \cref{partition} is to guarantee the stability of the learning procedure.}

For example, the function outlined in \cref{algorithm:partitioning_interval_algebra} is a partitioning function for the interval algebra over naturals introduced in \cref{example:interval-algebra}.
For $l_1 = \{2, 7, 10\}$ and $l_2 = \{5\}$, \cref{algorithm:partitioning_interval_algebra} returns $\varphi_1 = [0, 5) \cup [7, \infty)$ and $\varphi_2 = [5, 7)$.

\begin{algorithm}[tbp]
    \caption{A partitioning function for the interval algebra over $\N$}\label{algorithm:partitioning_interval_algebra}
 \LongVersion{\small}\ShortVersion{\footnotesize}
    \KwIn{A list $l_{1}\ldots l_{k}$ of pairwise disjoint finite sets $l_i \subseteq \N$}
    $\varphi_{1}, \varphi_2, \dots, \varphi_{k} \gets \bot, \bot, \dots, \bot$;\,
    $b \gets +\infty$\;
    \While{$\exists i \in \{1,2,\dots,k\}.\, l_i \neq \emptyset$} {
        \KwPop{} the maximum $a \in l_i$ in $l_{1}\ldots l_{k}$ \KwFrom{} $l_i$\;
        $\varphi_i \gets \varphi_i \lor [a, b)$;\,
        $b \gets a$;\,
        $i' \gets i$\;
    }
    $\varphi_{i'} \gets \varphi_{i'} \lor [0, a)$\;
    \KwReturn{} $\varphi_{1} \varphi_2 \dots \varphi_{k}$
\end{algorithm}

$\Lambda^*$ constructs a s-FA from an evidence automaton $\Autom^e = (\domain,\Loc,\initLoc,F,$ $\delta^e)$ by the following algorithm, which we call \texttt{sepPred}\label{seppred}.
Let $P$ be a partitioning function over $\Algebra$.
For each $q \in \Loc$, let $l_{q}\colon \Loc \to \finpowerset{\domain}$ be\LongVersion{ such that $l_q(q')$ is the set of the characters $a$ such that $q \xrightarrow{a} q'$ (\ie{}} $l_{q}(q') = \{ a \in \domain \mid q' = \delta^e(q, a)\}$\LongVersion{)}.
Then, we give the list of sets
$L_{q}=l_{q}(q_1) l_{q}(q_2) \dots l_{q}(q_n)$, where $\Loc = \{q_1,q_2, \ldots,q_n\}$, to the partitioning function $\partitioning$.
Let\LongVersion{ the list} $\guard_{q}^{q_1} \dots \guard_{q}^{q_n} = P(L_{q})$.
Finally, we add each $(q, \guard_{q}^{q_i}, q_i)$ to $\delta$ unless $\guard_{q}^{q_i} = \bot$. 
 \SetKwFunction{FMakeClosed}{makeClosed}
 \SetKwFunction{FMakeConsistent}{makeConsistent}
 \SetKwFunction{FMakeEvidenceClosed}{makeEvidenceClosed}
 \SetKwFunction{FGenerateHypothesis}{generateHypothesis}
 \SetKwFunction{FEqQuery}{\ensuremath{\texttt{eq}_{\targetLg}}}
\begin{LongVersionBlock}
 \begin{algorithm}[tb]
    \caption{\LongVersion{The }$\Lambda^*$\LongVersion{ algorithm}~\cite{DBLP:conf/tacas/DrewsD17} for active learning of s-FAs}\label{alg1}
    \footnotesize
    \KwIn{A Boolean algebra $\Algebra$ with domain $\domain$ and a partitioning function $\partitioning$}
    \KwOut{An s-FA $\hypothesis$ satisfying $\Lg(\hypothesis) = \targetLg$}
 $\Prefixes \gets \{\emptyword\}$;\, $\NextPrefixes \gets \{ a \}$ for some $a \in \domain$;\, $\Suffixes \gets \{\emptyword\}$\label{initialize}\;
 \textbf{Construct the initial observation table $T = (\domain,\Prefixes,\NextPrefixes,\Suffixes,f)$}\label{initaltable}\;
 \While{\KTrue}{
    \While{$\Table$ is not cohesive}{\label{cohesive}
        \lIf{$\Table$ is not closed}{
            $\Table \gets \FMakeClosed{\Table}$\label{close}
        }
        \lElseIf{$\Table$ is not consistent}{
            $\Table \gets \FMakeConsistent{\Table}$\label{make-consisitent}
        }
        \lElseIf{$\Table$ is not evidence-closed}{
            $\Table \gets \FMakeEvidenceClosed{\Table}$\label{evidence-close}
        }
    }\label{cohesiveend}
    $\hypothesis \gets \FGenerateHypothesis{\Table, \partitioning}$\label{construct-hypothesis}\tcp*[f]{Predicates are inferred by $\partitioning$}\;
    \Switch{$\FEqQuery(\hypothesis)$}{
        \lCase{\KTrue} {
            \KwReturn{} $\hypothesis$\label{return}
        }
        \lCase{$\cex \in \domain^*$} {
            \KwPush{} $\prefixes(\cex) \setminus (\Prefixes \cup \NextPrefixes)$ \KwTo{} $\NextPrefixes$\label{process-counterexample}
        }
    }
 }
 \end{algorithm}
 \cref{alg1} outlines $\Lambda^*$. 
 It first initializes $S$ to $\lbrace \epsilon \rbrace$, $R$ to $\lbrace a \rbrace $, and $E$ to $\lbrace \epsilon \rbrace$, where $a$ is an arbitrary character in $\domain$ (\cref{initialize}), and constructs the initial observation table $T = (\domain,S,R,E,f)$ (\cref{initaltable}). 
 \reviewer{3}{* Algorithm 2: There is no explanation at all what the make(-) operations do.}
 \mw{I do not think we have enough space to write more than this.}
 Then, the operations $\FMakeClosed$, $\FMakeConsistent$, and $\FMakeEvidenceClosed$ are applied repeatedly
 until the observation table is cohesive (\crefrange{cohesive}{cohesiveend}).
 Once we obtain a cohesive observation table, 
 the learner constructs
 an s-FA $\hypothesis$ %
 from the observation table %
 and poses an equivalence query with it %
 (\cref{construct-hypothesis}). 
 If the oracle returns ``\KTrue'',  $\Lambda^*$ returns $\hypothesis$ and terminates (\cref{return}).
Otherwise, the oracle returns a counterexample $\cex \in \Lg(\hypothesis) \setdiff \targetLg$; then, $\Lambda^*$ adds all the prefixes of $\cex$ to $R$, except for those already in $S \cup R$ (\cref{process-counterexample}), and repeats this process. %
\end{LongVersionBlock}

\subsection{Learnability of Boolean Algebras}\label{learnability}

In~\cite{DBLP:conf/tacas/DrewsD17},
Drews and D'Antoni discussed the learnability of s-FA using
a mathematical notion called a \emph{generator}, which was used to capture the learnability of a Boolean algebra. 
They formulated the learnability of Boolean algebras with respect to the generator $g$ as \emph{$s_g$-learnability} and used it to discuss the termination and complexity of the entire algorithm.
A generator $g$ is a function that models one-step refinement of the sampled data $L = l_1\dots l_h$ to be given to a partitioning function to learn a partition.
Intuitively, $L' = g(L, \partitioning(L), \targetP)$ represents a part of the learning behavior of $\Lambda^*$, where $L$ corresponds to $L_q$ in \texttt{sepPred}:
the learner constructs the predicates in a hypothesis automaton using $\partitioning(L)$ and refines $L$ via queries to the oracle, which has the knowledge of the target partition $\targetP$.
The definition of generators captures this behavior: if a hypothesis partition $\partitioning(L)$ generated from $L$ is different from the target partition $\targetP$, the returned $L'$ refines $L$ with a witness of $\partitioning(L) \neq \targetP$.
\begin{definition}
 [Generator~\cite{DBLP:conf/tacas/DrewsD17}]\label{generator}
 For a Boolean algebra $\Algebra$ and a partitioning function $\partitioning$,
 a \emph{generator} is a function $g \colon {(\finpowerset{\domain})}^{*} \times \Partitions \times \Partitions \to (\finpowerset{\domain})^{*}$ \st{}, for any $L = l_1\dots l_h  \in {(\finpowerset{\domain})}^{*}$, $\targetP \in \Partitions$, and $L' = g(L, \partitioning(L), \targetP) = l_1\dots l_{h'} $:
    \begin{itemize}
    \item if $l_{1}, l_{2} \dots l_{h}$ are pairwise disjoint, then $l'_{1}, l'_{2} \dots l'_{h'}$ are also pairwise disjoint;
    \item %
    $\bigcup_{i \in \{1,2,\dots,h\}} l_i \subseteq \bigcup_{i' \in \{1,2,\dots,h'\}} l'_{i'}$;
    \item $L'$ is a refined partition with respect to $L$, \ie{} for any $a, a' \in \domain$ satisfying $a \in l_i$ and $a' \in l_j$ for some $i, j \in \{1,2,\dots,h\}$ satisfying $i \neq j$,
          there is no $i' \in \{1,2,\dots,h'\}$ satisfying $a \in l'_{i'}$ and $a' \in l'_{i'}$; and
    \item Exactly one of the following conditions holds:
          \begin{itemize}
           \item $\partitioning(L) \neq \targetP$, $\bigcup_{i \in \{1,2,\dots,h\}} l_i \subsetneq \bigcup_{i' \in \{1,2,\dots,h'\}} l'_{i'}$, and $h = h'$;
           \item $\partitioning(L) \neq \targetP$, $\bigcup_{i \in \{1,2,\dots,h\}} l_i \subseteq \bigcup_{i' \in \{1,2,\dots,h'\}} l'_{i'}$, and $h < h'$; or
           \item $\partitioning(L) = \targetP$, $\bigcup_{i \in \{1,2,\dots,h\}} l_i \subseteq \bigcup_{i' \in \{1,2,\dots,h'\}} l'_{i'}$, and $h = h'$.
          \end{itemize}
    \end{itemize}
\end{definition}

For a Boolean algebra $\Algebra$, its partitioning function $\partitioning$, and a partition $\targetP \in \Partitions$,
we inductively define $L^{\targetP}_{0}, L^{\targetP}_{1}, \dots$ as
\begin{myitemize}
 \item $L^{\targetP}_{0} = \emptyset$ and
 \item $L^{\targetP}_{i} = g(L^{\targetP}_{i-1}, \partitioning(L^{\targetP}_{i-1}), \targetP)$.
\end{myitemize}
A generator $g$ for $\Algebra$ and $\partitioning$ is \emph{finitely exhaustive} for $\targetP \in \Partitions$ if,
for the above sequence $L^{\targetP}_{0}, L^{\targetP}_{1}, \dots$, there is $i \in \N$ such that
$L^{\targetP}_{i} = L^{\targetP}_{j}$ for any $j \geq i$; we write $L^{\targetP}_{\infty}$ for such $L^{\targetP}_{i}$.
By the definition of generators, $\partitioning(L^{\targetP}_{\infty}) = \targetP$ holds.
Finally, we present \emph{$s_g$-learnability}, which defines the difficulty of learning partitions consisting of the predicates in the Boolean algebra $\Algebra$ using a generator $g$ and a partitioning function $P$. Intuitively, for $\targetP \in \Partitions$, $s_g$ is an upper bound of the number of samples required to learn $\targetP$ with $g$. %

\begin{definition}
 [$s_g$-learnability~\cite{DBLP:conf/tacas/DrewsD17}]
 Let $\Algebra$ be a Boolean algebra, $P$ be a partitioning function, and $g$ be a generator. 
 $\Algebra$ is \emph{$s_g$-learnable} with $\partitioning$ 
 if for any $\targetP \in \Partitions$, 
 $g$ is finitely exhaustive for $\targetP$. 
\end{definition}
\ks{What is the definition of $\flatten{L}$?}
\mw{This is the set obtained by flattening a list of sets. This is defined at the beginning of \cref{section:preliminaries}}

For instance, the equality algebra (\cref{example:equality-algebra}) for any finite alphabet and interval algebra (\cref{example:interval-algebra}) for integers are $s_g$-learnable for any generator $g$ for some partitioning functions.
Moreover, for any $s_g$ learnable Boolean algebras, their product is also $s_g$ learnable.
See~\cite{DBLP:conf/tacas/DrewsD17} for other examples.

\section{Active Learning of Symbolic Mealy Automata}

We present our algorithm, $\LambdaM$, for learning s-MAs. 
Our algorithm is based on $\Lambda^*$ with an extension to 
learn an output function that maps a pair of a state and an input character to an output character.

In the learning algorithm~\cite{DBLP:phd/de/Niese2003} for (non-symbolic) Mealy automata, an output function is learned %
by recording the output for \emph{all} the elements of the input alphabet for each state in an observation table.
However, since the input alphabet of an s-MA may be infinite, this technique does not work. 

To overcome this issue, we maintain a finite subset $\Sigma_E$ of the input alphabet 
to learn an output function and refine it as learning progresses.
Concretely, we introduce a new condition \emph{output-closedness} of an observation table and introduce a set $\Sigma_E$ to record the essential characters.

\subsection{Observation Table and Algorithm Description}

Given a Boolean algebra $\Algebra$, $\LambdaM$ learns an s-MA that returns the same output for all inputs as the target $\targetM$ using two kinds of queries to an oracle: 
In an \emph{output query}, the learner submits a word $w \in \domain^+$ to the oracle and obtains an output character $\targetM(w)$. In an \emph{equivalence query}, the learner submits a hypothesis s-MA $\hypothesisM$ to the oracle, and the oracle returns either ``\KTrue'' indicating that for all $w \in \domain^+$, we have $\hypothesisM(w) = \targetM(w)$, or there is a counterexample $\cex \in \domain^+$ such that $\hypothesisM(\cex) \neq \targetM(\cex)$.

$\LambdaM$ also uses an observation table. Unlike $\Lambda^*$, an observation table keeps information on an output character of target $\targetM$ instead of acceptance and non-acceptance. Moreover, the column index is changed from $E$ to $\Sigma_E \cup E$.

\begin{definition}
[Observation Table for $\LambdaM$]\label{mot}
An \emph{observation table} $T$ is a tuple $T = (\domain, S, R, \Sigma_E, E, f)$, where $\domain$ is the domain of the Boolean algebra, $S,R,E \subseteq \domain^*$\LongVersion{ are finite subsets of words}, $\Sigma_E \subseteq \domain$\LongVersion{ is a finite subset of characters}, and
$f \colon (S\cup R)\times (\Sigma_E \cup E) \to \OUTPUT$ is\LongVersion{ a function} such that $f(w, e)= \targetM(w\cdot e)$. In addition, the following conditions hold:
\begin{ienumeration}
    \item $S\cup R$ is prefix-closed;
    \item $\Sigma_E \cup E$ is suffix-closed;
    \item $\epsilon \in S$ and $\Sigma_E$ is non-empty. %
\end{ienumeration} 
\end{definition}

$\LambdaM$ updates the observation table and constructs an s-MA. For this, the table must satisfy certain properties. We call such a table \emph{cohesive}. 
For $w \in S\cup R$, %
we let $\row(w)\colon \Sigma_E \cup E \to \OUTPUT$ be $\row(w)(e) = f(w,e)$.

\begin{ShortVersionBlock}
    \begin{definition}
        [Cohesiveness]
        Let $T = (\domain, S, R, E,f)$ be an observation table.
        $T$ is \emph{closed} if for any $r \in R$, there is $s \in S$ such that $\row(s) = \row(r)$.
        $T$ is \emph{consistent} if for any $w_1,w_2 \in S\cup R$ satisfying $\row(w_1) = \row(w_2)$, for all $a \in \domain$, $w_1 \cdot a,w_2 \cdot a \in  S\cup R$ implies $\row(w_1\cdot a) = \row(w_2\cdot a)$.
        $T$ is \emph{evidence-closed} if for any $s \in S$ and $e \in \Sigma_E$, $s\cdot e \in S\cup R$ holds. 
        $T$ is \emph{output-closed} if for any $w \in S\cup R$ and for any $a \in \domain$, $w\cdot a \in S\cup R$ implies $a \in \Sigma_E$.
        $T$ is \emph{cohesive} if it is closed, consistent, evidence-closed, and output-closed.
    \end{definition}    
\end{ShortVersionBlock}
\begin{LongVersionBlock}
    \begin{definition}
        [Cohesiveness]
        An observation table $T = (\domain, S, R, E,f)$ is \emph{cohesive} if it satisfies the following conditions.
        \begin{description}
            \item[Closedness] For all $r \in R$ there exists $s \in S$ such that $\row(s) = \row(r)$.
            \item[Consistency] For all $w_1,w_2 \in S\cup R$ such that $\row(w_1) = \row(w_2)$, for all $a \in \domain$ such that $w_1 \cdot a,w_2 \cdot a \in  S\cup R$, $\row(w_1\cdot a) = \row(w_2\cdot a)$.
            \item [Evidence-Closedness] For all $s \in S$ and $e \in \Sigma_E$, $s\cdot e \in S\cup R$. 
            \item [Output-Closedness] For all $w \in S\cup R$, for all $a \in \domain$, if $w\cdot a \in S\cup R $ then $a \in \Sigma_E$.
        \end{description}
    \end{definition}    
\end{LongVersionBlock}
\todo{If the space allows, we want to briefly remark that the definition of evidence-closedness is updated.}
 Together with prefix-closedness of $S \cup R$,\LongVersion{ the new condition} output-closedness guarantees that all characters appearing in the observation table are in $\Sigma_E$. 
Unlike $\Lambda^*$, evidence-closedness focuses on $\Sigma_E$ rather than on the entire column indices\LongVersion{ $\Sigma_E \cup E$}.

From a cohesive observation table, we construct a deterministic Mealy automaton called \emph{evidence Mealy automaton}.
The input alphabet of an evidence Mealy automaton in $\LambdaM$ is \emph{not} $\domain$ but $\Sigma_E$.
Since the observation table is evidence-closed,
$\row(s\cdot a)$ and $f(s,a)$ are defined for all $s \in S$ and for all $a \in  \Sigma_E$, and thus,
an evidence Mealy automaton is \emph{complete}.

\begin{definition}
 [Evidence Mealy Automata]\label{construction}
 For a cohesive observation table $T = (\domain, S, R, \Sigma_E, E, f)$,
 the \emph{evidence Mealy automaton} $\hypothesisMe$ is $\hypothesisMe = (\Sigma_E,\Loc,\initLoc,\OUTPUT,\transe,\outpe)$, where
\begin{myitemize}
\item $\Loc = \{\row(s) \mid s \in S \}$,
\item $\initLoc = \row(\epsilon)$,
\item $\transe(\row(s),a) = \row(s\cdot a)$ for all $s \in S$, $a \in \Sigma_E$, and
\item $\outpe(\row(s),a) = f(s,a)$
for all $s \in S$, $a \in \Sigma_E$.
\end{myitemize}
\end{definition}

Let $\hypothesisMe = (\Sigma_E,\Loc,\initLoc,\OUTPUT,\transe,\outpe)$ be an evidence Mealy automaton and $\partitioning$ be a partitioning function. An s-MA $\hypothesisM = (\Algebra,\Loc,\initLoc,\OUTPUT,\delta)$ is defined by generating partitions using $\partitioning$. The following algorithm constructs the transition-output relation. 
For each $q \in \Loc$, let $l_{q}\colon \Loc \times \OUTPUT \to \finpowerset{\domain}$ be the function mapping the target state $q' \in \Loc$ and\LongVersion{ the output character} $o \in \OUTPUT$ to the set of characters to transit and output from $q$, \ie{} $l_{q}(q',o) = \{ a \in \domain \mid q' = \transe(q, a)$ and $ o= \outpe(q,a)\}$.
Then, the list of the set
$L_{q}=l_{q}(q_1,o_1)l_{q}(q_1,o_2)\ldots l_{q}(q_1,o_m) l_{q}(q_2,o_1) \dots l_{q}(q_n,o_m)$, where $\Loc = \{q_1,q_2, \ldots,q_n\}$ and $\OUTPUT = \{o_1,o_2,\ldots,o_m\}$, is given to the partitioning function.
For the list of predicates $\guard_{q,o_1}^{q_1} \dots \guard_{q,o_m}^{q_1} \guard_{q,o_1}^{q_2} \dots \guard_{q,o_m}^{q_n} = P(L_{q})$,
we add each $(q, \guard_{q,o_j}^{q_i}, q_i,o_j)$ to $\delta$ unless $\guard_{q,o_j}^{q_i} = \bot$.

Finally, we confirm that an s-MA constructed from a cohesive observation table has the minimum number of states.
These results are natural and similar to well-known properties in a classical setting.
However, they are not immediately obtained from standard results, \eg{} due to the subtleties related to the use of the partitioning function to construct an s-MA from an evidence Mealy automaton.

\newcommand{\symboliccompatibilityStatement}{ Let $T=(\domain, S, R, \Sigma_E, E, f)$ be a cohesive observation table, and $\partitioning$ be a partitioning function. The s-MA $\hypothesisM$ constructed from $T$ and $\partitioning$ is \emph{symbolic compatible} with $T$, \ie{} for all $s \in S\cup R$, for all $e \in \Sigma_E \cup E$, we have $\hypothesisM(s\cdot e) = f(s, e)$.
}
\begin{theorem}
    [Symbolic Compatibility]\label{theo:symbcomp}
\symboliccompatibilityStatement{}
\qed{}
\end{theorem}

\newcommand{\symbolicminimalityStatement}
{Let $T$ be a cohesive observation table, and $P$ be a partitioning function. The 
s-MA $\hypothesisM$ constructed from $T$ and $P$ has the minimum number of states among s-MAs symbolic compatible with $T$.}
\begin{theorem}
    [Minimality]\label{theo:symbolicminimality}
\symbolicminimalityStatement{}
\qed{}
\end{theorem}

\SetKwFunction{FMakeOutputClosed}{makeOutputClosed}

\cref{alg2} outlines $\LambdaM$. Given a Boolean algebra $\Algebra$ and a partitioning function $P$, it initializes $S$ to $\{ \epsilon \}$, $R$ to $\{ a \} $, $\Sigma_E$ to $\{ a \}$ and $E$ to $\emptyset$ where $a$ is an arbitrary character in $\domain$ (\cref{initializeM}) and constructs the initial observation table $T = (\domain,S,R,\Sigma_E,E,f)$ (\cref{initaltableM}). In contrast to the initialization in $\Lambda^*$, $\Sigma_E$ is initialized by $\{a\}$ for some $a \in \domain$, and $E$ is initialized by $\emptyset$. This is because $\LambdaM$ learns Mealy-style outputs.

Until the observation table becomes cohesive, the operations $\FMakeClosed$, $\FMakeConsistent$, $\FMakeEvidenceClosed$, and $\FMakeOutputClosed$ are repeatedly applied (\cref{cohesiveM}-\ref{cohesiveendM}).
Each of the operations works as follows. 
\begin{itemize}
    \item $\FMakeClosed$ picks $r \in R$ satisfying $\row(r) \neq \row(s)$ for any $s \in S$ and moves $r$ from $R$ to $S$.
    \item $\FMakeConsistent$ picks $w_1,w_2 \in S\cup R$, $a \in \domain$, and $e \in \Sigma_E \cup E$ satisfying $\row(w_1) = \row(w_2)$, $w_1 \cdot a,w_2 \cdot a \in S\cup R$, and $f(w_1 \cdot a ,e) \neq f(w_2 \cdot a, e)$. Then, it adds $a\cdot  e$ to $E$.
    \item $\FMakeEvidenceClosed$ picks $s \in S$ and $e \in \Sigma_E$ such that $s\cdot e \notin S\cup R$ and adds it to $R$.
    \item $\FMakeOutputClosed$ picks $w \in S\cup R$ and $a \in \domain$ satisfying $w \cdot a \in  S\cup R$ and $a \notin \Sigma_E$. Then it adds $a$ to $\Sigma_E$.
\end{itemize}

Once the observation table becomes cohesive,
$\LambdaM$ constructs an s-MA $\hypothesisM$ from the observation table by
\begin{ienumeration}
 \item building an evidence Mealy automaton and 
 \item applying \texttt{sepPred} to form an s-MA.
\end{ienumeration}
Then, $\LambdaM$ makes an equivalence query with $\hypothesisM$ (\cref{construct-hypothesisM}).
If the oracle returns ``\KTrue'', $\LambdaM$ returns $\hypothesisM$ and terminates (\cref{returnM}).
Otherwise, the oracle returns a counterexample $\cex \in \domain^+$, and $\LambdaM$ adds all prefixes of $\cex \in \domain^+$ to $R$, except those that are already in $S \cup R$ (\cref{process-counterexampleM}) and goes back to \cref{mainloopM}.
 %
%
%
%
%
%
%
%
%

%
\begin{comment}
Given a Boolean algebra $\Algebra$ and a partitioning function $P$, it initializes $S$ to $\{ \epsilon \}$, $R$ to $\{ a \} $, $\Sigma_E$ to $\{ a \}$ and $E$ to $\emptyset$, where $a$ is an arbitrary character in $\domain$ (\cref{initializeM}), and constructs the initial observation table\LongVersion{ $T = (\domain,S,R,\Sigma_E,E,f)$} (\cref{initaltableM}). 
The initialization of $\Sigma_E$ by $\{a\}$ is for the learning of Mealy-style outputs.
%
%
%

Then, the operations $\FMakeClosed$, $\FMakeConsistent$, $\FMakeEvidenceClosed$, and $\FMakeOutputClosed$ are applied repeatedly
until the observation table is cohesive (\crefrange{cohesiveM}{cohesiveendM}).
We changed these operations from $\Lambda^*$, \eg{} for the introduction of $\SigmaEf$.
See \cref{section:algorithm_detail:ours} for the details.
%
%
%
%
%
%
%
%
%
%
%
%
%
%
%
Once we obtain a cohesive observation table, the learner constructs an s-MA $\hypothesisM$ from the observation table %
and poses an equivalence query with it (\cref{construct-hypothesisM}).
If the oracle returns ``\KTrue'', $\LambdaM$ returns $\hypothesisM$ and terminates (\cref{returnM}). Otherwise, the oracle returns a counterexample $\cex \in \domain^+$; then $\LambdaM$ adds all prefixes of $\cex$ to $R$, except those already present in $S \cup R$ (\cref{process-counterexampleM}), and repeats this process.
\end{comment}

%
\SetKwFunction{FMakeClosed}{makeClosed}
\SetKwFunction{FMakeConsistent}{makeConsistent}
\SetKwFunction{FMakeEvidenceClosed}{makeEvidenceClosed}
\SetKwFunction{FMakeOutputClosed}{makeOutputClosed}
\SetKwFunction{FGenerateHypothesis}{generateHypothesis}
\SetKwFunction{FEqQuery}{\ensuremath{\texttt{eq}_{\targetM}}}
\begin{algorithm}[tbp]
    \caption{\LongVersion{The }$\LambdaM$\LongVersion{ algorithm}for active learning of s-MAs}\label{alg2}
    \LongVersion{\small}\ShortVersion{\footnotesize}
    \KwIn{A Boolean algebra $\Algebra$ with domain $\domain$ and a partitioning function $\partitioning$}
    \KwOut{An s-MA $\hypothesisM$ satisfying for all $w \in \domain^+$, $\hypothesisM(w) = \targetM(w)$}
$\Prefixes \gets \{\emptyword\}$;\, $\NextPrefixes \gets \{ a \}$;\, $\Sigma_E \gets \{ a \}$\, for some $a \in \domain$;\, $\Suffixes \gets \emptyset$\label{initializeM}\;
\textbf{Construct the initial observation table $T = (\domain,\Prefixes,\NextPrefixes,\Sigma_E,\Suffixes,f)$}\label{initaltableM}\;
\While{true}{ \label{mainloopM}
    \While{$\Table$ is not cohesive}{\label{cohesiveM}
        \lIf{$\Table$ is not closed}{
            $\Table \gets \FMakeClosed{\Table}$\label{closeM}
        }
        \lElseIf{$\Table$ is not consistent}{
            $\Table \gets \FMakeConsistent{\Table}$\label{make-consisitentM}
        }
        \lElseIf{$\Table$ is not evidence-closed}{
            $\Table \gets \FMakeEvidenceClosed{\Table}$\label{evidence-closeM}
        }
        \lElseIf{$\Table$ is not output-closed}{
            $\Table \gets \FMakeOutputClosed{\Table}$\label{output-closeM}
        }
    }\label{cohesiveendM}
    $\hypothesisM \gets \FGenerateHypothesis{\Table, \partitioning}$\label{construct-hypothesisM}\tcp*[f]{Predicates are inferred by $\partitioning$}\;
    \Switch{$\FEqQuery(\hypothesisM)$}{
        \lCase{true} {
            \KwReturn{} $\hypothesisM$\label{returnM}
        }
        \lCase{$\cex \in \domain^+$} {
            \KwPush{} $\prefixes(\cex) \setminus (\Prefixes \cup \NextPrefixes)$ \KwTo{} $\NextPrefixes$\label{process-counterexampleM}
        }
    }
}
\end{algorithm}

\subsection{Termination and Complexity of $\LambdaM$}
We prove the termination of $\LambdaM$ using generators and $s_g$-learnability defined in \cref{learnability}.
For the proof, we formulate the behavior of the oracle using the generator to learn the partition on the transitions from each state.
Namely, we assume that by fixing an oracle, a generator for learning each partition is implicitly fixed, and the oracle returns counterexamples following the generator.
Thus, for each equivalence query, the oracle returns a counterexample only using the characters returned by the generators.
Formally, we assume that an equivalence oracle, together with the target Mealy automaton, induces a mapping from its state $\loc_i \in \Loc$ to a generator $g_i$, which determines the generator for each partition.

The following formulates the set of characters for learning the partitions in the target s-MA.
Notice that $\SigmaEf$ captures how the counterexamples returned by the oracle are helpful, as well as the difficulty of the learning itself: $\SigmaEf$ blows up if the counterexample returned by the oracle is not helpful or uses redundant characters when answering equivalence queries.

\begin{definition}
 [$l_g$, $\SigmaEf$]
 For a generator $g$ of $\Algebra$ and $\partitioning$ such that 
 $\Algebra$ is $s_g$-learnable with $\partitioning$,
 we let $l_g\colon \Partitions \to \finpowerset{\domain}$ be such that $l_g(\partition) = \flatten{L^{\partition}_{\infty}}$.
 For an s-MA $\targetM$ with states $\Loc$, we let $\SigmaEf =\bigcup_{q_i \in \Loc}{l_{g_i}({\partition}_i)}$, where ${\partition}_i$ is the partition at state $q_i$ of $\targetM$ and $g_i$ is the generator to learn ${\partition}_i$.
\end{definition}
We formulate the behavior of an oracle as follows.
Our formulation is permissive in the sense that we focus only on $L^{\partition}_{\infty}$ rather than on each $L^{\partition}_{j}$.

\begin{definition}
[Follow]
Let $\targetM$ be the target s-MA, with states $\Loc = \{q_1, \dots, q_n\}$.
An oracle \emph{follows} generators $g_1, \dots, g_n$ at $q_1, \dots, q_n$ if any counterexample $\cex$ returned in equivalence queries satisfies $\cex \in (\SigmaEf)^+$. 
\end{definition}

\LongVersion{The following is our main theorem, termination of the $\LambdaM$ algorithm.}

\newcommand{\terminationStatement}{%
   Let $\targetM = (\Algebra,\Loc,\initLoc,\OUTPUT,\delta) $ be the target s-MA, with $\Loc = \{q_1, q_2, \dots, q_n\}$.
   If the oracle follows $g_1, g_2, \dots, g_n$ at $q_1, q_2, \dots, q_n$ and for any $q_i \in \Loc$, $\Algebra$ is $s_{g_i}$-learnable with $\partitioning$, the $\LambdaM$ algorithm terminates with a finite number of output and equivalence queries.
}
\begin{theorem}
 [Termination]\label{theorem:termination}
 \terminationStatement{}
\end{theorem}
\begin{proof}[Sketch]
Let $T = (\domain, S, R, \Sigma_E, E, f)$ be an observation table. By symbolic compatibility and minimality, we have $|S| \leq |Q|$. Since for any $\loc_i \in \Loc$, $\Algebra$ is $s_{g_i}$-learnable with $\partitioning$, $\SigmaEf$ is finite. Since the oracle follows the generators, we have $\Sigma_E \subseteq \SigmaEf$.
In \cref{alg2}, the loop from \cref{cohesiveM} terminates due to $|S| \leq |\Loc|$ and $\Sigma_E \subseteq \SigmaEf$.
The loop from \cref{mainloopM} terminates because: 
\begin{ienumeration}
    \item each iteration strictly increases either $|S|$ or $|\Sigma_E|$ and 
    \item if $|S| = |Q|$ and $\Sigma_E = \SigmaEf$, the learner can identify the correct s-MA by symbolic compatibility, minimality, and the condition~\ref{item:partitioning-stability}) in \cref{partition}.
\end{ienumeration}
\qed{}
\end{proof}

We show the query complexity of the $\LambdaM$ algorithm with respect to $|\SigmaEf|$.

\newcommand{\queryComplexityStatement}{
Let $\targetM$ be the target s-MA, $n$ be the number of states of $\targetM$, and $m$ be the maximum length of counterexamples returned from the oracle for equivalence queries.
    If the oracle follows $g_1, \dots, g_n$ at $q_1, \dots, q_n$, and for all $q_i \in \Loc$, $\Algebra$ is $s_{g_i}$-learnable with $\partitioning$,
    the number of output and equivalence queries used by the $\LambdaM$ algorithm is bounded by 
    $(|\SigmaEf| + m + 1) \times n^2 + (2 m + |\SigmaEf| + 1) \times |\SigmaEf| \times n + m \times |\SigmaEf|^2$ and 
    $n + |\SigmaEf|$, respectively.
}
\begin{theorem}
 [Query Complexity of $\LambdaM$]
 \label{theorem:complexity}
 \queryComplexityStatement
 \qed{}
 \end{theorem}

\begin{figure}[t]
\centering
\begin{tikzpicture}[node distance=2.4cm,  scale=0.\LongVersion{70}\ShortVersion{60}, on grid, auto, every node/.style={transform shape}]
    \node (q0) [state, initial, initial text = {}] {$q_0$};
    \node (q1) [state, right = of q0] {$q_1$};
    \node (q2) [state, right = of q1] {$q_2$};
    \node (q3) [state, right = of q2] {$q_3$};
    \node (q4) [state, right = of q3] {$q_4$};
    \node (q5) [state, right = of q4] {$q_5$};

    \path [-stealth, thick]
        (q0) edge node {$[0,10)\mid0$} (q1) 
        (q1) edge node {$[0,10)\mid0$} (q2) 
        (q2) edge node {$[0,10)\mid0$} (q3)
        (q3) edge node {$[0,10)\mid0$} (q4)
        (q4) edge node {$[0,10)\mid0$} (q5)
        
        (q0) edge [loop above]  node[align = center] {$[10,20)\mid 1$ \\ $[20,+\infty)\mid 2$}()
        (q1) edge [loop above]  node[align = center] {$[10,20)\mid 1$ \\ $[20,+\infty)\mid 2$}()
        (q2) edge [loop above]  node[align = center] {$[20,+\infty)\mid 2$}()
        (q3) edge [loop above]  node[align = center] {$[10,20)\mid 1$ \\ $[20,+\infty)\mid 2$}()
        (q4) edge [loop above]  node[align = center] {$[10,20)\mid 1$}()
        (q5) edge [loop right]  node[align = center] {$[0,10)\mid 0$ \\ $[10,20) \mid 1$ \\ $[20, +\infty)\mid 2$}()

        (q2) edge [loop below] node{$[10,20)\mid -1$}()
        (q4) edge [loop below] node{$[20,+\infty)\mid -1$}()
    ;
\end{tikzpicture}     
\caption{s-MA illustrating the independent discovery of the state space and $\SigmaEf$.}
\label{eq-example}
\end{figure}
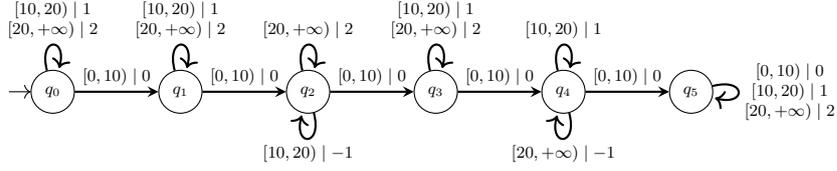
\cref{theorem:complexity} suggests that the number of states of $\targetM$ and the complexity of identifying the predicates in $\targetM$ independently increase the number of equivalence queries in $\LambdaM$.
This independent influence indeed occurs for the oracle returning the lexicographically minimum counterexample for each equivalence query, which is reasonably smart.
Consider s-MAs $\M_{n, k}$ with $2n$ states and $|\SigmaEf| = k$, where 
\begin{ienumeration}
    \item the state space has a sequential shape, 
    \item at the $2i$-th state, $\M_{n, k}$ returns a different output character only for the $i$-th input character in $\SigmaEf$, and
    \item at the initial state, $\M_{n, k}$ returns a different output character for each $a \in \SigmaEf$.
\end{ienumeration}
\cref{eq-example} illustrates one of the instances.
For equivalence queries,
since each $a \in \SigmaEf \setminus \Sigma_E$ is a counterexample, the oracle initially uses each of them as a counterexample;
the oracle then uses words to discover new states as counterexamples.
\cref{proposition:lower_bound} formally states this independent behavior.

\newcommand{\lowerBoundStatement}{%
   Let $\Algebra$ be the interval algebra over naturals (\cref{example:interval-algebra}) and $\partitioning$ be the partitioning function in \cref{algorithm:partitioning_interval_algebra} for $\Algebra$.
   For the oracle returning the lexicographically minimum counterexample for each equivalence query,
   for any $n, k \in \N$ satisfying $n \geq 2$ and $k \geq n$,
   there is an s-MA $\targetM$ with $2 n$ states such that $|\SigmaEf| = k$ and
   $\LambdaM$ requires at least $n + k$
   equivalence queries to learn $\targetM$.
}
\begin{proposition}[Lower Bound]
    \label{proposition:lower_bound}
    \lowerBoundStatement{}
    \qed{}
\end{proposition}

\section{Experiments}

We implemented our algorithm in Java\footnote{Our implementation is distributed under Apache 2.0 License at \url{https://github.com/SoftwareFoundationGroupAtKyotoU/learningsma}}.
More specifically, we implemented 
\begin{ienumeration}
    \item a learner that executes the $\LambdaM$ algorithm and
    \item an oracle that answers the learner's queries based on the target s-MA given as input.
\end{ienumeration}
Our implementation is based on the library~\cite{library}, which implements various algorithms for s-FA.
We implemented the oracle based on the following design decisions:
\begin{ienumeration}
    \item the oracle answers an equivalence query by exhaustively comparing the target and hypothesis s-MAs,
    \item the returned counterexample may vary even if the same hypothesis is given, and
    \item the oracle constructs counterexamples only using the essential characters $\SigmaEf$.
\end{ienumeration}
Throughout this section, we denote the number of states of the target s-MA by $n$, and the maximum length of counterexamples returned by the oracle for equivalence queries by $m$.
\noindent\textbf{Environments.}
All experiments were performed on a server with Intel(R) Xeon(R) CPU E5-2687W v3 @ 3.10GHz, 251GiB RAM running Ubuntu 22.04.4 LTS.
For each example, we executed each experiment for 10 times. %

\subsection{Efficiency for Practical Examples}%
\label{section:experiments:practical}

To observe the efficiency for practical examples,  we took two benchmarks from MATLAB/Simulink demonstrations: \MH{} (Mars Helicopter)~\cite{MH} and \ATGS{} (Automatic Transmission Gear System)~\cite{ATGS}. 
\cref{practical} summarizes them.
\MH{} takes inputs consisting of both Boolean and numeric values.
\ATGS{} takes purely numeric inputs.
We use the equality algebra (\cref{example:equality-algebra}) over $\{0, 1\}$
for Boolean values and the interval algebra (\cref{example:interval-algebra}) for numeric values. We use the product algebra in~\cite{DBLP:conf/tacas/DrewsD17} for their products.
 We report the number of queries, total runtime including the time to answer queries, the size of the final observation table, and the length of the longest counterexamples returned for equivalence queries.
\begin{table}[tb]
\caption{Summary of the benchmarks and the experimental results for RQ1. The columns ``$|\OUTPUT|$'', ``$n$'', and ``$|\delta|$'' show the size of the output alphabet, the number of states, and the number of transitions\LongVersion{, respectively}. The columns ``\# Bool. '' and ``\# Num. '' show the number of Boolean and numeric values in each input character. The columns ``\# of eq.'', ``\# of oq.'', and ``runtime'' show the average number of equivalence and output queries made during learning and the average runtime. The columns ``$|\SigmaEf|$'', ``$|R|$'', and ``$|E|$'' show their size in the final observation table. The column ``$m$'' shows the average of the maximum length of the counterexamples returned for equivalence queries during learning.}\label{practical}
\centering
 \scriptsize
 \begin{tabular}{cccccc|ccccccc}
  \toprule
  {} &  $|\Gamma|$ &  $n$ &  $|\delta|$ &  \# Bool. & \# Num. &  \# of eq. &    \# of oq. &        runtime [s.] & $|\SigmaEf|$ &    $|R|$ & $|E|$ & $m$\\
  \midrule
  \MH &             4 &      5 &          12 &      1 & 3&   36.0 &    6516.0 &     10.9  &         36.0 &    176.0 &   0.0 & 4\\
  \ATGS     &           4 &     16 &          34 &     0 & 2 &   66.0 &   86446.8 &   2481.2  &         68.0 &   1152.2 &   6.0 & 12 \\
  \bottomrule
 \end{tabular}
\end{table}

\Cref{practical} summarizes the experimental results\LongVersion{ for the practical examples}.
In \cref{practical}, we observe that the number of equivalence queries is at most $|\SigmaEf|$.
This is in contrast to the theoretical upper bound ($n + |\SigmaEf|$) in \cref{theorem:complexity} and the lower bound ($n/2 + |\SigmaEf|$) in \cref{proposition:lower_bound} for an extreme case.\footnote{The fact that the numbers of equivalence queries in \cref{practical} are less than $n/2 + |\SigmaEf|$ does not contradict with \cref{proposition:lower_bound}; this proposition asserts the \emph{existence} of an s-MA that requires at least $n/2 + |\SigmaEf|$ equivalence queries.}
This is because $\LambdaM$ usually identifies the state space and the essential characters $\SigmaEf$ simultaneously.
Moreover, for the target s-MA of \MH{}, any pair of states can be distinguished with a suffix of length $1$.
For such an automaton, $\LambdaM$ identifies the state space only using membership queries once
the essential characters $\SigmaEf$ are identified. Thus, the number of equivalence queries was equal to $|\SigmaEf|$ for \MH{}.
In \ATGS{}, there are pairs of states requiring a suffix of length 2 to be distinguished.
Although the number of equivalence queries can be more than $|\SigmaEf|$ for such an automaton,
\LongVersion{\cref{practical} shows that} it was less than $|\SigmaEf|$.
This is likely because multiple essential characters are identified by one equivalence query.

In \cref{practical}, we observe that the number of output queries is also smaller than their theoretical upper bound, which is
14,309 and 168,592, for \MH{} and \ATGS{}, respectively.
In particular, 
$|R|$ and $|E|$ are smaller than their upper bounds.
$|R|$ is smaller than the upper bound because:
\begin{ienumeration}
    \item $\LambdaM$ made fewer equivalence queries than the upper bound and processed fewer counterexamples, and %
\item %
not all counterexamples are of length $m$. 
\end{ienumeration}
$|E|$ is smaller than the upper bound because not all pairs of states require suffixes of length longer than 1 to be distinguished.

Furthermore, we believe that in practice, the number of equivalence queries is usually linear to $\SigmaEf$ and $|R|$ and $|E|$ tend to be smaller than their theoretical upper bounds. 
For instance, \cite{DBLP:conf/icgi/KrugerGV23} reports that
in many practical non-symbolic examples, any pair of states can be distinguished by suffixes of length 1, in which case the number of equivalence queries is at most $|\SigmaEf|$ and $|E| = 0$.
\reviewer{2}{My main critique is on the first experiment (of the two). RQ1 asks whether the algorithm "requires less queries for practical examples than the theoretical worst case", to which the authors answer "Yes" based on two examples. This is way too little, and I see no reason why this should give an answer to such a general question, and how to generalise from two examples. (There is a claim before "we anticipate that in many practical cases, the number of equivalence queries is usually linear to [the size of the needed essential input alphabet] (...)" which is argued with a reference which observes most counterexamples are size 1 in a real benchmark set; it is not obvious to me from this analysis how this affects the current paper, which considers symbolic Mealy machines).}
\mw{Thanks to your comment, we noticed that the present argument and conclusion of RQ1 were stronger than our intention. We wanted to emphasize that (1) the actual scalability can be much better than the worst-case scenario we theoretically derived, and (2) one previously published result also suggests that, in many practical cases, $|E|$ is indeed very small. We will correct the argument from L499 and "Answer to RQ1" in the final version.}
\reviewer{2}{Further, a weakness in the experiments is that the equivalence oracle has knowledge of the target system; as acknowledged in lines 505-509 this is not realistic. Usually, some form of testing is used to approximate it. Perhaps this type of experiment for the type of question about the size of the essential input alphabet; but it does not say much on the performance of the algorithm in a practical black-box scenario, and in particular it is not clear to me how the interpret that 41 minutes were needed to learn one of the examples (repeated in the introduction); in practice, most of the effort in automata learning is in evaluating equivalence queries through testing.}
In \cref{practical}, we observe that \ATGS{} was learned in 41.353 minutes. 
Although our implementation requires an oracle to have a model of a target system as an s-MA to answer equivalence queries, which itself is not a very realistic assumption, 
this computational cost is not excessive.
This suggests that our learning algorithm has the potential to be used in realistic scenarios, at least for learning the behavior of a system with reasonable execution time.

\subsection{Scalability for Random Symbolic Mealy Automata}\label{section:experiments:random}

To observe the experimental scalability of our algorithm,
we used randomly generated benchmarks. We used the interval algebra over naturals in \cref{example:interval-algebra}. The size of the output alphabet is fixed to 3. The values specified in the random generation are the number of states of the automaton and $\SigmaEf$ for the current oracle. The number of states is 10, 20, 40, 80, and $|\SigmaEf|$ is 10, 20, 30, 40; Thus, we have 16 configurations in total. For each configuration, we randomly generated 10 s-MAs.
See \cref{section:random_benchmark_detail} for the details of the random generation.
We measured the number of output and equivalence queries and $|R|$ and $|E|$ of the final observation table. %
\begin{table}[tb]
\caption{Summary of the experimental setting and the results. The columns ``$n$'' and ``$|\SigmaEf|$'' show the number of states and the size of $\SigmaEf$ of the\LongVersion{ random} s-MAs. The columns ``\# of eq.''\ and ``\# of oq.''\ show the average number of equivalence and output queries made during learning. The columns ``$|R|$'' and ``$|E|$'' show the size of $R$ and $E$ in the final observation table\LongVersion{, respectively}.
}\label{rq1:randomsummary}
\scriptsize
\centering
 \begin{tabular}{cc|rrrr}
  \toprule
  $n$ & $|\SigmaEf|$ &  \# of eq. &  \# of oq. &   $|R|$ &  $|E|$ \\
  \midrule
  10 &           10 &      10.00 &    1015.60 &   91.56 &    0.0 \\
  10 &           20 &       20.00 &     4075.60 &  193.78 &    0.0 \\
  10 &           30 &      29.99 &    9104.70 &  293.49 &    0.0 \\
  10 &           40 &       40.00 &    16161.60 &  394.04 &    0.0 \\
  20 &           10 &      10.10 &    2035.95 &  181.57 &    0.1 \\
  20 &           20 &       20.00 &     8030.80 &  381.54 &    0.0 \\
  20 &           30 &      30.00 &   18080.70 &  582.69 &    0.0 \\
  20 &           40 &       40.00 &    32110.80 &  782.77 &    0.0 \\
  \bottomrule
 \end{tabular}
\hfil
\begin{tabular}{cc|rrrr}
\toprule
$n$ & $|\SigmaEf|$ &  \# of eq. &  \# of oq. &   $|R|$ &  $|E|$ \\
\midrule
 40 &           10 &      10.00 &    4010.90 &  361.09 &    0.0 \\
 40 &           20 &       20.00 &    16036.20 &  761.81 &    0.0 \\
 40 &           30 &      30.00 &   36079.20 & 1162.64 &    0.0 \\
 40 &           40 &       40.00 &    64175.60 & 1564.39 &    0.0 \\
 80 &           10 &      10.18 &    8186.59 &  722.55 &    0.2 \\
 80 &           20 &       20.00 &    32031.80 & 1521.59 &    0.0 \\
 80 &           30 &      30.00 &   72072.00 & 2322.40 &    0.0 \\
 80 &           40 &       40.00 &   128100.80 & 3122.52 &    0.0 \\
\bottomrule
\end{tabular}
\end{table}
\begin{figure}[tb]
    \centering
    \includegraphics[width = \LongVersion{5.0}\ShortVersion{4.5}cm]{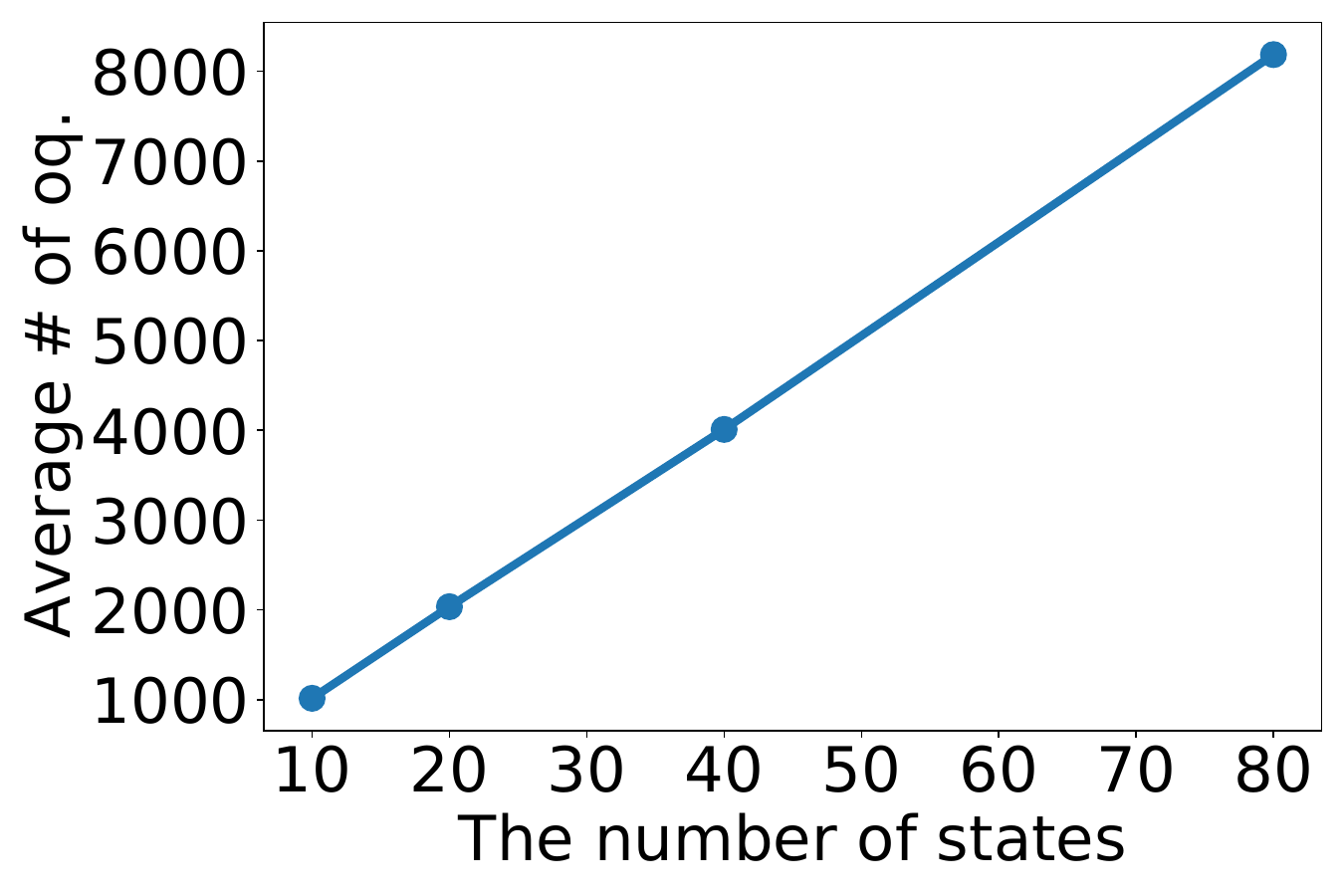}
    \qquad
    \includegraphics[width = \LongVersion{5.0}\ShortVersion{4.5}cm]{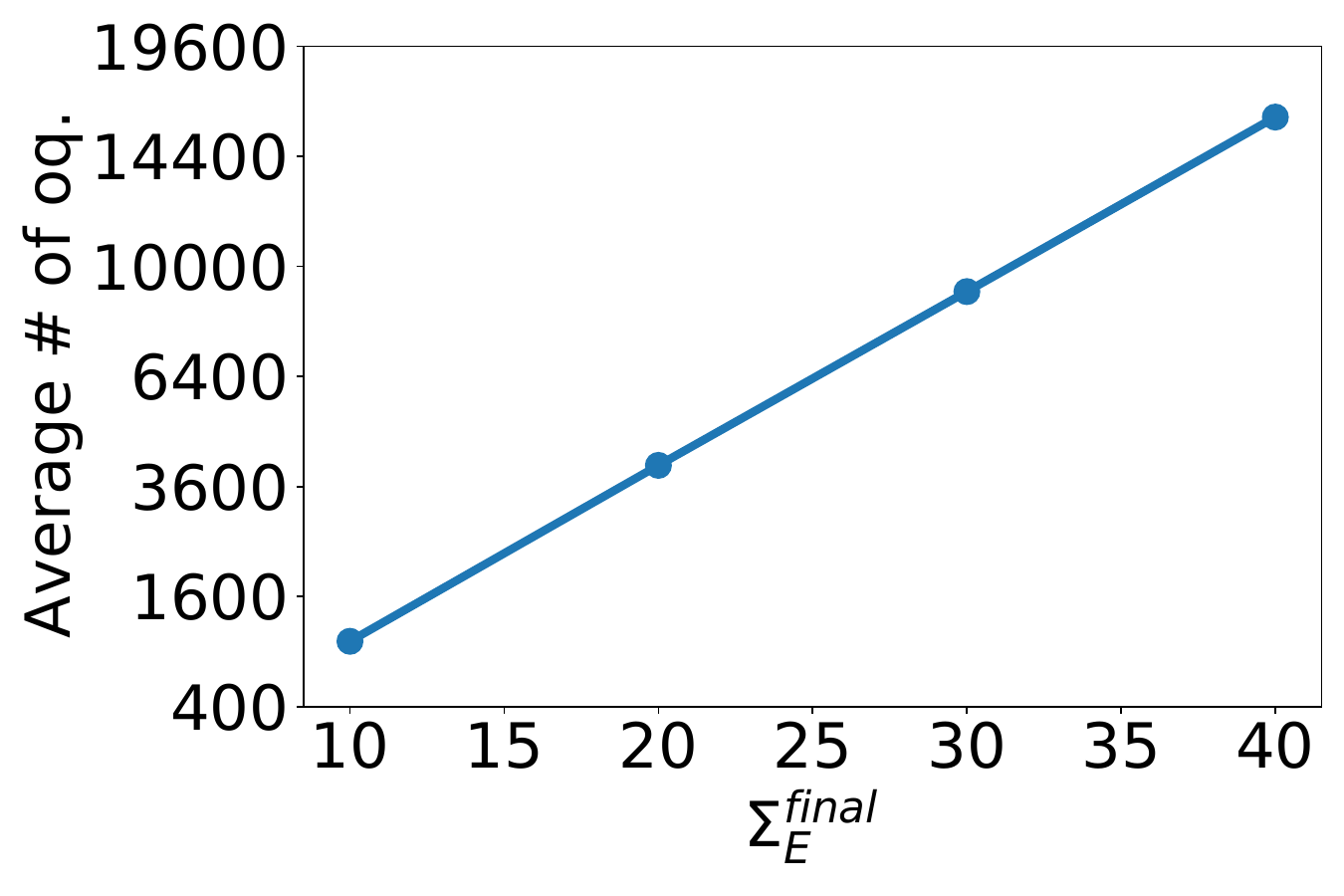}
    \caption{The number of output queries, where the size of $\SigmaEf$ is fixed to 10 (left) and the number of states is fixed to 10 (right). The y-axis of the figure right is on a quadratic scale.}\label{rq1:outputqueries}
\end{figure}

\Cref{rq1:randomsummary} summarizes the experimental results for random s-MAs, and \cref{rq1:outputqueries} plots the number of output queries \wrt{} the generation parameters.
Similarly to the results for RQ1 in \cref{practical}, \cref{rq1:randomsummary} shows that the number of equivalence queries is almost $|\SigmaEf|$. This is for the same reason as discussed in \cref{section:experiments:practical}.
In \cref{rq1:outputqueries}, we observe that the number of output queries is linear to the number of states. This is in contrast to the theoretical upper bound, which is quadratic to the number of states. %
This is because, as shown in \cref{rq1:randomsummary}, $|E|$ is almost 0, and the number of equivalence queries is almost $|\SigmaEf|$. Thus, $|R|$ is almost $m \times |\SigmaEf|$, in which case the $n^2$ term in the upper bound disappears.
\section{Conclusions and Future Work}\label{section:conclusions}
We propose the $\LambdaM$ algorithm for active learning of s-MAs. The central idea is to explicitly maintain and revise the set $\Sigma_E$ of essential inputs during learning. 
In general, the discovery of the state space and the essential inputs can occur independently, and the number of states and the number of essential inputs can independently affect the number of queries.
In contrast, our experimental results show that identification typically occurs simultaneously and that query complexity is typically much better.

 One future direction is to improve the query complexity with the counterexample analysis in~\cite{DBLP:conf/siemens/RivestS93}, as in~\cite{DBLP:conf/fm/ShahbazG09} for Mealy automata.
 Extending recent algorithms for automata learning, \eg{}~\cite{DBLP:conf/cav/ArgyrosD18,DBLP:conf/tacas/VaandragerGRW22}, to the current setting is also future work.
\begin{LongVersionBlock}
Theoretically, investigating variants of \cref{proposition:lower_bound} for other Boolean algebras, or giving a lower bound independent of the choice of Boolean algebras, is an interesting direction.
 Another future direction is to investigate techniques to employ $\LambdaM$ for approximating a black-box system.
 For instance, an efficient equivalence testing algorithm in a black-box setting is a viable future work.
The use of $\LambdaM$ in the context of \emph{black-box checking}~\cite{DBLP:conf/forte/PeledVY99} to test black-box systems, potentially with theoretical justification, is also a future work.
\end{LongVersionBlock}

\subsubsection*{Acknowledgements.}

This work was partially supported by JST PRESTO Grant No.\ JPMJPR22CA, JST CREST Grant No.\ JPMJCR2012, JSPS KAKENHI Grant No.\ 22K17873, JSPS KAKENHI Grant No.\ 25H01113, and JST BOOST Grant No.\ JPMJBY24H8.

\newpage

\bibliographystyle{splncs04}
\bibliography{sample}

\ifdefined\VersionLong%
\newpage
\appendix

\section{Omitted Proofs}

\subsection{Proof of \cref{theo:symbcomp}}
Here, we prove \cref{theo:symbcomp}, the symbolic compatibility of $\hypothesisM$. Firstly, we show the following lemma.
\begin{lemma}\label{closed-consistent}
     Let $T=(\domain, S, R, \Sigma_E, E, f)$ be a cohesive observation table. For all $s_1,s_2 \in S\cup R$, for all $a \in \Sigma_E$ satisfying $s_1 \cdot a$, $s_2 \cdot a \in S \cup R$ if $\row(s_1) = \row(s_2)$ then $\row(s_1\cdot a) = \row(s_2\cdot a)$ and $f(s_1, a) = f(s_2, a)$. 
\end{lemma}

\begin{proof}
Since the observation table $T$ is consistent and evidence-closed, for all $s_1,s_2 \in S\cup R$, for all $a \in \Sigma_E$, if $\row(s_1) = \row(s_2)$, then $\row(s_1\cdot a) = \row(s_2\cdot a)$ holds. Since the observation table $T$ is closed, there exists $s \in S$ such that $\row(s) = \row(s_1\cdot a) = \row(s_2\cdot a)$. Since $\Sigma_E$ is non-empty, for all $s_1,s_2 \in S\cup R$, for all $a \in \Sigma_E$, if $\row(s_1) = \row(s_2)$ then$f(s_1, a) = f(s_2, a)$. 

\end{proof}

We next prove the following lemma. This lemma shows that all words in $S\cup R$ of the observation table correspond to the states in $\hypothesisMe$. Note that this lemma also shows that all states in $\hypothesisMe$ are reachable from the initial state.
\begin{lemma}\label{lemma:reach}
  Let $T=(\domain, S, R, \Sigma_E, E, f)$ be a cohesive observation table. If $\hypothesisMe = (\Sigma_E,\Loc,\initLoc,\OUTPUT,\transe,\outpe)$ is the evidence Mealy automaton constructed from $T$, for all $w \in S\cup R$, $\transe(\initLoc,w) = \row(w)$.
\end{lemma}
\begin{proof}
    We prove this lemma by induction on the length of $w$. If the length of $w$ is $0$, it is clear that $\transe(\initLoc,\epsilon) = \initLoc = \row(\epsilon)$.
    We assume that for all $w \in S\cup R$ of length $k$, $\transe(\initLoc,w) = \row(w)$ holds. Let $w' = w \cdot a \in S\cup R$ be a word of length $k + 1$, where $w \in S\cup R$ is a word of length $k$ and $a \in \Sigma_E$ is a character. It is easily seen that $w \in S\cup R$ and $a \in \Sigma_E$ since $S\cup R$ is prefix-closed and $S,R \subseteq \Sigma_E^*$ holds. Since $T$ is closed, there exists $s \in S$ such that $\row(w) = \row(s)$. We can conclude the following.
    \begin{align*}
  \transe(\initLoc,w') &= \transe(\initLoc,w\cdot a) & (\text{$w' = w \cdot a$})   \\
          &=  \transe(\transe(\initLoc,w),a) & (\text{the definition of $\transe$}) \\
          &= \transe(\row(w),a) & (\text{induction hypothesis}) \\
           &= \transe(\row(s),a) &(\text{$\row(w) = \row(s)$})  \\
          &= \row(s\cdot a) &(\text{by \cref{construction}})\\
           &= \row(w\cdot a) & (\text{by \cref{closed-consistent}})\\
          &= \row(w') & (\text{$w\cdot a = w'$})
\end{align*}

\end{proof}

 Here, we prove \cref{theo:comp}, the evidence compatibility of the evidence automaton $\hypothesisMe$.  This lemma shows that every cell of the observation table corresponds to the output of the evidence Mealy automaton.
\newcommand{\evidencecompatibilityStatement}{
 Let $T=(\domain, S, R, \Sigma_E, E, f)$ be a cohesive observation table. A Mealy automaton $\eMealy$ is \emph{evidence compatible} with $T$ if for all $s \in S\cup R$, for all $e \in \Sigma_E \cup E$, $\eMealy(s\cdot e) = f(s, e)$. For the evidence Mealy automaton $\hypothesisMe$ constructed from $T$, $\hypothesisMe$ is evidence compatible with $T$.}
\begin{lemma}
[Evidence Compatibility]\label{theo:comp}
\evidencecompatibilityStatement{}

\end{lemma}
\begin{proof}
    We prove this lemma by induction on the length of $e \in \Sigma_E \cup E$. If the length of $e$ is $1$, then $e \in \Sigma_E$. Since $T$ is closed, for all $w \in S\cup R$, there exists $s \in S$ such that $\row(w) = \row(s)$.
    \begin{align*}
        \outpe(\transe(\initLoc,w),e) &= \outpe(\row(w),e) & (\text{by \cref{lemma:reach}}) \\
        &= \outpe(\row(s),e) & (\text{$\row(w) = \row(s)$})\\
         &= f(s, e) & (\text{by \cref{construction}}) \\ 
        &= f(w, e) & (\text{by $\row(w) = \row(s)$})
    \end{align*}
   We assume that for all $w \in S\cup R$, for all $e \in \Sigma_E \cup E$ of length $k$,  $\outpe(\transe(\initLoc,w),e) = f(w, e)$ holds. Let $e' = a \cdot e\in \Sigma_E \cup E$ be a word of length $k + 1$, where $a \in \Sigma_E$ is a character and $e \in \Sigma_E \cup E$ is a word of length $k$. It is easily seen that $a \in \Sigma_E$ and $e \in \Sigma_E \cup E$ since $\Sigma_E \cup E$ is suffix-closed and $E \subseteq \Sigma_E^*$ holds. Since $T$ is closed, there exists a word $s \in S$ such that $\row(w) =\row(s)$. We can conclude the following.
      \begin{align*}
        \outpe(\transe(\initLoc,w),e') &= \outpe(\transe(\initLoc,w),a\cdot e) & (\text{$e' = a \cdot e$}) \\
        &= \outpe(\row(w),a\cdot e) & (\text{by \cref{lemma:reach}})\\
        &= \outpe(\row(s),a\cdot e) & (\text{$\row(w) = \row(s)$})\\        
        &= \outpe(\transe(\row(s),a), e) & (\text{the definition of $\outpe$})\\
        &= \outpe(\row(s\cdot a), e) & (\text{by \cref{construction}})\\
        &= \outpe(\transe(\initLoc,s\cdot a), e) & (\text{by \cref{lemma:reach}})\\
        &= f(s\cdot a,  e) & (\text{induction hypothesis})\\
        &= \hypothesisMe(s\cdot a \cdot e) & (\text{the definition of $T$}) \\
        &= f(s,a \cdot e)  & (\text{the definition of $T$})  \\
        &= f(s, e') & (\text{$a\cdot e = e'$})\\
        &= f(w, e') & (\text{by $\row(s) = \row(w)$})
    \end{align*}
    
    \end{proof}

Finally, we prove \cref{theo:symbcomp}, the symbolic compatibility of the s-MA $\hypothesisM$.
\recallResult{theo:symbcomp}{\symboliccompatibilityStatement}
\begin{proof}
    Let $\hypothesisMe$ be the evidence Mealy automaton constructed from $T$. By \cref{partition}, for all $w \in \Sigma_E^+$, $\hypothesisMe(w) = \hypothesisM(w)$. Since $\hypothesisMe$ is evidence compatible with $T$, for all $s \in S\cup R$, for all $e \in \Sigma_E \cup E$, $\hypothesisMe(s\cdot e) = f(s, e)$. Thus for all $s \in S\cup R$, for all $e \in \Sigma_E \cup E$, $\hypothesisM(s\cdot e) = f(s,e)$.

\end{proof}

\subsection{Proof of \cref{theo:symbolicminimality}}
Here, we prove \cref{theo:symbolicminimality}. Firstly, we show the following lemma.
\begin{lemma}\label{symbolic-sufft}
    Let $\M = (\Algebra,\Loc,\initLoc,\OUTPUT,\delta)$ be an s-MA. For all $x\in \domain^*$ and $w\in \domain^+$, $\M(x\cdot w) = \outp(\trans(\initLoc,x),w)$.
\end{lemma}
\begin{proof}
 We prove this lemma by induction on the length of $x$. If the length of $x$ is $0$, it is clear that $\M(w) = \outp(\trans(\initLoc,\epsilon),w)$.
     We assume that for all $x \in \domain^*$ of length $k$, and for all $w \in \domain^+$, $\M(x\cdot w) = \outp(\trans(\initLoc,x),w)$ holds.
     Let $x' = x \cdot a \in \domain^*$ be a word of length $k+1$, where $x \in \domain^*$
     is a word of length $k$ and $a \in \domain$ is a character. We have the following.
     \begin{align*}
 \outp(\trans(\initLoc,x'),w) &= \outp(\trans(\initLoc,x\cdot a),w)  & (\text{$x' = x\cdot a$})\\
          &=  \outp(\trans(\initLoc,x\cdot a),w) & (\text{the definition of $\M$})\\
          &=  \outp(\trans(\trans(\initLoc,x),a),w)& (\text{the definition of $\trans$})\\
          &=  \outp(\trans(\initLoc,x),a\cdot w) & (\text{the definition of $\outp$})\\
          &= \M(x\cdot a \cdot w) & (\text{induction hypothesis})\\
          &= \M(x'\cdot w) & (\text{$x\cdot a = x'$})
\end{align*}

\end{proof}

The following proves \cref{theo:symbolicminimality}.
\recallResult{theo:symbolicminimality}{\symbolicminimalityStatement}
\begin{proof}
    Let $T = (\domain, S, R, \Sigma_E, E, f)$.
    We prove this theorem by contradiction.
    To obtain a contradiction, suppose that there is an s-MA ${\M}' = (\domain,\Loc',\initLoc', \OUTPUT,{\transition}')$ compatible with $T$ such that the number of states of ${\M}'$ is less than that of $\hypothesisM$.
    Since we have $|\Prefixes| = |\Loc| > |\Loc'|$, by the pigeon hole principle,
    there exists $s_1,s_2 \in S$ such that $\row(s_1) \neq \row(s_2)$ but ${\trans}'(\initLoc',s_1) = {\trans}'(\initLoc',s_2)$. Since $\row(s_1) \neq \row(s_2)$ holds, there exists $e \in \Sigma_E \cup E$ such that $f(s_1,e) \neq f(s_2,e)$. Since ${\M}'$ is symbolic compatible with $T$, we have ${\M}'(s_1\cdot e) = f(s_1,e)$ and  ${\M}'(s_2\cdot e) = f(s_2,e)$.
    By $f(s_1,e) \neq f(s_2,e)$, we have ${\M}'(s_1\cdot e) \neq {\M}'(s_2\cdot e)$.
    In contrast,
    by \cref{symbolic-sufft}, we have
    ${\M}'(s_1 \cdot e) = {\outp}'({\trans}'(\initLoc',s_1),e)$ and ${\M}'(s_2 \cdot e) = {\outp}'({\trans}'(\initLoc',s_2),e)$. By ${\trans}'(\initLoc',s_1) = {\trans}'(\initLoc',s_2)$, we have ${\M}'(s_1 \cdot e) = {\M}'(s_2 \cdot e)$. However, this contradicts  ${\M}'(s_1\cdot e) \neq {\M}'(s_2\cdot e)$. 
    Therefore, the s-MA $\hypothesisM$ constructed from $T$ has the minimum number of states among s-MAs compatible with $T$. 
    \end{proof}

\subsection{Preservation Property of \cref{alg2}}

The following shows that \cref{alg2} preserves the requirements for observation tables in \cref{mot} as well as guarantees that all characters appearing in the observation table are in $\Sigma_E$.

\begin{theorem}
    Let $T = (\domain, S, R, \Sigma_E, E, f)$ be an observation table at \cref{construct-hypothesisM} of \cref{alg2}. 
    The following statements hold: \begin{ienumeration}
    \item $S\cup R$ is prefix-closed;
    \item $\Sigma_E \cup E$ is suffix-closed;
    \item $\epsilon \in S$ and $\Sigma_E$ is non-empty;
    \item $S,R,E \subseteq \Sigma_E^*$.
\end{ienumeration} 
\end{theorem}
\begin{proof}

The operations that add words to $S\cup R$ are when adding all prefixes of a counterexample and $\FMakeEvidenceClosed$. Neither does it violate the prefix-closedness of $S\cup R$. Therefore, $S\cup R$ is prefix-closed. 

$\FMakeOutputClosed$ and $\FMakeConsistent$ may add new words to $\Sigma_E \cup E$.
$\FMakeOutputClosed$ adds a character to $\Sigma_E$. Therefore, it does not violate the suffix-closedness of $\Sigma_E \cup E$. $\FMakeConsistent$ adds a word $a \cdot e$ such that $a \in \domain$ and $e \in \Sigma_E \cup E$. Therefore, it does not violate the suffix-closedness of $\Sigma_E \cup E$.

Since $\emptyword$ is added to $S$ in the initialization, $\emptyword \in S$ is obvious. 
It is clear that $\Sigma_E$ is non-empty since the initialization adds a character to $\Sigma_E$.

Since $T$ is output-closed and $S\cup R$ is prefix-closed, $S,R \subseteq \Sigma_E^*$ holds.
The only operation that adds a word to $E$ is $\FMakeConsistent$, which concatenates the last character of a word in $S\cup R$ and a word in $E$ and adds it to $E$. From $S, R \subseteq \Sigma_E^*$, the last character of a word in $S\cup R$ is in $\Sigma_E$. Since $\Sigma_E \cup E$ is suffix-closed, $\Sigma_E \cup E \subseteq \Sigma_E^*$, which proves $E \subseteq \Sigma_E^*$. 

\end{proof}
\subsection{Proof of \cref{theorem:termination}}

Here, we show the proof of \cref{theorem:termination}. We have divided the proof into a sequence of lemmas.
Firstly, we prove that the loop from \cref{cohesiveM} that makes an observation table cohesive terminates.

\begin{lemma}\label{lemma:not-exceed}
     Let $T = (\domain, S, R, \Sigma_E, E, f)$ be an observation table in \cref{alg2}. Let $\targetM = (\Algebra,\Loc,\initLoc,\OUTPUT,\delta) $ be the target s-MA, with $\Loc = \{q_1, q_2, \dots, q_n\}$.
     If the oracle follows $g_1, g_2, \dots, g_n$ at $q_1, q_2, \dots, q_n$, we have $|S| \leq |\Loc|$ and $\Sigma_E \subseteq \SigmaEf$. 
\end{lemma}
\begin{proof}
    We prove $|S| \leq |\Loc|$ by contradiction. To obtain a contradiction, suppose that $|S| > |\Loc|$. By the pigeon hole principle, there exist $s_1 ,s_2 \in S$ such that $\row(s_1) \neq \row(s_2)$ but $\trans(\initLoc, s_1) = \trans(\initLoc,s_2)$. By $\row(s_1) \neq \row(s_2)$, there exists $e \in \Sigma_E \cup E$ such that $f(s_1,e) \neq f(s_2,e)$. In contrast, by $\trans(\initLoc, s_1) = \trans(\initLoc,s_2)$, for all $w \in \domain^+$, $\outp(\trans(\initLoc,s_1),w) = \outp(\trans(\initLoc, s_2),w)$. By symbolic compatibility, $f(s_1,e) = \targetM(s_1\cdot e) = \outp(\trans(\initLoc,s_1),w) = \outp(\trans(\initLoc, s_2),w) = \targetM(s_2\cdot e) = f(s_2,e)$. This contradicts $f(s_1,e) \neq f(s_2,e)$. Therefore, we have $|S| \leq |\Loc|$. It follows immediately that $\Sigma_E \subseteq \SigmaEf$ since the oracle follows $g_1, g_2, \dots, g_n$ at $q_1, q_2, \dots, q_n$.
    
\end{proof}

\begin{lemma}\label{termination-cohesive}
    Let $\targetM = (\Algebra,\Loc,\initLoc,\OUTPUT,\delta) $ be the target s-MA, with $\Loc = \{q_1, q_2, \dots, q_n\}$.
   If for all $q_i \in \Loc$, $\Algebra$ is $s_{g_i}$-learnable with $\partitioning$, and the oracle follows $g_1, g_2, \dots, g_n$ at $q_1, q_2, \dots, q_n$, the loop that makes an observation table cohesive always terminates.
\end{lemma}
\begin{proof}
There are four operations that make the observation table cohesive. For each of them, we show that it terminates. 

    If the observation table is not closed, there is a row in $R$ of the table that is different from all rows in the $S$. The operation $\FMakeClosed$ (\cref{closeM}) adds this row to $S$. This leads to an increase in the number of $S$ increases. From \cref{lemma:not-exceed}, we have $|S| \leq |Q|$. Therefore, it always terminates.

    If the observation table is not consistent, there is a word that distinguishes two equal rows in $S\cup R$. The operation $\FMakeConsistent$ (\cref{make-consisitentM}) adds this word to $E$. This leads to an increase in the number of $S$. From $|S| \leq |Q|$, it always terminates.
    
    The number of applications of the operation $\FMakeOutputClosed$ (\cref{output-closeM}) is less than or equal to the number of characters that appear in counterexamples.  Since for all $q_i \in \Loc$, $\Algebra$ is $s_{g_i}$-learnable with $\partitioning$, the set of characters that appear in counterexamples, \ie{} $\SigmaEf$ is finite. Therefore, it always terminates.
    
    Since for all $q_i \in \Loc$, $\Algebra$ is $s_{g_i}$-learnable with $\partitioning$, $\SigmaEf$ is finite. 
    Since $|S| \leq |Q|$ and $\Sigma_E \subseteq \SigmaEf$, the operation $\FMakeEvidenceClosed$ (\cref{evidence-closeM}) always terminates.
    \end{proof}

    We next prove that the loop from \cref{mainloopM} that constructs a hypothesis always terminates. Firstly, we prove the loop from \cref{mainloopM} increases either $S$ or $\Sigma_E$.
    \begin{lemma}\label{lemma:strictly-increase}
   Let $\targetM = (\Algebra,\Loc,\initLoc,\OUTPUT,\delta) $ be the target s-MA, with $\Loc = \{q_1, q_2, \dots, q_n\}$.
   If for all $q_i \in \Loc$, $\Algebra$ is $s_{g_i}$-learnable with $\partitioning$ and the oracle follows $g_1, g_2, \dots, g_n$ at $q_1, q_2, \dots, q_n$, the loop that constructs a hypothesis strictly increases either the number of states of hypothesis or the number of essential characters. 
    \end{lemma}
    \begin{proof}
    Let $\hypothesisM^e = (\Sigma_E,\Loc,\initLoc,\OUTPUT,\transe,\outpe)$ 
 be an evidence Mealy automaton and  $\hypothesisM = (\Algebra,\Loc,\initLoc,\OUTPUT,\delta)$ be the s-MA constructed from an observation table $T = (\domain, S, R, \Sigma_E, E, f)$. The oracle returns a counterexample $\cex \in \domain^+$ such that $\targetM(\cex) \neq \hypothesisM(\cex)$. Then $\LambdaM$ adds all prefixes of $\cex$ to $R$. Let $T' = (\domain, S', R', \Sigma_E', E', f')$ the observation table after making cohesive. Let  $\hypothesisMeNext = (\Sigma_E',\Loc',\initLoc',\OUTPUT,{\transe}',{\outpe}')$ and $\hypothesisMNext = (\Algebra,\Loc',\initLoc',\OUTPUT,{\delta}')$ be the next evidence Mealy automaton and the s-MA.

    There are two cases for the counterexample $\cex$: the case $\cex \notin \Sigma_E^+$ and the case $\cex \in \Sigma_E^+$.

    Firstly, we consider the case $\cex \notin \Sigma_E^+$. In this case, we show that $(\SigmaEf \cap \Sigma_E) \subsetneq (\SigmaEf \cap \Sigma_E')$. By $\cex = a_1 \ldots a_m \notin \Sigma_E^+$, there exists $i \in \{1,\ldots,m\}$ such that $a_i \notin \Sigma_E$. Since the equivalence oracle follows generators, for any $i \in \{1,\dots,m\}$, we have $a_i \in \SigmaEf$. Since all prefixes of the counterexample $\cex$ are added to $R$, the operation $\FMakeOutputClosed$ adds $a_i$ to $\Sigma_E$. Therefore we have $(\SigmaEf \cap \Sigma_E) \subsetneq (\SigmaEf \cap \Sigma_E')$.

    We now turn to the case $\cex \in \Sigma_E^+$. In this case, we show that the number of states of $\hypothesisMeNext$ is greater than that of $\hypothesisMe$, \ie{} $|S'| > |S|$. We have divided the proof into a sequence of lemmas.
\end{proof}

Firstly, we show the following lemmas.
\begin{lemma}\label{evidence-sufft}
    Let $\M = (\INPUT, \Loc, \initLoc, \OUTPUT, \trans, \outp)$ be a Mealy automaton.
    For all $x\in \INPUT^*$ and $w\in \INPUT^+$, $\M(x\cdot w) = \outp(\trans(\initLoc,x),w)$.
\end{lemma}
\begin{proof}
 We prove this lemma by induction on the length of $x$. If the length of $x$ is $0$, it is clear that $\M(w) = \outp(\trans(\initLoc,\emptyword),w)$. We assume that for all $x \in \INPUT^*$ of length $k$, and for all $w \in \INPUT^+$ $\M(x\cdot w) = \outp(\trans(\initLoc,x),w)$.
     Let $x' = x \cdot a \in \INPUT^*$ be a word of length $k+1$, where $x \in \INPUT^*$ is a word of length $k$ and $a \in \INPUT$ is a character. We have the following.
     \begin{align*}
 \outp(\trans(\initLoc,x'),w) &= \outp(\trans(\initLoc,x\cdot a),w) & (\text{$x'= x \cdot a$})\\
          &=  \outp(\trans(\trans(\initLoc,x),a),w) & (\text{the definition of $\trans$})\\
          &=  \outp(\trans(\initLoc,x),a\cdot w) & (\text{the definition of $\outp$})\\
          &= \M(x\cdot a \cdot w)& (\text{induction hypothesis})\\
          &= \M(x'\cdot w) & (\text{$x\cdot a = x'$})
\end{align*}

\end{proof}
We next prove \cref{evidence-minimality}. This lemma states that the evidence Mealy automaton $\hypothesisMe$ has the minimum number of states among any other Mealy automaton evidence compatible with the observation table.
\begin{lemma}\label{evidence-minimality}
     Let $T=(\domain, S, R, \Sigma_E, E, f)$ be a cohesive observation table and let $\hypothesisMe =(\Sigma_E,\Loc,\initLoc,\OUTPUT,\transe,\outpe)$ be the evidence Mealy automaton constructed from $T$. $\hypothesisMe$ has the minimum number of states among the Mealy automata that are evidence compatible with $T$.
\end{lemma}
\begin{proof}
    We prove this theorem by contradiction. To obtain a contradiction, suppose that there is a Mealy automaton ${\M^e}' = (\Sigma_E,\Loc',\initLoc',\OUTPUT,{\transe}',{\outpe}')$ evidence compatible with $T$ such that the number of states of ${\M^e}'$ is less than that of $\hypothesisMe$. Since we have $|S| = |Q| > |Q'|$, by the pigeon hole principle,   there exists $s_1,s_2 \in S$ such that $\row(s_1) \neq \row(s_2)$ but ${\transe}'(\initLoc',s_1) = {\transe}'(\initLoc',s_2)$. By $\row(s_1) \neq \row(s_2)$, there exists $e \in \Sigma_E \cup E$ such that $f(s_1,e) \neq f(s_2,e)$. Since ${\M^e}'$ is evidence compatible with $T$, we have ${\M^e}'(s_1\cdot e) = f(s_1,e)$ and ${\M^e}'(s_2\cdot e) = f(s_2,e)$. By $f(s_1,e) \neq f(s_2,e)$, ${\M^e}'(s_1\cdot e) \neq {\M^e}'(s_2\cdot e)$. In contrast, by \cref{evidence-sufft}, we have ${\M^e}'(s_1 \cdot e) = {\outpe}'({\transe}'(\initLoc',s_1),e)$ and ${\M^e}'(s_2 \cdot e) = {\outpe}'({\transe}'(\initLoc',s_2),e)$. By ${\transe}'(\initLoc',s_1) = {\transe}'(\initLoc',s_2)$, we have ${\M^e}'(s_1 \cdot e) = {\M^e}'(s_2 \cdot e)$. However, this contradicts  ${\M^e}'(s_1\cdot e) \neq {\M^e}'(s_2\cdot e)$. 
    Therefore, the evidence Mealy automaton $\hypothesisMe$ constructed from $T$ has the minimum number of states among the Mealy automata evidence compatible with $T$. 
    
    \end{proof}
    We next prove the property \emph{reduced}, which states that each $s \in S$ is unique.
    \begin{definition}
        [reduced]
        An observation table $T = (\domain, S, R, \Sigma_E, E, f)$ is reduced if for all $s_1, s_2 \in S$, $s_1 \neq s_2$ implies $\row(s_1) \neq \row(s_2)$.
    \end{definition}
    \begin{lemma} \label{reduced}
     At \cref{construct-hypothesisM} of \cref{alg2}, the cohesive observation table is reduced. 
\end{lemma}
\begin{proof}
Let $T = (\domain, S, R,\Sigma_E, E, f)$ be the observation table.
The only operations that add words to $S$ are when adding $\emptyword$ in initialization and $\FMakeClosed$. The operation $\FMakeClosed$ adds $r \in R$ such that for all $s \in S$, $\row(r) \neq \row(s)$. These do not violate reducedness of the observation table. Therefore, $T$ is always reduced. 
    \end{proof}
We next prove \cref{injective} and \cref{right-toal}. These lemmas state that the states of Mealy automaton ${\M^e}'$ evidence compatible with $T$ and the states of the evidence automaton $\hypothesisMe$ have a one-to-one correspondence. 
\begin{lemma}\label{injective}
Let $T = (\domain, S, R,\Sigma_E, E, f)$ be a cohesive observation table. For any Mealy automaton ${\M^e}' = (\Sigma_E,\Loc',\initLoc',\OUTPUT,{\transe}',{\outpe}')$, if ${\M^e}'$ is evidence compatible with $T$, for all $s_1, s_2 \in S$, if $s_1 \neq s_2$ then ${\transe}'(\initLoc',s_1) \neq {\transe}'(\initLoc', s_2)$.
\end{lemma}
\begin{proof}
    By \cref{reduced}, for all $s_1, s_2 \in S$, we have $\row(s_1) \neq \row(s_2)$.
    Therefore, there exists $e \in \Sigma_E \cup E$ such that $f(s_1,e) \neq f(s_2,e)$. Since ${\M^e}'$ is evidence compatible with $T$, we have $f(s_1,e) = {\M^e}'(s_1\cdot e) = {\outpe}'({\transe}'(\initLoc',s_1), e)$ and $f(s_2,e) = {\M^e}'(s_2\cdot e) = {\outpe}'({\transe}'(\initLoc',s_2), e)$. By $f(s_1,e) \neq f(s_2,e)$, we have ${\transe}'(\initLoc',s_1) \neq {\transe}'(\initLoc', s_2)$. 
\end{proof}

\begin{lemma} \label{right-toal}
 Let $T = (\domain, S, R,\Sigma_E, E, f)$ be  a cohesive observation table, and let $\hypothesisMe = (\Sigma_E,\Loc,\initLoc,\OUTPUT,\transe,\outpe)$ be the evidence Mealy automaton constructed from $T$. For any Mealy automaton ${\M^e}' = (\Sigma_E,\Loc',\initLoc',\OUTPUT,{\transe}',{\outpe}')$, if ${\M^e}'$ is evidence compatible with $T$, and the number of states of $\hypothesisMe$ is equal to that of ${\M^e}'$, for all $q \in \Loc'$, there exists $s \in S$ satisfying $q = {\transe}'(\initLoc, s)$.
\end{lemma}
\begin{proof}
     Since ${\M^e}'$ and $\hypothesisMe$ have the same number of states and \cref{reduced} holds, we have $|S| = |Q| = |Q'|$. In general, an injection between sets with the same number of elements is a surjection.
     \Cref{injective} holds and $|S| = |Q| = |Q'|$ is finite, which proves the lemma. 
\end{proof}
We next prove \cref{preserve-transition} and \cref{preserve-output}. These lemmas state that the transition function and the output function of a Mealy automaton ${\M^e}'$ evidence compatible with $T$ are equal to those of the evidence automaton $\hypothesisMe$.
\begin{lemma}\label{preserve-transition}
    Let $T = (\domain, S, R,\Sigma_E, E, f)$ be a cohesive observation table, and let $\hypothesisMe = (\Sigma_E,\Loc,\initLoc,\OUTPUT,\transe,\outpe)$ be the evidence Mealy automaton constructed from $T$. For any Mealy automaton ${\M^e}' = (\Sigma_E,\Loc',\initLoc',\OUTPUT,{\transe}',{\outpe}')$, if ${\M^e}'$ is evidence compatible with $T$, and the number of states of $\hypothesisMe$ is equal to that of ${\M^e}'$, for all $s_1,s_2,s_3 \in S$, $a \in \Sigma_E$, if $\transe(\transe(\initLoc,s_1), a) = \transe(\initLoc,s_2)$ and ${\transe}'({\transe}'(\initLoc',s_1),a) = {\transe}'(\initLoc',s_3)$ then $s_2 = s_3$.
\end{lemma}
   \begin{proof}
       We prove this lemma by contradiction.
       To obtain a contradiction, suppose that there exist $s_1,s_2,s_3 \in S$, $a \in \Sigma_E$ satisfying $\transe(\transe(\initLoc,s_1), a) = \transe(\initLoc,s_2)$, ${\transe}'({\transe}'(\initLoc',s_1),a) = {\transe}'(\initLoc',s_3)$, and $s_2 \neq s_3$.
       By \cref{reduced}, there exists $e \in \Sigma_E \cup E$ such that $f(s_2,e) \neq f(s_3,e)$. Since ${\M^e}'$ and $\hypothesisMe$ are evidence compatible with $T$, we have 
       $\hypothesisMe(s_2\cdot e) = f(s_2,e)$ and 
       ${\M^e}'(s_3\cdot e) = f(s_3,e)$. Therefore $\hypothesisMe(s_2\cdot e) \neq {\M^e}'(s_3\cdot e)$. In contrast, since $\hypothesisMe$ is compatible with $T$, we have the following.
       \begin{align*}
        \hypothesisMe(s_2\cdot e) &= 
        \outpe(\initLoc,s_2\cdot e) & (\text{the definition of $\hypothesisMe$})\\
        &=\outpe(\transe(\initLoc,s_2),e) & (\text{by \cref{evidence-sufft}})\\
&=\outpe(\transe(\transe(\initLoc,s_1), a),e) & (\text{$ \transe(\initLoc,s_2)= \transe(\transe(\initLoc,s_1), a)$}) \\ 
        &= \outpe(\transe(\initLoc,s_1), a\cdot e) & (\text{the definition of $\outpe$}) \\ 
        &= \outpe(\initLoc,s_1\cdot a\cdot e) & (\text{by \cref{evidence-sufft}}) \\
        &= \hypothesisMe(s_2\cdot e) & (\text{the definition of $\hypothesisMe$})\\
        &= f(s_1 \cdot a ,e) & (\text{evidence compatibility of $\hypothesisMe$})
        \end{align*}
        By evidence-closedness of $T$, $s_1 \cdot a \in S\cup R$. Since $\hypothesisMe$ is evidence compatible with $T$, we have $\hypothesisMe(s_1\cdot a\cdot e) = f(s_1 \cdot a ,e) $. Similarly, we have the following. 
       \begin{align*}
           {\M^e}'(s_3\cdot e) &= 
           {\outpe}'(\transe(\initLoc,s_3),e) & (\text{the definition of ${\M^e}'$})\\ 
           &= {\outpe}'(\transe(\initLoc,s_3),e) & (\text{by \cref{evidence-sufft}})\\ 
           &= {\outpe}'(\transe(\transe(\initLoc,s_1), a),e) & (\text{${\transe}'(\initLoc',s_3) = {\transe}'({\transe}'(\initLoc',s_1),a)$})\\ 
           &= {\outpe}'(\transe(\initLoc,s_1), a\cdot e)&(\text{the definition of $\outpe$})\\ 
           &= {\outpe}'(\initLoc,s_1\cdot a\cdot e) & (\text{by \cref{evidence-sufft}})\\
           &= {\M^e}'(s_1\cdot a\cdot e) & (\text{the definition of ${\M^e}'$})\\
           &= f(s_1 \cdot a, e) & (\text{evidence compatibility of ${\M^e}$})
       \end{align*}
       Therefore, $\hypothesisMe(s_2\cdot e) = {\M^e}'(s_3\cdot e)$.
       
       However, this contradicts $\hypothesisMe(s_2\cdot e) \neq {\M^e}'(s_3\cdot e)$. Therefore $s_2 = s_3$. 
   \end{proof}
\begin{lemma}\label{preserve-output}
    Let $T = (\domain, S, R,\Sigma_E, E, f)$ be a cohesive observation table, and let $\hypothesisMe = (\Sigma_E,\Loc,\initLoc,\OUTPUT,\transe,\outpe)$ be the evidence Mealy automaton constructed from $T$. For any Mealy automaton ${\M^e}' = (\Sigma_E,\Loc',\initLoc',\OUTPUT,{\transe}',{\outpe}')$, if ${\M^e}'$ is evidence compatible with $T$, and the number of states of $\hypothesisMe$ is equal to that of ${\M^e}'$, for all $s \in S$, $a \in \Sigma_E$, $\outpe(\transe(\initLoc,s),a) = {\outpe}'({\transe}'(\initLoc',s),a)$. 
\end{lemma}
\begin{proof}
 We have $\outpe(\transe(\initLoc,s),a) = {\outpe}'({\transe}'(\initLoc',s),a)$ because$\hypothesisMe$ and ${\M^e}'$ are evidence compatible with $T$.
\end{proof}
We finally prove \cref{bigger}. This lemma states that the number of states of $\hypothesisMeNext$ is greater than that of $\hypothesisMe$. We have divided the proof into a sequence of lemmas.
\begin{lemma} \label{table-state-exist}
         Let $T = (\domain, S, R,\Sigma_E, E, f)$ be a cohesive observation table, and let $\hypothesisMe = (\Sigma_E,\Loc,\initLoc,\OUTPUT,\transe,\outpe)$ be the evidence Mealy automaton constructed from $T$. For any Mealy automaton ${\M^e}' = (\Sigma_E,\Loc',\initLoc',\OUTPUT,{\transe}',{\outpe}')$, if ${\M^e}'$ is evidence compatible with $T$, and the number of states of $\hypothesisMe$ is equal to that of ${\M^e}'$, for all $w \in S\cup R$ , there exists $s \in S$ such that $\transe(\initLoc,w) = \transe(\initLoc, s)$ and ${\transe}'(\initLoc',w) = {\transe}'(\initLoc', s)$. 
\end{lemma}
\begin{proof}
     We prove this lemma by induction on the length of $w$. If the length of $w$ is $0$, it is clear that by $\emptyword \in S$,  $\transe(\initLoc,\emptyword) = \transe(\initLoc, \emptyword)$ and ${\transe}'(\initLoc',\emptyword) = {\transe}'(\initLoc', \emptyword)$ hold.
     We assume that for all $w \in S\cup R$ of length $k$, there exists $s \in S$ such that $\transe(\initLoc,w) = \transe(\initLoc, s)$ and ${\transe}'(\initLoc',w) = {\transe}'(\initLoc', s)$ hold.
      Let $w' = w\cdot a \in S\cup R$ be a word of length $k + 1$, where $w \in \Sigma_E^*$ is a word of length $k$ and $a \in \Sigma_E$ is a character.
      Since $S\cup R$ is prefix-closed, $w\in S\cup R$ holds. Then, we have the following.
       \begin{align*}
         \transe(\initLoc,w') 
         & = \transe(\initLoc, w\cdot a) & (\text{$w' = w\cdot a$})\\
         &= \transe(\transe(\initLoc,w), a) & (\text{the definition of $\transe$}) \\
         &= \transe(\transe(\initLoc,s), a) &(\text{induction hypothesis})\\ 
         &=
         \transe(\row(s),a) & (\text{by \cref{lemma:reach}})\\
         &= \row(s\cdot a) &( \text{by \cref{construction}})
         \end{align*}
         By closedness of $T$, there is $s_1\in S$ such that $\row(s \cdot a) = \row(s_1)$. For such $s_1$, by \cref{lemma:reach}, we have $\row(s_1) = \transe(\initLoc,s_1)$. We proceed to show that ${\transe}'(\initLoc', w') = {\transe}'(\initLoc',s_1)$. Similarly, we have the following.
      \begin{align*}
         {\transe}'(\initLoc',w') 
         &= {\transe}'(\initLoc',w\cdot a) & (\text{$w' = w \cdot a$})\\
         &= {\transe}'({\transe}'(\initLoc',w), a) & (\text{the definition of ${\hypothesisMe}'$}) \\
         &= {\transe}'({\transe}'(\initLoc',s), a) & (\text{induction hypothesis})\\ 
         &= {\transe}'(\initLoc',s \cdot a) & (\text{the definition of ${\hypothesisMe}'$})
         \end{align*}
         By \cref{right-toal}, there is $s_2 \in S$ such that 
         ${\transe}'(\initLoc',s \cdot a) = {\transe}'(\initLoc',s_2)$. 
         Since $\transe(\transe(\initLoc,s),a) = \transe(\initLoc,s_1)$ and ${\transe}'({\transe}'(\initLoc',s),a) = {\transe}'(\initLoc',s_2)$ hold, by \cref{preserve-transition}, we have $s_1 = s_2$. This completes the proof. 
\end{proof}
\begin{lemma}\label{word-state-exist}
         Let $T = (\domain, S, R,\Sigma_E, E, f)$ be a cohesive observation table, and let $\hypothesisMe = (\Sigma_E,\Loc,\initLoc,\OUTPUT,\transe,\outpe)$ be an evidence Mealy automaton constructed from $T$. For any Mealy automaton ${\M^e}' = (\Sigma_E,\Loc',\initLoc',\OUTPUT,{\transe}',{\outpe}')$, if ${\M^e}'$ is evidence compatible with $T$, and the number of states of $\hypothesisMe$ is equal to that of ${\M^e}'$, for all $w \in \Sigma_E^*$ , there exists $s \in S$ such that $\transe(\initLoc,w) = \transe(\initLoc, s)$ and ${\transe}'(\initLoc',w) = {\transe}'(\initLoc', s)$.
\end{lemma}
\begin{proof}
     We prove this lemma by induction on the length of $w$. 
     If the length of $w$ is $0$, it is clear since $\emptyword \in S$.
     We assume that for all $w \in \Sigma_E^*$ of length $k$, there exists $s \in S$ such that $\transe(\initLoc,w) = \transe(\initLoc, s)$ and ${\transe}'(\initLoc',w) = {\transe}'(\initLoc', s)$ hold.
      Let $w' = w\cdot a \in \Sigma_E^+$ be a word of length $k + 1$, where $w \in \Sigma_E^*$ is a word of length $k$ and $a \in \Sigma_E$ is a character. We have the following.
     \begin{align*}
         \transe(\initLoc,w') 
         &= \transe(\initLoc,w \cdot a) & (\text{$w' = w \cdot a$})\\
         &= \transe(\transe(\initLoc,w), a) &(\text{the definition of $\transe$}) \\
         &= \transe(\transe(\initLoc,s), a) &(\text{induction hypothesis})\\ 
         &= \transe(\initLoc, s \cdot a) &(\text{the definition of $\transe$})
     \end{align*} 
     Similarly, we have the following.
     \begin{align*}
         {\transe}'(\initLoc',w') 
         &= {\transe}'(\initLoc',w \cdot a) & (\text{$w' = w \cdot a$})\\
         &={\transe}'({\transe}'(\initLoc',w),a) & (\text{the definition of ${\transe}'$})\\ 
         &={\transe}'({\transe}'(\initLoc',s),a) & (\text{induction hypothesis}    ) \\
         &={\transe}'(\initLoc', s \cdot a) & (\text{the definition of ${\transe}'$})
     \end{align*}
     Since $T$ is evidence-closed, $s \cdot a \in S\cup R$ holds. By \cref{table-state-exist}, there exists $s' \in S$ such that $\transe(\initLoc,s\cdot a) = \transe(\initLoc, s')$ and ${\transe}'(\initLoc',s \cdot a) = {\transe}'(\initLoc', s')$, and the proof is complete. 
\end{proof}
\begin{lemma} \label{uniqueness}
     Let $T = (\domain, S, R,\Sigma_E, E, f)$ be a cohesive observation table, and let $\hypothesisMe = (\Sigma_E,\Loc,\initLoc,\OUTPUT,\transe,\outpe)$ be an evidence Mealy automaton constructed from $T$. For any Mealy automaton ${\M^e}' = (\Sigma_E,\Loc',\initLoc',\OUTPUT,{\transe}',{\outpe}')$, if ${\M^e}'$ is evidence compatible with $T$, and the number of states of $\hypothesisMe$ is equal to that of ${\M^e}'$, for all $w \in \Sigma_E^+$, we have $\hypothesisMe(w) = {\M^e}'(w)$.
\end{lemma}
\begin{proof}
    Let $w' = w\cdot a \in \Sigma_E^+$, where $w \in \Sigma_E^*$ and $a \in \Sigma_E$. 
    By \cref{word-state-exist}, there exists $s\in S$ such that $\transe(\initLoc,w) = \transe(\initLoc, s)$ and ${\transe}'(\initLoc',w) = {\transe}'(\initLoc', s)$. 
 By \cref{preserve-output}, $\outpe(\transe(\initLoc,s),a) = {\outpe}'({\transe}'(\initLoc',s),a)$ holds. 
Therefore, we have the following.
     \begin{align*}
        \hypothesisMe(w')
        &=\hypothesisMe(w\cdot a) & (\text{$w' = w\cdot a$})\\
        &= \outpe(\initLoc, w \cdot a) & (\text{the definition of $\hypothesisMe$})\\
        &= \outpe(\transe(\initLoc,w),a) & (\text{by \cref{evidence-sufft}}) \\
        &= \outpe(\transe(\initLoc,s),a) & (\text{by \cref{word-state-exist}}) \\
        &=  {\outpe}'({\transe}'(\initLoc',s),a) & (\text{by \cref{preserve-output}})\\
        &= {\outpe}'({\transe}'(\initLoc',w),a) & (\text{by \cref{word-state-exist}})\\
        &= {\outpe}'(\initLoc',w \cdot a) & (\text{by \cref{evidence-sufft}})\\
        &= {\M^e}'(w \cdot a) & (\text{the definition of ${\M^e}'$})\\
        &= {\M^e}'(w') & (\text{$w\cdot a = w'$})
    \end{align*}

\end{proof}
\begin{lemma}\label{smallest}
    Let $T=(\domain, S, R, \Sigma_E, E, f)$ be a cohesive observation table, and let $\hypothesisMe = (\Sigma_E,\Loc,\initLoc,\OUTPUT,\transe,\outpe)$ be the evidence Mealy automaton constructed from $T$. For any Mealy automaton ${\M^e}' = (\Sigma_E,\Loc',\initLoc',\OUTPUT,{\transe}',{\outpe}')$, if ${\M^e}'$ is evidence compatible with $T$, and there exists $w \in \Sigma_E^+$ such that $\hypothesisMe(w) \neq {\M^e}'(w)$, ${\M^e}'$ has more states than $\hypothesisMe$.
\end{lemma}
\begin{proof}
    We prove this lemma by contradiction. To obtain a contradiction, suppose that there exists ${\M^e}'$ such that ${\M^e}'$ is evidence compatible with $T$, the number of ${\M^e}'$ is equal to that of $\hypothesisMe$, and there exists $w \in \Sigma_E^+$ such that $\hypothesisMe(w) \neq {\M^e}'(w)$. In contrast, by \cref{uniqueness}, 
    for all $w \in \Sigma_E^+$  $\hypothesisMe(w) = {\M^e}'(w)$. This contradicts there exists $w \in \Sigma_E^+$ such that $\hypothesisMe(w) \neq {\M^e}'(w)$. By \cref{evidence-minimality}, 
$\hypothesisMe$ has the minimum number of states among any other Mealy automaton evidence compatible with $T$. Hence ${\M^e}'$ has more states than $\hypothesisMe$. 
\end{proof}
\begin{lemma}\label{bigger}
Let $T = (\domain, S, R, \Sigma_E, E, f)$ be a cohesive observation table, and 
let $\hypothesisM^e = (\Sigma_E,\Loc,\initLoc,\OUTPUT,\transe,\outpe)$ 
 be the evidence Mealy automaton constructed from $T$.
 If the oracle returns a counterexample $\cex \in \Sigma_E^+$ such that $\targetM(\cex) \neq \hypothesisM(\cex)$, let $T' = (\domain, S', R', \Sigma_E', E', f')$ be the observation table after making cohesive, and let  $\hypothesisMeNext = (\Sigma_E',\Loc',\initLoc',\OUTPUT,{\transe}',{\outpe}')$ be the next evidence Mealy automaton, then $\hypothesisMeNext$ has more states than $\hypothesisMe$.
\end{lemma}
\begin{proof}
    It is sufficient to show that $\hypothesisMe(\cex) \neq \hypothesisMeNext(\cex)$. Let $\cex = c'\cdot a$ where $c' \in \Sigma_E^*$ and $a \in \Sigma_E$. Then we have the following.
\begin{align*}
    \hypothesisMeNext(\cex) &= 
    \hypothesisMeNext(c' \cdot a) & (\text{$\cex = c' \cdot a$})\\
    &= f'(c',a) & (\text{evidence compatibility of $\hypothesisMeNext$})\\ &
    = \targetM(c' \cdot a) & (\text{the definition of $T'$})\\
     &
    = \targetM(cex) & (\text{$c' \cdot a = \cex$})\\
    &\neq \hypothesisMe(\cex) & (\text{$\cex$ is the counterexample})
\end{align*}  
Since $\hypothesisMeNext$ is compatible with $T$, by \cref{smallest}, $\hypothesisMeNext$ has more states than $\hypothesisMe$. 
\end{proof}

\begin{lemma}\label{lemma:correctness}
Let $T = (\domain, S, R, \Sigma_E, E, f)$ be a cohesive observation table, and let $\targetM$ be the target s-MA, with $\Loc = \{q_1, q_2, \dots, q_n\}$.
   If the oracle follows $g_1, g_2, \dots, g_n$ at $q_1, q_2, \dots, q_n$ and for any $q_i \in \Loc$, $\Algebra$ is $s_{g_i}$-learnable with $\partitioning$, $|S| = |Q|$, and $\Sigma_E = \SigmaEf$, the $\LambdaM$ algorithm can correctly identify $\targetM$.
\end{lemma}
\begin{proof}
The learned partitions are not broken even if redundant characters are in the lists by condition \ref{item:partitioning-stability}) in \cref{partition}. Therefore, if $\Sigma_E = \SigmaEf$, $\LambdaM$ can identify all the partitions of $\targetM$.

    Let $\hypothesisM = (\domain,\Loc,\initLoc,\OUTPUT,\transition)$ be the s-MA constructed from $T$. We prove this lemma by contradiction. To obtain a contradiction, suppose that $|S| = |Q|$ but $\hypothesisM$ is not equal to $\targetM$. Then, the oracle returns $\cex \in \SigmaEf$, and the next observation table $T' = (\domain, S,' R,' \Sigma_E,' E,' f')$ is constructed. By \cref{bigger}, the next hypothesis has more states than $\hypothesisM$, \ie{} $|S'| > |Q|$. This contradicts \cref{lemma:not-exceed}. 
    
\end{proof}

The following proves \cref{theorem:termination}.

\recallResult{theorem:termination}{\terminationStatement}
\begin{proof}
     By \cref{termination-cohesive}, the loop from \cref{cohesiveM} will terminate. By \cref{lemma:strictly-increase} and \cref{lemma:correctness}, the loop from \cref{mainloopM} will terminate, and $\LambdaM$ correctly identifies $\targetM$.

\end{proof}

\subsection{Proof of \cref{theorem:complexity}}

Here, we show the proof of \cref{theorem:complexity}.
\recallResult{theorem:complexity}{\queryComplexityStatement}

\begin{proof}
    Each equivalence query strictly increases $\Sigma_E$ or $\Prefixes$;
    the latter is by adding a new prefix $\prefix$ to $\NextPrefixes$ violating closedness.
    Thus, the number of equivalence queries is less than $n + |\SigmaEf|$.
    The number of output queries is bounded by the maximum size of the observation table.
    \LongVersion{By \Cref{reduced},}\ShortVersion{Since} for any $\prefix, \prefix' \in \Prefixes$,
    $\prefix \neq \prefix'$ implies $\row(\prefix) \neq \row(\prefix')$\LongVersion{. }%
    \LongVersion{Therefore}, we have $|\Prefixes| \leq n$.
    In $\LambdaM$,
    only $\FMakeConsistent$ increases $\Suffixes$.
    Since $\FMakeConsistent$ is executed when inconsistent prefixes are found,
    the number of its execution is at most $n - 1$.
    Since each execution of $\FMakeConsistent$ increases $\Suffixes$ by one,
    we have $|\Suffixes| < n$.
    $\NextPrefixes$ can be increased when handling a counterexample from an equivalence query or by $\FMakeEvidenceClosed$.
    Since the number of equivalence queries is less than $n + |\SigmaEf|$, and
    $\NextPrefixes$ can increase at most by $m$ for each equivalence query,
    counterexamples can increase $\NextPrefixes$ at most by  $m \times (n + |\SigmaEf|)$ in total.
    By $\FMakeEvidenceClosed$, $\NextPrefixes$ can increase at most by $|\Prefixes| \times |\SigmaEf|$, which is bounded by $n \times |\SigmaEf|$.
    Thus, we have $|\NextPrefixes| \leq m \times (n + |\SigmaEf|) + n \times |\SigmaEf|$.
    Overall, the number of output queries is bounded by $(|\Prefixes| + |\NextPrefixes|) \times (|\SigmaEf| + |\Suffixes|) < (n + m \times (n + |\SigmaEf|) + n \times |\SigmaEf|) \times (|\SigmaEf| + n ) = (|\SigmaEf| + m + 1) \times n^2 + (2 m + |\SigmaEf| + 1) \times |\SigmaEf| \times n + m |\SigmaEf|^2$.
\end{proof}

\subsection{Proof of \cref{proposition:lower_bound}}
Here, we show the proof of \cref{proposition:lower_bound}. Our construction in the proof is inspired by that of~\cite[Theorem~3]{DBLP:conf/icgi/KrugerGV23}.
\recallResult{proposition:lower_bound}{\lowerBoundStatement}
\begin{proof}
Let $n \geq 2$ and $k \geq n$ be natural numbers. We consider the s-MA $\M_{n,k} = (\Algebra,\Loc,\initLoc,\OUTPUT,\transition)$, where $\Algebra$ is the interval algebra over naturals (\cref{example:interval-algebra}), $\Loc = \{q_0,q_1,q_2,\dots, q_{2n-1}\}$, $\initLoc = q_0$, $\OUTPUT = \{-1,0,1,\dots,k-1\}$, $\transition = \transition^{\text{spine}} \cup \transition^{\text{loop-even}} \cup \transition^{\text{loop-odd}}\cup \transition^{\text{end}}$, and $\transition^{\text{spine}}$, $\transition^{\text{loop-even}}$, $\transition^{\text{loop-odd}}$, and $\transition^{\text{end}}$ are as follows:

\begin{itemize}
\item $\transition^{\text{spine}} = \{ (q_{i}, [0,10), q_{i+1},0) \mid 0 \leq i < 2n - 1\}$

\item $\transition^{\text{loop-even}} = \{ (q_{2m}, [10l, 10l + 10), q_{2m},l) \mid 0 \leq m < n, 1\leq l < k -1, l \neq m\}
\cup 
\{(q_{2m}, [10(k-1), +\infty), q_{2m}, k-1) \mid 0 \leq m < n, m \neq k-1\} 
\cup 
\{(q_{2m}, [10m, 10m + 10), q_{2m}, -1) \mid 0 \leq m < n, m \neq k -1\} 
\cup 
\{(q_{2(k-1)}, [10(k-1), +\infty),q_{2(k-1)},-1) \mid k = n\}$

\item $\transition^{\text{loop-odd}} = \{ (q_{2m + 1},[10l,10l + 10), q_{2m + 1},l) \mid 0 \leq m < n, 1 \leq l < k -1\} \cup \{ (q_{2m + 1}, [10(k-1), +\infty), q_{2m+1}, k-1)\mid 0 \leq m < n\}$

\item $\transition^{\text{end}} = \{(q_{2n -1}, [10l, 10l + 10), q_{2n -1}, l) \mid 0 \leq l < k - 1\}  \cup \{(q_{2n -1}, [10(k-1), +\infty), q_{2n-1}, k-1) \}$
\end{itemize}

\Cref{eq-example} shows $\M_{3, 3}$, \ie{} $\M_{n,k}$ %
with $n = 3$ and $k = 3$.

We prove that the learner learns $\M_{n,k}$ using at least $n + k$ equivalence queries, where $|\SigmaEf| = k$. By assumption, the oracle returns lexicographically minimal counterexamples. 
For simplicity, we assume that 
the first character chosen by the learner be $0$
\footnote{The number of equivalence queries is the same if the first character is one of $0, 10, \dots, 10 (k-1)$. Otherwise, the learner requires another equivalence query to identify all of them.}.
Firstly, the learner learns the partition and the corresponding output at the initial state $q_0$. The oracle returns $10,20,30,\ldots,10(k -1)$ as a counterexample. This learning corresponds to the learning of the partition, and the number of equivalence queries corresponds to $k-1 = |\SigmaEf| - 1 $. \Cref{one-state} shows the observation table and the hypothesis after learning the partition at the initial state.
\begin{figure}[t]
    \centering
    \quad
    \begin{minipage}{.4\linewidth}
  \begin{tabular}{l|ccccc|}
      &0& 10 & 20 & $\dots$ & $10(k-1)$  \\ \hline
    $\epsilon$  & 0 & 1 & 2 & $\dots$ & $k-1$ \\ \hline
    0 & 0 & 1 & 2 & $\dots$ & $k-1$  \\
    10 & 0 & 1 & 2 & $\dots$ & $k-1$ \\
    20 & 0 & 1 & 2 & $\dots$ & $k-1$\\
    $\vdots$ & $\vdots$ & $\vdots$ & $\vdots$ & $\dots$ & $\vdots$ \\
    $10(k-1)$ & 0 & 1 & 2 & $\dots$ & $k-1$ \\
  \end{tabular}
  \end{minipage}
  \hfill`
      \begin{minipage}{.3\linewidth}
   \begin{tikzpicture} [node distance = 2.5cm, on grid, auto]
 
\node (q0) [state, initial, initial text = {}] {$q_0$};
 
\path [-stealth, thick]
   (q0) edge [loop above]  node[align = center] {$[0,10)\mid 0$ \\ $[10,20)\mid 1$ \\ $\vdots$ \\ $[10(k-1), +\infty) \mid k-1$}();
\end{tikzpicture}
  \end{minipage}\quad
     \caption{The observation table and the hypothesis s-MA just after identifying all $\SigmaEf$.}\label{one-state}
\end{figure}

The learner then proceeds to learn the states of $\M_{n,k}$. As a counterexample to distinguish the states, $0\cdot 0\cdot 10$ is first returned as a counterexample, and the states $q_1$ and $q_2$ appear in the hypothesis. Next, $0\cdot 0\cdot 0\cdot 0\cdot 20$ is returned as a counterexample, and the states $q_3$ and $q_4$ appear in the hypothesis. Similarly, equivalence queries are made $n - 1$ times until states $q_{2n-3}$ and $q_{2n-2}$ appear in the hypothesis. Finally, the learner asks an equivalence query to make state $q_{2n-1}$ appear. \Cref{2n1-state} shows the final observation table. Therefore, $n$ queries are required to learn the states of $\M_{n,k}$. Finally, the learner asks an equivalence query that the answer is true. Thus, the overall number of equivalence queries is $(k - 1) + n + 1 = n + k$.
\begin{figure}[t]
  \centering
  \scalebox{.8}{
  \begin{tabular}{l|ccccccc|ccccc}
      & 0 & 10 & 20 & $\dots$ & $10(n-1)$ & $\dots$ & $10(k-1)$ & $0 \cdot 10$ & $0\cdot 20$ & $\dots$ & $0\cdot 10(n-1)$ & $0\cdot 0\cdot 10$  \\ \hline
    $\epsilon$  & 0 & 1 & 2 & $\dots$ & $n-1$ & $\dots$ & $k-1$ & 1 & 2 & $\dots$ & $n-1$ & -1\\ 
    0 & 0 & 1 & 2 & $\dots$ & $n-1$ & $\dots$  & $k-1$ & -1 & 2 & $\dots$ & $n-1$ & 1\\ 
    $0\cdot 0$ & 0 & -1 & 2 & $\dots$ & $n-1$ & $\dots$ & $k-1$ & 1 & 2 & $\dots$ & $n-1$ & 1\\
    $0\cdot 0\cdot 0$ & 0 & 1 & 2 & $\dots$ & $n-1$ & $\dots$  & $k-1$ & 1 & -1 & $\dots$ & $n-1$ & 1\\ 
    $0\cdot 0\cdot 0\cdot 0$ & 0 & 1 & -1 & $\dots$ & $n-1$ & $\dots$  & $k-1$ & 1 & 2 & $\dots$ & $n-1$ & 1\\
    \qquad$\vdots$ & $\vdots$ &$\vdots$ &$\vdots$ &$\vdots$ &$\vdots$ &$\vdots$ &$\vdots$ &$\vdots$ &$\vdots$ &$\vdots$ &$\vdots$ &$\vdots$ \\
     $\underbrace{0\cdot 0\cdot \cdots \cdot 0}_{2n-3\text{ times}}$ & 0 & 1 & 2 & $\dots$ & $n-1$ & $\dots$  & $k-1$ & 1 & 2 & $\dots$ & $-1$ & $1$\\ 
    $\underbrace{0\cdot 0\cdot \cdots \cdot 0 \cdot 0}_{2n-2\text{ times}}$ & 0 & 1 & 2 & $\dots$ & $-1$ & $\dots$  & $k-1$ & 1 & 2 & $\dots$ & $n-1$ & $1$\\ 
    $\underbrace{0\cdot 0\cdot \cdots \cdot 0 \cdot 0\cdot 0}_{2n -1\text{ times}}$ & 0 & 1 & 2 & $\dots$& $n-1$ & $\dots$   & $k-1$ & 1 & 2 & $\dots$ & $n-1$ & 1\\ 
    \hline
     \qquad$\vdots$ & $\vdots$ &$\vdots$ &$\vdots$ &$\vdots$ &$\vdots$ &$\vdots$ &$\vdots$ &$\vdots$ &$\vdots$ &$\vdots$ & $\vdots$ & $\vdots$
  \end{tabular}}
  \caption{The final observation table.}\label{2n1-state}
\end{figure}

\end{proof}

\section{A Worked Example}
We consider the target s-MA $\targetM$ shown in \cref{ex-target}.
We use the interval algebra over naturals (\cref{example:interval-algebra}) and the partitioning function in \cref{algorithm:partitioning_interval_algebra}.

First, the initial observation table $T_1$ in \cref{fig:initial-round} is created. $T_1$ is cohesive, then the first hypothesis evidence Mealy automaton and hypothesis s-MA in \cref{fig:initial-round} are constructed. 

The first hypothesis is not correct. Thus, the oracle returns a counterexample $20$. $20$ is added to $R$. $T_2$ in \cref{fig:second-round} is the observation table after the counterexample is added. $T_2$ is not output-closed since $20 \notin \Sigma_E$. \FMakeOutputClosed is applied to $T_2$, so $20$ is added to $\Sigma_E$. $T_3$ is the observation table after \FMakeOutputClosed is applied. $T_3$ is cohesive, then the second hypothesis evidence Mealy automaton and hypothesis s-MA in \cref{fig:second-round} is constructed.

The second hypothesis is not correct. Thus, the oracle returns a counterexample $0\cdot 0 \cdot 0$. Then, $0 \cdot 0$ and $0 \cdot 0 \cdot 0$ are added to $R$. $T_4$ in \cref{fig:third-round} is the observation table after the counterexample is added. $T_4$ is not consistent since $\row(\emptyword) = \row(0)$ but $\row(0) \neq \row(0\cdot 0)$. \FMakeConsistent is applied to $T_4$, then $0 \cdot 0$ is added to $E$. $T_5$ in \cref{fig:third-round} is the observation table after \FMakeConsistent is applied. $T_5$ is not closed since there is no $s \in S$ such that $ \row(0) = \row(s)$, $ \row(20) = \row(s)$, and $ \row(0 \cdot 0) = \row(s)$. \FMakeClosed is applied to $T_5$, then $0, 20, 0\cdot 0$ is moved to $S$ from $R$. $T_6$ in \cref{fig:third-round} is the observation table after \FMakeClosed is applied. $T_6$ in \cref{fig:third-round} is not evidence-closed. Then, \FMakeEvidenceClosed is applied to $T_6$. $T_7$ is the observation table after \FMakeEvidenceClosed is applied. $T_7$ is cohesive, then the third hypothesis evidence Mealy automaton and hypothesis s-MA in \cref{fig:third-round} is constructed. 

The third hypothesis is not correct. Thus, the oracle returns a counterexample $0 \cdot 0 \cdot 10 \cdot 0$. Then, $0\cdot 0 \cdot 10$ and $0 \cdot 0 \cdot 10 \cdot 0$ are added to $R$. $T_8$ in \cref{fig:final-round} is the observation table after the counterexample is added. $T_8$ is not output-closed since $10 \notin \Sigma_E$. \FMakeOutputClosed is applied to $T_8$, then $10$ is added to $\Sigma_E$. $T_9$ in \cref{fig:final-round} is the observation table after \FMakeOutputClosed is applied. $T_9$ is not evidence-closed. Then, \FMakeEvidenceClosed is applied to $T_9$. $T_{10}$ in \cref{fig:final-round} is the observation table after \FMakeEvidenceClosed is applied. $T_{10}$ is cohesive, then the fourth hypothesis evidence Mealy automaton and hypothesis s-MA in \cref{fig:final-round} is constructed. This hypothesis is equal to $\targetM$. Therefore, the $\LambdaM$ algorithm returns it and terminates.
\begin{figure}[tbp]
\begin{tikzpicture} [node distance = 3cm, on grid, auto]
 
\node (q0) [state, initial, initial text = {}] {$q_0$};
\node (q1) [state, right = of q0] {$q_1$};
\node (q2) [state, below = of q1] {$q_2$};
\node (q3) [state, below = of q0] {$q_3$};
 
\path [-stealth, thick]
    (q0) edge  []node {$[0,20)\mid \text{S}$} (q1) 
    (q0) edge [loop above]  node {$[20,\infty)\mid \text{B}$}()
    (q1) edge [loop above] node {$[20,\infty)\mid \text{B}$}()
    (q1) edge  [bend left ]node {$[0,20)\mid \text{S}$} (q2)
    (q2) edge  [bend left ]node {$[10,\infty)\mid \text{P}$} (q1)
    (q2) edge []  node {$[0,10)\mid \text{P}$}(q3)
    (q3) edge  []node {$[0,\infty)\mid \text{P}$}   (q0);
\end{tikzpicture}
\centering
\caption{The target s-MA $\targetM$}
\label{ex-target}
\end{figure}

\begin{figure}[tbp]
\begin{minipage}{.3\linewidth}
\begin{tabular}{l|c}
      $T_1$& 0  \\\hline
      $\epsilon$ & S  \\\hline
      0 & S  \\
    \end{tabular}
\end{minipage}
\begin{minipage}{.3\linewidth}
\begin{tikzpicture} [node distance = 2.5cm, on grid, auto]
 
\node (q0) [state, initial, initial text = {}] {$q_0$};
 
\path [-stealth, thick]
   (q0) edge [loop above]  node[align = center] {$0\mid \text{S}$}();
\end{tikzpicture} 
\end{minipage}
\begin{minipage}{.3\linewidth}
\begin{tikzpicture} [node distance = 2.5cm, on grid, auto]
 
\node (q0) [state, initial, initial text = {}] {$q_0$};
 
\path [-stealth, thick]
   (q0) edge [loop above]  node[align = center] {$[0,\infty)\mid \text{S}$}();
\end{tikzpicture} 
\end{minipage}
\centering
\caption{The initial observation table, the first hypothesis evidence Mealy automaton, and the first hypothesis s-MA.}
\label{fig:initial-round}
\end{figure}
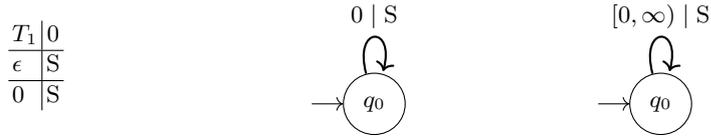
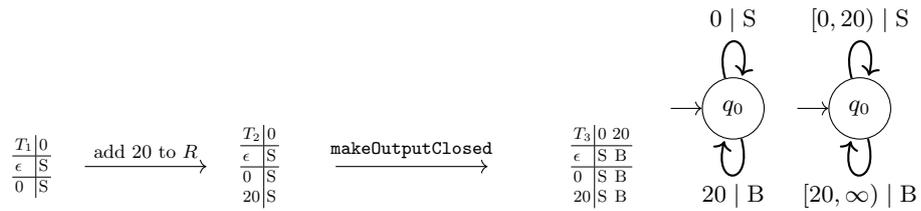
\begin{figure}[tbp]
\begin{tikzpicture}
  \node at (0, 0) {
  \scalebox{.7}{
   \begin{tabular}{l|c}
      $T_1$& 0  \\\hline
      $\epsilon$ & S  \\\hline
      0 & S  \\
    \end{tabular}
    }
  };

  \draw[->] (0.7,0) -- (2.3, 0) node[midway, above] {\scalebox{.8}{add 20 to $R$}};

  \node at (3, 0) {
  \scalebox{.7}{
  \begin{tabular}{l|c}
      $T_2$& 0  \\\hline
      $\epsilon$ & S  \\\hline
      0 & S \\
      20 & S  \\
    \end{tabular}
    }
  };
  \draw[->] (4,0) -- (6, 0) node[midway, above] {\scalebox{.8}{$\FMakeOutputClosed$}};

  \node at (7.5, 0) {
  \scalebox{.7}{
  \begin{tabular}{l|cc}
      $T_3$& 0 & 20 \\\hline
      $\epsilon$ & S & B \\\hline
      0 & S & B\\
      20 & S & B \\
    \end{tabular}
    }
  };
  
\end{tikzpicture}
\begin{tikzpicture} [node distance = 2.5cm, on grid, auto]
 
\node (q0) [state, initial, initial text = {}] {$q_0$};
 
\path [-stealth, thick]
   (q0) edge [loop above]  node[align = center] {$0\mid \text{S}$}()
   (q0) edge [loop below]  node[align = center] {$20\mid \text{B}$}();
\end{tikzpicture} 
\begin{tikzpicture} [node distance = 2.5cm, on grid, auto]
 
\node (q0) [state, initial, initial text = {}] {$q_0$};
 
\path [-stealth, thick]
   (q0) edge [loop above]  node[align = center] {$[0,20)\mid \text{S}$}()
   (q0) edge [loop below]  node[align = center] {$[20,\infty)\mid \text{B}$}();
\end{tikzpicture} 
\centering
\caption{The observation tables in the second loop and the second hypothesis evidence and s-MA.}
\label{fig:second-round}
\end{figure}

\begin{figure}[tbp]
\begin{tikzpicture}
  \node at (0, 0) {
  \scalebox{.7}{
   \begin{tabular}{l|cc}
      $T_3$& 0 & 20 \\\hline
      $\epsilon$ & S & B \\\hline
      0 & S & B\\
      20 & S & B \\
    \end{tabular}
    }
  };

  \draw[->] (1.5,0) -- (3, 0) node[midway, above,align = center] {add $0\cdot 0$, \\$0\cdot 0 \cdot 0$ to $R$};

  \node at (5, 0) {
  \scalebox{.7}{
   \begin{tabular}{l|cc}
      $T_4$& 0 & 20 \\\hline
      $\epsilon$ & S & B \\\hline
      0 & S & B\\
      20 & S & B \\
      $0\cdot 0$ & P & P \\
      $0\cdot 0 \cdot 0$ & P & P \\
    \end{tabular}
    }
  };
  \draw[->] (7,0) -- (9, 0) node[midway, above] {$\FMakeConsistent$};

  \node at (11, 0) {
  \scalebox{.7}{
  \begin{tabular}{l|cc|c}
      $T_5$& 0 & 20 & $0\cdot 0$\\\hline
      $\epsilon$ & S & B & S\\\hline
      0 & S & B & P\\
      20 & S & B & S \\
      $0\cdot 0$ & P & P & P \\
      $0\cdot 0 \cdot 0$ & P & P & S \\
    \end{tabular}
    }
  };
  
\end{tikzpicture}

\begin{tikzpicture}
  \draw[->] (0,0) -- (2, 0) node[midway, above] {$\FMakeClosed$};

  \node at (3.8, 0) {
  \scalebox{.7}{
  \begin{tabular}{l|cc|c}
      $T_6$& 0 & 20 & $0\cdot 0$\\\hline
      $\epsilon$ & S & B & S\\
      0 & S & B & P\\
      $0\cdot 0$ & P & P & P \\
      $0\cdot 0 \cdot 0$ & P & P & S \\\hline
      20 & S & B & S \\
    \end{tabular}
    }
  };
  \draw[->] (6,0) -- (9, 0) node[midway, above] {$\FMakeEvidenceClosed$};

  \node at (11.5, 0) {
  \scalebox{.7}{
  \begin{tabular}{l|cc|c}
      $T_7$& 0 & 20 & $0\cdot 0$\\\hline
      $\epsilon$ & S & B & S\\
      0 & S & B & P\\
      $0\cdot 0$ & P & P & P \\
      $0\cdot 0 \cdot 0$ & P & P & S \\\hline
      20 & S & B & S \\
      $0\cdot 20$ & S & B & P \\
      $0\cdot 0 \cdot 20$ & S & B & P \\
      $0\cdot 0 \cdot 0\cdot 0$ & S & B & S \\
      $0\cdot 0 \cdot 0 \cdot 20$ & S & B & S \\
    \end{tabular}
    }
  };
\end{tikzpicture}

\begin{tikzpicture} [node distance = 3cm, on grid, auto]
 
\node (q0) [state, initial, initial text = {}] {$q_0$};
\node (q1) [state, right = of q0] {$q_1$};
\node (q2) [state, below = of q1] {$q_2$};
\node (q3) [state, below = of q0] {$q_3$};
 
\path [-stealth, thick]
    (q0) edge  []node {$0\mid \text{S}$} (q1) 
    (q0) edge [loop above]  node {$20\mid \text{B}$}()
    (q1) edge [loop above] node {$20\mid \text{B}$}()
    (q1) edge  [bend left ]node {$0\mid \text{S}$} (q2)
    (q2) edge  [bend left ]node {$20\mid \text{P}$} (q1)
    (q2) edge []  node {$0\mid \text{P}$}(q3)
    (q3) edge  []node {$0,20\mid \text{P}$}   (q0);
\end{tikzpicture}
\begin{tikzpicture} [node distance = 3cm, on grid, auto]
 
\node (q0) [state, initial, initial text = {}] {$q_0$};
\node (q1) [state, right = of q0] {$q_1$};
\node (q2) [state, below = of q1] {$q_2$};
\node (q3) [state, below = of q0] {$q_3$};
 
\path [-stealth, thick]
    (q0) edge  []node {$[0,20)\mid \text{S}$} (q1) 
    (q0) edge [loop above]  node {$[20,\infty)\mid \text{B}$}()
    (q1) edge [loop above] node {$[20,\infty)\mid \text{B}$}()
    (q1) edge  [bend left ]node {$[0,20)\mid \text{S}$} (q2)
    (q2) edge  [bend left ]node {$[20,\infty)\mid \text{P}$} (q1)
    (q2) edge []  node {$[0,20)\mid \text{P}$}(q3)
    (q3) edge  []node {$[0,\infty)\mid \text{P}$}   (q0);
\end{tikzpicture}
\caption{The observation tables in the third loop and the third hypothesis evidence and s-MA.}
\label{fig:third-round}
\end{figure}
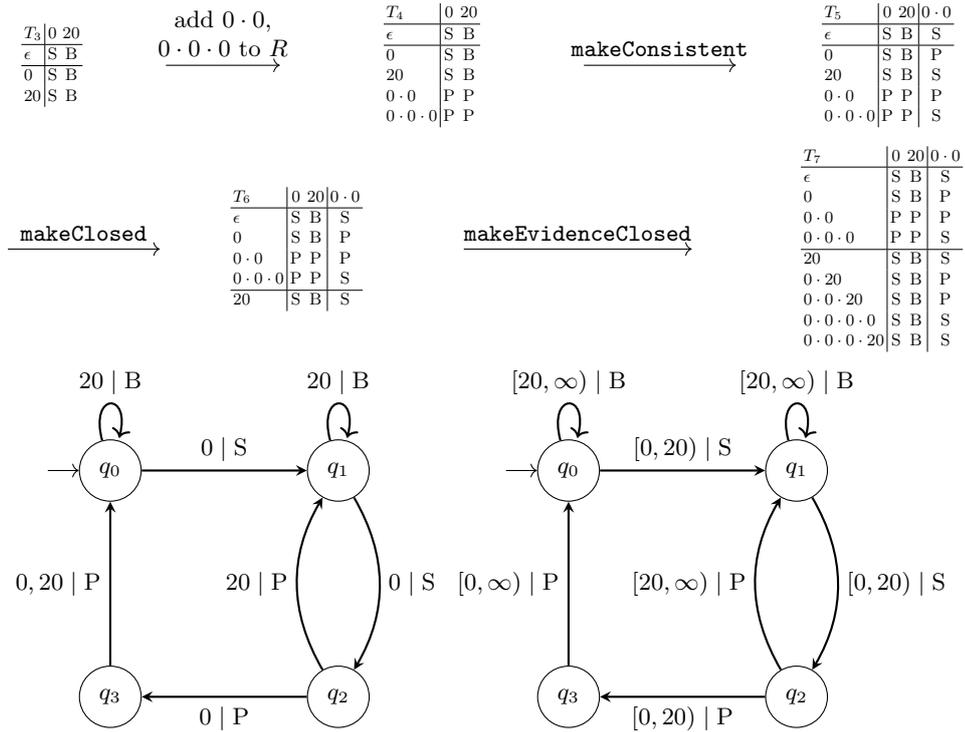

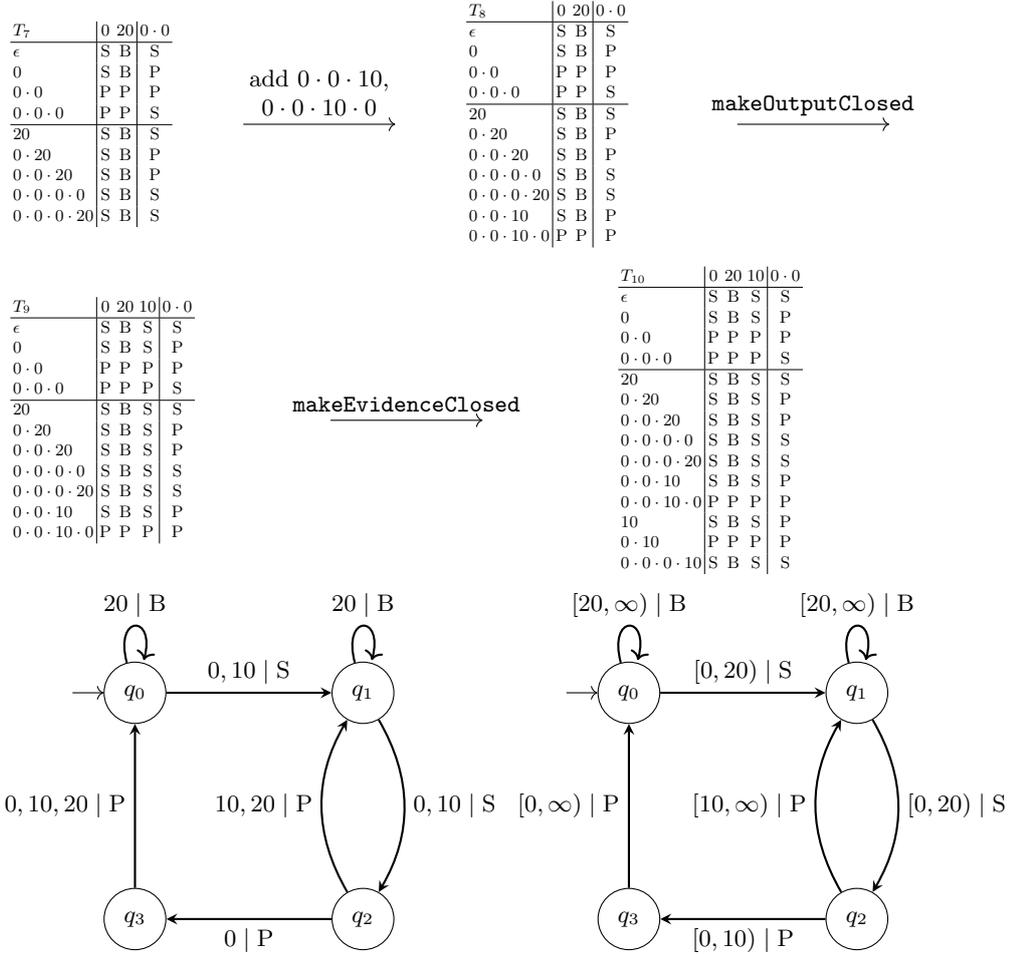
\begin{figure}
\begin{tikzpicture}
  \node at (0, 0) {
  \scalebox{.7}{
  \begin{tabular}{l|cc|c}
      $T_7$& 0 & 20 & $0\cdot 0$\\\hline
      $\epsilon$ & S & B & S\\
      0 & S & B & P\\
      $0\cdot 0$ & P & P & P \\
      $0\cdot 0 \cdot 0$ & P & P & S \\\hline
      20 & S & B & S \\
      $0\cdot 20$ & S & B & P \\
      $0\cdot 0 \cdot 20$ & S & B & P \\
      $0\cdot 0 \cdot 0\cdot 0$ & S & B & S \\
      $0\cdot 0 \cdot 0 \cdot 20$ & S & B & S \\
    \end{tabular}
    }
  };

  \draw[->] (2,0) -- (4, 0) node[midway, above, align = center] {add $0\cdot 0 \cdot 10$, \\ $0\cdot 0 \cdot 10 \cdot 0$};

  \node at (6, 0) {
  \scalebox{.7}{
  \begin{tabular}{l|cc|c}
      $T_8$& 0 & 20 & $0\cdot 0$\\\hline
      $\epsilon$ & S & B & S\\
      0 & S & B & P\\
      $0\cdot 0$ & P & P & P \\
      $0\cdot 0 \cdot 0$ & P & P & S \\\hline
      20 & S & B & S \\
      $0\cdot 20$ & S & B & P \\
      $0\cdot 0 \cdot 20$ & S & B & P \\
      $0\cdot 0 \cdot 0\cdot 0$ & S & B & S \\
      $0\cdot 0 \cdot 0 \cdot 20$ & S & B & S \\
      $0\cdot 0 \cdot 10$ & S & B & P \\
      $0\cdot 0 \cdot 10 \cdot 0$ & P & P & P \\
    \end{tabular}
    }
  };
  \draw[->] (8.5,0) -- (10.5, 0) node[midway, above, align = center] {$\FMakeOutputClosed$};

\end{tikzpicture}

\begin{tikzpicture}
\node at (0, 0) {
\scalebox{.7}{
  \begin{tabular}{l|ccc|c}
      $T_9$& 0 & 20 & 10& $0\cdot 0$\\\hline
      $\epsilon$ & S & B &S & S\\
      0 & S & B &S & P\\
      $0\cdot 0$ & P & P &P& P \\
      $0\cdot 0 \cdot 0$ & P & P&P & S \\\hline
      20 & S & B &S& S \\
      $0\cdot 20$ & S & B & S& P \\
      $0\cdot 0 \cdot 20$ & S & B &S& P \\
      $0\cdot 0 \cdot 0\cdot 0$ & S & B &S& S \\
      $0\cdot 0 \cdot 0 \cdot 20$ & S & B &S& S \\
      $0\cdot 0 \cdot 10$ & S & B &S& P \\
      $0\cdot 0 \cdot 10 \cdot 0$ & P & P &P& P \\
    \end{tabular}
    }
  };
      \draw[->] (3,0) -- (5, 0) node[midway, above, align = center] {$\FMakeEvidenceClosed$};

  \node at (8, 0) {
  \scalebox{.7}{
  \begin{tabular}{l|ccc|c}
      $T_{10}$& 0 & 20 & 10& $0\cdot 0$\\\hline
      $\epsilon$ & S & B &S & S\\
      0 & S & B &S & P\\
      $0\cdot 0$ & P & P &P& P \\
      $0\cdot 0 \cdot 0$ & P & P&P & S \\\hline
      20 & S & B &S& S \\
      $0\cdot 20$ & S & B & S& P \\
      $0\cdot 0 \cdot 20$ & S & B &S& P \\
      $0\cdot 0 \cdot 0\cdot 0$ & S & B &S& S \\
      $0\cdot 0 \cdot 0 \cdot 20$ & S & B &S& S \\
      $0\cdot 0 \cdot 10$ & S & B &S& P \\
      $0\cdot 0 \cdot 10 \cdot 0$ & P & P &P& P \\
      $10$ & S & B & S& P \\
      $0\cdot10$ & P & P & P& P \\
      $0\cdot 0 \cdot 0\cdot 10$ & S & B & S& S \\
    \end{tabular}
    }
  };
\end{tikzpicture}

\begin{tikzpicture} [node distance = 3cm, on grid, auto]
 
\node (q0) [state, initial, initial text = {}] {$q_0$};
\node (q1) [state, right = of q0] {$q_1$};
\node (q2) [state, below = of q1] {$q_2$};
\node (q3) [state, below = of q0] {$q_3$};
 
\path [-stealth, thick]
    (q0) edge  []node {$0,10\mid \text{S}$} (q1) 
    (q0) edge [loop above]  node {$20\mid \text{B}$}()
    (q1) edge [loop above] node {$20\mid \text{B}$}()
    (q1) edge  [bend left ]node {$0,10\mid \text{S}$} (q2)
    (q2) edge  [bend left ]node {$10,20\mid \text{P}$} (q1)
    (q2) edge []  node {$0\mid \text{P}$}(q3)
    (q3) edge  []node {$0,10,20\mid \text{P}$}   (q0);
\end{tikzpicture}
\begin{tikzpicture} [node distance = 3cm, on grid, auto]
 
\node (q0) [state, initial, initial text = {}] {$q_0$};
\node (q1) [state, right = of q0] {$q_1$};
\node (q2) [state, below = of q1] {$q_2$};
\node (q3) [state, below = of q0] {$q_3$};
 
\path [-stealth, thick]
    (q0) edge  []node {$[0,20)\mid \text{S}$} (q1) 
    (q0) edge [loop above]  node {$[20,\infty)\mid \text{B}$}()
    (q1) edge [loop above] node {$[20,\infty)\mid \text{B}$}()
    (q1) edge  [bend left ]node {$[0,20)\mid \text{S}$} (q2)
    (q2) edge  [bend left ]node {$[10,\infty)\mid \text{P}$} (q1)
    (q2) edge []  node {$[0,10)\mid \text{P}$}(q3)
    (q3) edge  []node {$[0,\infty)\mid \text{P}$}   (q0);
\end{tikzpicture}
\caption{The observation tables in the final loop and the final hypothesis evidence and s-MA.}
\label{fig:final-round}
\end{figure}

\section{Detail of the Benchmarks}
We used two Boolean algebras and their product algebra. We used an interval algebra over naturals and an interval algebra over 64bit floating-point numbers.
The partitioning function for the interval algebra over double-precision floating-point numbers can handle only left-closed and right-open intervals.
So, we encode, for example, $0 < x \leq 100$ with $\na(0) \leq x < \na(100)$, where $\na(y)$ is the nearest floating-point number greater than $y$. 
For the Boolean values, we use the interval algebra over naturals restricted to $\{0, 1\}$.

\subsection{\MH{} (Mars Helicopter)}
\MH{} is a benchmark on a control system for a helicopter on Mars. We took the model from~\cite{MH}.  
Each input character of \MH{} consists of the following four values: a Boolean value prepareFlight and the floating-point values for the altitude, the temperature, and the state of charge. 
The output alphabet is $\OUTPUT = \{\emptyset,\{\text{heater}\},\{\text{fly}\}, \{\text{fly,altitudeRef}\}\}$. 
\Cref{MH} shows the target s-MA. %

\begin{figure}
\centering
\begin{tikzpicture}[shorten >=1pt, node distance=3cm, on grid, auto,font=\scriptsize]

\tikzstyle{every state}=[fill=none,draw=black,text=black]

\node (q0) [state, initial, initial text = {}] {$q_0$};
\node (q1) [state, above = of q0] {$q_1$};
\node (q2) [state, below = of q0] {$q_2$};
\node (q3) [state, below = of q2] {$q_3$};
\node (q4) [state, below = of q3] {$q_4$};

\path [-stealth, thick]
(q0) edge [loop right] node [align = center]{$([1,1]\times[0,1e5)$\\$\times[-15,1e4)\times[0,\na(0.4)))$\\$\cup([0,0]\times[0,1e5)$\\$\times[-15,1e4)\times[0,\na(1)))$ $\mid$ $\emptyset$} ()
(q0) edge [bend left] node [align = center] {$[1,1]\times[0,1e5)$\\$\times[-15,1e4) \times[\na(0.4),\na(1))$ \\$\mid$ \{heater\}} (q2)
(q0) edge [bend left] node [align = center] {$[0,1]\times[0,1e5)$\\$\times[-274,-15)\times[0,\na(1))$ $\mid$  \{heater\}} (q1)
(q1) edge [loop above] node [align = center] {$[0,1]\times[0,1e5)\times[-274,\na(-10))\times[0.2,\na(1))$ $\mid$  \{heater\}} ()
(q1) edge [bend left] node [align = center] {$([0,1]\times[0,1e5)$\\$\times[\na(-10),1e4)\times[0,\na(1)))$\\$\cup([0,1]\times[0,1e5)$\\$\times[-274,\na(-10))\times[0,0.2))$ $\mid$ $\emptyset$} (q0)
(q2) edge [loop right] node [align = center]{$[1,1]\times[0,1e5)$\\$\times[-274,\na(10))\times[0.3,\na(1))$ \\$\mid$  \{heater\}} ()
(q2) edge node [align = center]{$[1,1]\times[0,1e5)$\\$\times[\na(10),1e4)\times[0.3,\na(1))$ \\$\mid$  \{fly,altitudeRef\}} (q3)
(q2) edge [bend left] node [align = center]{$([0,0]\times[0,1e5)$\\$\times[-274,1e4)\times[0,\na(1)))$\\$\cup([1,1]\times[0,1e5)$\\$\times[-274,1e4)\times[0,0.3))$ $\mid$  $\emptyset$} (q0)
(q3) edge [loop right] node [align = center]{$[1,1]\times[0,1e5)$\\$\times[-274,1e4)\times[0.3,\na(1))$ \\$\mid$  \{fly,altitudeRef\}} ()
(q3) edge node [align = center]{$([0,0]\times[0,1e5)$\\$\times[-274,1e4)\times[0,\na(1)))$\\$\cup([1,1]\times[0,1e5)$\\$\times[-274,1e4)\times[0,0.3))$ $\mid$  \{fly\}} (q4)
(q4) edge [loop below] node {$[0,1]\times[0.1,1e5)\times[-274,1e4)\times[0,\na(1))$ $\mid$  \{fly\}} ()
(q4) edge [bend left=20] node [align = center]{$[0,1]\times[0,0.1)$\\$\times[-274,1e4)\times[0,\na(1))$ $\mid$ $\emptyset$} (q0);

\end{tikzpicture}
\caption{\MH}\label{MH}
\end{figure}

\subsection{\ATGS{} (Automatic Transmission Gear System)}

\ATGS{} is a benchmark on an automatic transmission gear system for a car. We took the model from \cite{ATGS}. \ATGS{} takes two inputs: the throttle and the velocity. 
The output alphabet is $\OUTPUT = \{\text{gear1},\text{gear2},\text{gear3},\text{gear4}\}$. \Cref{ATGS} shows the s-MA representing the behavior of it, where the predicates in \cref{ATGS} are as follows.

\begin{figure}
    \centering
    \begin{tikzpicture}[shorten >=1pt, node distance=3cm, on grid, auto,scale=1.0,font=\tiny,every node/.style={transform shape,initial text=}]

    \tikzstyle{every state}=[fill=none,draw=black,text=black]

    \node (q100) [state, initial, initial text = {}] {$q_{100}$};
    \node (q101) [state, below left = of q100] {$q_{101}$};
    \node (q102) [state, below right = of q100] {$q_{102}$};
    \node (q200) [state, below = of q102] {$q_{200}$};
    \node (q201) [state, below right = of q200] {$q_{201}$};
    \node (q202) [state, below = 2cm of q201] {$q_{202}$};
    \node (q210) [state, left = 4cm of q200] {$q_{210}$};
    \node (q220) [state, above = 2cm of q210] {$q_{220}$};
    \node (q300) [state, left = of q202] {$q_{300}$};
    \node (q301) [state, below right = of q300] {$q_{301}$};
    \node (q302) [state, below =2cm of q301] {$q_{302}$};
    \node (q310) [state, left = 4cm of q300] {$q_{310}$};
    \node (q320) [state, above = 2cm of q310] {$q_{320}$};
    \node (q400) [state, left = of q302] {$q_{400}$};
    \node (q410) [state, left = 4cm of q400] {$q_{410}$};
    \node (q420) [state, above = 2cm of q410] {$q_{420}$};
    \path [-stealth, thick]
    (q100) edge [loop above] node [align = center]{$\varphi_{100\_100}$ $\mid$ gear1} ()
    (q100) edge [bend left] node [align = center]{$\varphi_{100\_101}$ $\mid$ gear1} (q101)
    (q101) edge [bend right] node [above, align = center]{$\varphi_{101\_102}$ $\mid$ gear1} (q102)
    (q101) edge [bend left] node [ align = center]{$\varphi_{101\_100}$ $\mid$ gear1} (q100)
    (q102) edge [] node [right, align = center]{$\varphi_{102\_200}$ $\mid$ gear2} (q200)
    (q102) edge [bend right] node [right, align = center]{$\varphi_{102\_100}$ $\mid$ gear1} (q100)
     (q200) edge [bend left] node [align = center]{$\varphi_{200\_201}$ $\mid$ gear2} (q201)
     (q200) edge [loop right] node [align = center]{$\varphi_{200\_200}$ $\mid$ gear2} ()
     (q200) edge [bend left] node [below, align = center]{$\varphi_{200\_210}$ $\mid$ gear2} (q210)
     (q201) edge [] node [align = center]{$\varphi_{201\_202}$ $\mid$ gear2} (q202)
     (q201) edge [bend left] node [align = center]{$\varphi_{201\_200}$ $\mid$ gear2} (q200)
     (q210) edge [] node [left, align = center]{$\varphi_{210\_220}$ $\mid$ gear2} (q220)
     (q210) edge [bend left] node [below, align = center]{$\varphi_{210\_200}$ $\mid$ gear2} (q200)
      (q220) edge [bend left] node [left, align = center]{$\varphi_{220\_200}$ $\mid$ gear2} (q200)
      (q220) edge [bend left = 90] node [left, align = center]{$\varphi_{220\_100}$ $\mid$ gear1} (q100)
      (q202) edge [bend left] node [left,align = center]{$\varphi_{202\_200}$ $\mid$ gear2} (q200)
     (q202) edge [] node [above, align = center]{$\varphi_{202\_300}$ $\mid$ gear3} (q300)
     (q300) edge [bend left] node [right, align = center]{$\varphi_{300\_301}$ $\mid$ gear3} (q301)
     (q300) edge [loop above] node [align = center]{$\varphi_{300\_300}$ $\mid$ gear3} ()
     (q301) edge [] node [align = center]{$\varphi_{301\_302}$ $\mid$ gear3} (q302)
     (q301) edge [bend left] node [align = center]{$\varphi_{301\_300}$ $\mid$ gear3} (q300)
     (q300) edge [bend left] node [align = center]{$\varphi_{300\_310}$ $\mid$ gear3} (q310)
     (q310) edge [] node [left, align = center]{$\varphi_{310\_320}$ $\mid$ gear3} (q320)
     (q320) edge [bend right = 20] node [left, align = center]{$\varphi_{320\_200}$ $\mid$ gear2} (q200)
     (q320) edge [bend left = 20] node [right, align = center]{$\varphi_{320\_300}$ $\mid$ gear3} (q300)
     (q310) edge [bend left] node [align = center]{$\varphi_{310\_300}$ $\mid$ gear3} (q300)
     (q302) edge [bend left] node [align = center]{$\varphi_{302\_300}$ $\mid$ gear3} (q300)
     (q302) edge [] node [align = center]{$\varphi_{302\_400}$ $\mid$ gear4} (q400)
     (q400) edge [bend left] node [align = center]{$\varphi_{400\_410}$ $\mid$ gear4} (q410)
     (q400) edge [loop below] node [align = center]{$\varphi_{400\_400}$ $\mid$ gear4} (q400)
     (q410) edge [bend left] node [align = center]{$\varphi_{410\_400}$ $\mid$ gear4} (q400)
     (q410) edge [] node [align = center]{$\varphi_{410\_420}$ $\mid$ gear4} (q420)
    (q420) edge [bend right = 20] node [left, align = center]{$\varphi_{420\_300}$ $\mid$ gear3} (q300)
     (q420) edge [bend left = 20] node [align = center]{$\varphi_{420\_400}$ $\mid$ gear4} (q400)

    ;
     \end{tikzpicture}
    \caption{\ATGS}\label{ATGS}
\end{figure}

$\varphi_{100\_100}=([50,90)\times[0,\na(23))\cup([0,35)\times[0,\na(10))\cup([35,50)\times[0,\na(15))\cup([90,100)\times[0,\na(40))$

$\varphi_{100\_101}=([50,90)\times[\na(23),1e6))\cup([0,35)\times[\na(10),1e6))\cup([35,50)\times[\na(15),1e6))\cup([90,100)\times[\na(40),1e6))$

$\varphi_{101\_102}=
([50,90)\times[23,1e6))\cup([0,35)\times[10,1e6))\cup([35,50)\times[15,1e6))\cup([90,100)\times[40,1e6))$

$\varphi_{101\_100}=
([50,90)\times[0,23))\cup([0,35)\times[0,10))\cup([35,50)\times[0,15))\cup([90,100)\times[0,40))$

$\varphi_{102\_200}=
([50,90)\times[23,1e6))\cup([0,35)\times[10,1e6))\cup([35,50)\times[15,1e6))\cup([90,100)\times[40,1e6))$

$\varphi_{102\_100}=
([50,90)\times[0,23))\cup([0,35)\times[0,10))\cup([35,50)\times[0,15))\cup([90,100)\times[0,40))$

$\varphi_{200\_201}=
([50,90)\times[\na(41),1e6))\cup([0,50)\times[\na(30),1e6))\cup([90,100)\times[\na(70),1e6))$

$\varphi_{200\_200}=
([0,50)\times[5,\na(30)))\cup([90,100)\times[30,\na(70)))\cup([50,90)\times[5,\na(41)))$

$\varphi_{200\_210}=
([0,90)\times[0,5))\cup([90,100)\times[0,30))$

$\varphi_{201\_202}=
([50,90)\times[41,1e6))\cup([0,50)\times[30,1e6))\cup([90,100)\times[70,1e6))$

$\varphi_{201\_200}=
([50,90)\times[0,41))\cup([0,50)\times[0,30))\cup([90,100)\times[0,70))$

$\varphi_{210\_220}=
([0,90)\times[0,\na(5)))\cup([90,100)\times[0,\na(30)))$

$\varphi_{210\_200}=
([0,90)\times[\na(5),1e6))\cup([90,100)\times[\na(30),1e6))$

$\varphi_{220\_200}=
([0,90)\times[\na(5),1e6))\cup([90,100)\times[\na(30),1e6))$

$\varphi_{220\_100}=
([0,90)\times[0,\na(5)))\cup([90,100)\times[0,\na(30)))$

$\varphi_{202\_200}=
([50,90)\times[0,41))\cup([0,50)\times[0,30))\cup([90,100)\times[0,70))$

$\varphi_{202\_300}=
([50,90)\times[41,1e6))\cup([0,50)\times[30,1e6))\cup([90,100)\times[70,1e6))$

$\varphi_{300\_301}=
([50,90)\times[\na(60),1e6))\cup([0,50)\times[\na(50),1e6))\cup([90,100)\times[\na(100),1e6))$

$\varphi_{300\_300}=
([0,40)\times[20,\na(50)))\cup([90,100)\times[50,\na(100)))\cup([40,50)\times[25,\na(50)))\cup([50,90)\times[30,\na(60)))$

$\varphi_{300\_310}=
([0,40)\times[0,20))\cup([40,50)\times[0,25))\cup([50,90)\times[0,30))\cup([90,100)\times[0,50))$

$\varphi_{301\_302}=
([50,90)\times[60,1e6))\cup([0,50)\times[50,1e6))\cup([90,100)\times[100,1e6))$

$\varphi_{301\_300}=
([50,90)\times[0,60))\cup([0,50)\times[0,50))\cup([90,100)\times[0,100))$

$\varphi_{310\_320}=
([0,40)\times[0,\na(20)))\cup([40,50)\times[0,\na(25)))\cup([50,90)\times[0,\na(30)))\cup([90,100)\times[0,\na(50)))$

$\varphi_{310\_300}=
([0,40)\times[\na(20),1e6))\cup([40,50)\times[\na(25),1e6))\cup([50,90)\times[\na(30),1e6))\cup([90,100)\times[\na(50),1e6))$

$\varphi_{320\_200}=
([0,40)\times[0,\na(20)))\cup([40,50)\times[0,\na(25)))\cup([50,90)\times[0,\na(30)))\cup([90,100)\times[0,\na(50)))$

$\varphi_{320\_300} = ([0,40)\times[\na(20),1e6))\cup([40,50)\times[\na(25),1e6))\cup([50,90)\times[\na(30),1e6))\cup([90,100)\times[\na(50),1e6))$

$\varphi_{302\_300} =
([50,90)\times[0,60))\cup([0,50)\times[0,50))\cup([90,100)\times[0,100))$

$\varphi_{302\_400} =
([50,90)\times[60,1e6))\cup([0,50)\times[50,1e6))\cup([90,100)\times[100,1e6))$

$\varphi_{400\_410} =
([0,40)\times[0,35))\cup([40,50)\times[0,40))\cup([50,90)\times[0,50))\cup([90,100)\times[0,80))$

$\varphi_{400\_400} =
([0,40)\times[35,1e6))\cup([40,50)\times[40,1e6))\cup([50,90)\times[50,1e6))\cup([90,100)\times[80,1e6))$

$\varphi_{410\_400} =
([0,40)\times[\na(35),1e6))\cup([40,50)\times[\na(40),1e6))\cup([50,90)\times[\na(50),1e6))\cup([90,100)\times[\na(80),1e6))$

$\varphi_{410\_420} =
([0,40)\times[0,\na(35)))\cup([40,50)\times[0,\na(40)))\cup([50,90)\times[0,\na(50)))\cup([90,100)\times[0,\na(80)))$

$\varphi_{420\_300}=
([0,40)\times[0,\na(35)))\cup([40,50)\times[0,\na(40)))\cup([50,90)\times[0,\na(50)))\cup([90,100)\times[0,\na(80)))$

$\varphi_{420\_400}=
([0,40)\times[\na(35),1e6))\cup([40,50)\times[\na(40),1e6))\cup([50,90)\times[\na(50),1e6))\cup([90,100)\times[\na(80),1e6))$

\subsection{Random Benchmarks}\label{section:random_benchmark_detail}
The random generation method for the s-MAs used in \cref{section:experiments:random} is as follows. 
First, we randomly generate a partition by sampling $|\SigmaEf| - 1$ naturals and make the corresponding partition with the partitioning function $\partitioning$ in \cref{algorithm:partitioning_interval_algebra}. 
Then, we randomly choose the destination state and the output character for each state and predicate.

\section{Detail of Experimental Results}

\subsection{$\MH$}
We show a series of counterexamples in one of the executions of learning $\MH$ as follows.

$(0, 0, -15, 0)$

$(1, 0, -15, \na(0.4))$

$(1, 0, -15, 0)$

$(1, 0, -274, 0)$

$(0, 0, -274, 0), (0, 0, -274, 0.2)$

$(0, 0, -274, 0), (0, 0, -15, 0.2)$

$(0, 0, -274, 0), (1, 0, -274, 0.2)$

$(0, 0, -274, 0), (1, 0, -15, 0.2)$

$(0, 0, -274, 0), (0, 0, \na(-10), 0.2)$

$(0, 0, -274, 0), (1, 0, \na(-10), 0.2)$

$(1, 0, \na(-10), \na(0.4))$

$(1, 0, -15, \na(0.4)), (1, 0, -274, 0.3)$

$(1, 0, -15, \na(0.4)), (1, 0, -15, 0.3)$

$(1, 0, -15, \na(0.4)), (1, 0, \na(-10), 0.3)$

$(1, 0, -15, \na(0.4)), (1, 0, \na(10), 0.3)$

$(1, 0, \na(10), \na(0.4))$

$(1, 0, -15, \na(0.4)), (1, 0, \na(10), 0)$

$(1, 0, -15, \na(0.4)), (1, 0, \na(10), 0.3), (0, 0, -274, 0), (0, 0.1, -274, 0)$

$(0, 0.1, -15, 0)$

$(0, 0, -274, 0), (0, 0.1, -274, 0.2)$

$(0, 0, -274, 0), (0, 0.1, -15, 0.2)$

$(0, 0, -274, 0), (0, 0.1, \na(-10), 0.2)$

$(1, 0, -15, \na(0.4)), (1, 0, \na(10), 0.3), (0, 0, -274, 0), (1, 0.1, -274, 0)$

$(1, 0.1, -15, 0)$

$(1, 0.1, -15, \na(0.4))$

$(0, 0, -274, 0), (1, 0.1, -274, 0.2)$

$(0, 0, -274, 0), (1, 0.1, -15, 0.2)$

$(0, 0, -274, 0), (1, 0.1, \na(-10), 0.2)$

$(1, 0.1, \na(-10), \na(0.4))$

$(1, 0, -15, \na(0.4)), (1, 0.1, -274, 0.3)$

$(1, 0, -15, \na(0.4)), (1, 0.1, -15, 0.3)$

$(1, 0, -15, \na(0.4)), (1, 0.1, \na(-10), 0.3)$

$(1, 0, -15, \na(0.4)), (1, 0.1, \na(10), 0.3)$

$(1, 0.1, \na(10), \na(0.4))$

$(1, 0, -15, \na(0.4)), (1, 0.1, \na(10), 0)$

\subsection{$\ATGS$}
We show a series of counterexamples in one of the executions of learning $\ATGS$ as follows.

\footnotesize
$(0, \na(10)), (0, 10), (0, 10)$

$(35, \na(10)), (0, 10), (0, 10)$

$(0, \na(10)), (0, 10), (35, 15)$

$(35, \na(15)), (0, 10), (0, 10)$

$(50, \na(15)), (0, 10), (0, 10)$

$(0, \na(10)), (50, 23), (0, \na(10))$

$(50, \na(23)), (0, 10), (0, 10)$

$(90, \na(23)), (0, 10), (0, 10)$

$(90, \na(40)), (0, 10), (0, 10)$

$(0, \na(10)), (0, 10), (90, 40)$

$(0, \na(10)), (0, 10), (0, 10), (0, 0), (0, 0), (0, 0)$

$(0, \na(10)), (0, 10), (0, 10), (35, 0), (35, 0), (35, 0)$

$(0, \na(10)), (0, 10), (0, 10), (50, 0), (50, 0), (50, 0)$

$(0, \na(10)), (0, 10), (0, 10), (0, 0), (0, \na(5)), (0, 0)$

$(0, \na(10)), (0, 10), (0, 10), (0, 0), (35, \na(5)), (0, 0)$

$(0, \na(10)), (0, 10), (0, 10), (0, 0), (50, \na(5)), (0, 0)$

$(0, \na(10)), (0, 10), (0, 10), (0, 0), (90, \na(30)), (0, 0)$

$(0, \na(10)), (0, 10), (0, 10), (0, 5), (0, 0), (0, 0)$

$(0, \na(10)), (0, 10), (0, 10), (35, 5), (0, 0), (0, 0)$

$(0, \na(10)), (0, 10), (0, 10), (50, 5), (0, 0), (0, 0)$

$(0, \na(10)), (0, 10), (0, 10), (90, 30), (0, 0), (0, 0)$

$(0, \na(10)), (0, 10), (0, 10), (0, \na(30)), (0, 30), (0, 30)$

$(0, \na(10)), (0, 10), (0, 10), (0, \na(30)), (35, 30), (0, \na(30))$

$(0, \na(10)), (0, 10), (0, 10), (35, \na(30)), (0, 30), (0, 30)$

$(0, \na(10)), (0, 10), (0, 10), (0, \na(30)), (0, 30), (50, 41)$

$(0, \na(10)), (0, 10), (0, 10), (50, \na(41)), (0, 30), (0, 30)$

$(0, \na(10)), (0, 10), (0, 10), (0, \na(30)), (90, 70), (0, \na(30))$

$(0, \na(10)), (0, 10), (0, 10), (90, \na(70)), (0, 30), (0, 30)$

$(0, \na(10)), (0, 10), (0, 10), (0, \na(30)), (0, 30), (0, 30), (0, 20), (0, 20), (0, 20)$

$(0, \na(10)), (0, 10), (0, 10), (0, \na(30)), (0, 30), (0, 30), (35, 20), (0, 20), (0, 20)$

$(0, \na(10)), (0, 10), (0, 10), (0, \na(30)), (0, 30), (0, 30), (50, 30), (0, 20), (0, 20)$

$(0, \na(10)), (0, 10), (0, 10), (0, \na(30)), (0, 30), (0, 30), (90, 50), (0, 20), (0, 20)$

$(0, \na(10)), (0, 10), (0, 10), (0, \na(30)), (0, 30), (0, 30), (40, 20), (0, 20), (0, 20)$

$(40, 0), (0, 10), (0, 10)$

$(40, \na(15)), (0, 10), (0, 10)$

$(0, \na(10)), (0, 10), (40, 15)$

$(0, \na(10)), (0, 10), (0, 10), (0, 0), (40, \na(5)), (0, 0)$

$(0, \na(10)), (0, 10), (0, 10), (40, 5), (0, 0), (0, 0)$

$(0, \na(10)), (0, 10), (0, 10), (0, \na(30)), (40, 30), (0, \na(30))$

$(0, \na(10)), (0, 10), (0, 10), (40, \na(30)), (0, \na(30)), (0, 30)$

$(0, \na(10)), (0, 10), (0, 10), (0, \na(30)), (0, 30), (0, 30), (40, 25), (0, 20), (0, 20)$

$(0, \na(10)), (0, 10), (0, 10), (0, \na(30)), (0, 30), (0, 30), (0, 0), (0, \na(20)), (0, 20)$

$(0, \na(10)), (0, 10), (0, 10), (0, \na(30)), (0, 30), (0, 30), (0, 0), (35, \na(20)), (0, 20)$

$(0, \na(10)), (0, 10), (0, 10), (0, \na(30)), (0, 30), (0, 30), (0, 0), (40, \na(25)), (0, 20)$

$(0, \na(10)), (0, 10), (0, 10), (0, \na(30)), (0, 30), (0, 30), (0, 0), (50, \na(30)), (0, 20)$

$(0, \na(10)), (0, 10), (0, 10), (0, \na(30)), (0, 30), (0, 30), (0, 0), (90, \na(50)), (0, 20)$

$(0, \na(10)), (0, 10), (0, 10), (0, \na(30)), (0, 30), (0, 30), (0, \na(50)), (0, 50), (0, 50)$

$(0, \na(10)), (0, 10), (0, 10), (0, \na(30)), (0, 30), (0, 30), (35, \na(50)), (0, 50), (0, 50)$

$(0, \na(10)), (0, 10), (0, 10), (0, \na(30)), (0, 30), (0, 30), (0, \na(50)), (0, 50), (35, 50)$

$(0, \na(10)), (0, 10), (0, 10), (0, \na(30)), (0, 30), (0, 30), (0, \na(50)), (40, 50), (0, \na(50))$

$(0, \na(10)), (0, 10), (0, 10), (0, \na(30)), (0, 30), (0, 30), (40, \na(50)), (0, 50), (0, 50)$

$(0, \na(10)), (0, 10), (0, 10), (0, \na(30)), (0, 30), (0, 30), (0, \na(50)), (50, 60), (0, \na(50))$

$(0, \na(10)), (0, 10), (0, 10), (0, \na(30)), (0, 30), (0, 30), (0, \na(50)), (90, 100), (0, \na(50))$

$(0, \na(10)), (0, 10), (0, 10), (0, \na(30)), (0, 30), (0, 30), (50, \na(60)), (0, \na(50)), (0, 50)$

$(0, \na(10)), (0, 10), (0, 10), (0, \na(30)), (0, 30), (0, 30), (90, \na(100)), (0, \na(50)), (0, 50)$

$(0, \na(10)), (0, 10), (0, 10), (0, \na(30)), (0, 30), (0, 30), (0, \na(50)), (0, 50), (0, 50), (0, 35), (0, 35), (0, 35)$

$(0, \na(10)), (0, 10), (0, 10), (0, \na(30)), (0, 30), (0, 30), (0, \na(50)), (0, 50), (0, 50), (35, 35), (0, 35), (0, 35)$

$(0, \na(10)), (0, 10), (0, 10), (0, \na(30)), (0, 30), (0, 30), (0, \na(50)), (0, 50), (0, 50), (40, 40), (0, 35), (0, 35)$

$(0, \na(10)), (0, 10), (0, 10), (0, \na(30)), (0, 30), (0, 30), (0, \na(50)), (0, 50), (0, 50), (50, 50), (0, 35), (0, 35)$

$(0, \na(10)), (0, 10), (0, 10), (0, \na(30)), (0, 30), (0, 30), (0, \na(50)), (0, 50), (0, 50), (90, 80), (0, 35), (0, 35)$

$(0, \na(10)), (0, 10), (0, 10), (0, \na(30)), (0, 30), (0, 30), (0, \na(50)), (0, 50), (0, 50), (0, 0), (0, \na(35)), (0, 35)$

$(0, \na(10)), (0, 10), (0, 10), (0, \na(30)), (0, 30), (0, 30), (0, \na(50)), (0, 50), (0, 50), (0, 0), (35, \na(35)), (0, 35)$

$(0, \na(10)), (0, 10), (0, 10), (0, \na(30)), (0, 30), (0, 30), (0, \na(50)), (0, 50), (0, 50), (0, 0), (40, \na(40)), (0, 35)$

$(0, \na(10)), (0, 10), (0, 10), (0, \na(30)), (0, 30), (0, 30), (0, \na(50)), (0, 50), (0, 50), (0, 0), (50, \na(50)), (0, 35)$

$(0, \na(10)), (0, 10), (0, 10), (0, \na(30)), (0, 30), (0, 30), (0, \na(50)), (0, 50), (0, 50), (0, 0), (90, \na(80)), (0, 35)$
\normalsize
\subsection{Random Benchmarks}
\Cref{table:detailResultRandom} summarizes the detailed results for the random benchmarks.
\begin{table}[tb]
\caption{Summary of the experimental setting and the results. The columns ``$n$'' and ``$|\SigmaEf|$'' show the number of states and the size of $\SigmaEf$ of the random s-MAs. The columns ``\# of eq.''\ and ``\# of oq.''\ show the average number of equivalence and output queries made during learning. The columns ``$|R|$'' and ``$|E|$'' show the size of $R$ and $E$ in the final observation table, respectively.
The columns applied by V() represent the respective variances.
}\label{table:detailResultRandom}
\scalebox{.9}{
\begin{tabular}{ll|rrrrrrrr}
\toprule
$n$ & $|\SigmaEf|$ & \# of eq. &  \# of oq. & $|R|$  &   $|E|$ &  V(\# of eq.) &  V(\# of oq.) &  V($|R|$) &  V($|E|$) \\
\midrule
    10 &     10 &  10.00 &   1015.60 &   91.56 & 0.0 &        0.00 &        78.64 &    0.79 &    0.00 \\
    10 &     20 &  20.00 &   4075.60 &  193.78 & 0.0 &        0.00 &      4036.64 &   10.09 &    0.00 \\
    10 &     30 &  29.99 &   9104.70 &  293.49 & 0.0 &        0.01 &      7406.91 &    8.23 &    0.00 \\
    10 &     40 &  40.00 &  16161.60 &  394.04 & 0.0 &        0.00 &      9661.44 &    6.04 &    0.00 \\
    20 &     10 &  10.10 &   2035.95 &  181.57 & 0.1 &        0.09 &      4146.25 &    0.79 &    0.09 \\
    20 &     20 &  20.00 &   8030.80 &  381.54 & 0.0 &        0.00 &       299.36 &    0.75 &    0.00 \\
    20 &     30 &  30.00 &  18080.70 &  582.69 & 0.0 &        0.00 &      4602.51 &    5.11 &    0.00 \\
    20 &     40 &  40.00 &  32110.80 &  782.77 & 0.0 &        0.00 &      8955.36 &    5.60 &    0.00 \\
    40 &     10 &  10.00 &   4010.90 &  361.09 & 0.0 &        0.00 &         8.19 &    0.08 &    0.00 \\
    40 &     20 &  20.00 &  16036.20 &  761.81 & 0.0 &        0.00 &       669.56 &    1.67 &    0.00 \\
    40 &     30 &  30.00 &  36079.20 & 1162.64 & 0.0 &        0.00 &      3753.36 &    4.17 &    0.00 \\
    40 &     40 &  40.00 &  64175.60 & 1564.39 & 0.0 &        0.00 &     20348.64 &   12.72 &    0.00 \\
    80 &     10 &  10.18 &   8186.59 &  722.55 & 0.2 &        0.19 &    113588.20 &    4.35 &    0.16 \\
    80 &     20 &  20.00 &  32031.80 & 1521.59 & 0.0 &        0.00 &       696.76 &    1.74 &    0.00 \\
    80 &     30 &  30.00 &  72072.00 & 2322.40 & 0.0 &        0.00 &      4014.00 &    4.46 &    0.00 \\
    80 &     40 &  40.00 & 128100.80 & 3122.52 & 0.0 &        0.00 &      7791.36 &    4.87 &    0.00 \\
\bottomrule
\end{tabular}
}
\end{table}
\fi
\end{document}